\documentclass[a4paper,english,titlepage,12pt]{report}

\usepackage[english]{babel}
\usepackage{amsmath,amssymb,amsthm}
\usepackage[numbers,compress]{natbib}
\usepackage{graphicx}

\newcommand{\nat}{\mathbb{N}_0}

\newcommand{\rrm}{\mathrm{rm}}
\newcommand{\sg}{\mathrm{sg}}
\newcommand{\Con}{\mathrm{Con}}
\newcommand{\modx}{\mathrm{mod}}

\newcommand{\dotminus}{\div}
\newcommand{\bigeps}{{\cal E}}
\newcommand{\ffrm}{\mathrm{rm}}
\newcommand{\ffsg}{\mathrm{sg}}
\newcommand{\ccxs}{\mathrm{XS}}

\newcommand{\ccclh}{\mathrm{S}}

\newcommand{\ccr}{\mathrm{BA}}
\newcommand{\cccr}{\mathrm{BA}^{\#}}
\newcommand{\cccrf}{\mathrm{BA}^{\#}_f}
\newcommand{\ccs}{\mathrm{S}}
\newcommand{\cck}{\mathrm{K}}
\newcommand{\ccfp}{\mathrm{FP}}
\newcommand{\ccfl}{\mathrm{FL}}

\newcommand{\lllog}{\mathrm{log}}

\newcommand{\Sum}{\displaystyle\sum\limits}

\newcommand{\Prod}{\displaystyle\prod\limits}

\newcommand{\ffrep}{\mathrm{rep}}
\newcommand{\ffincrx}{\mathrm{incrx}}
\newcommand{\ffincr}{\mathrm{incr}}
\newcommand{\ffdecr}{\mathrm{decr}}
\newcommand{\ffswap}{\mathrm{swap}}
\newcommand{\ffsa}{\mathrm{a}}
\newcommand{\ffsp}{\mathrm{p}}
\newcommand{\ffor}{\mathrm{or}}
\newcommand{\ffxor}{\mathrm{xor}}
\newcommand{\ffnot}{\mathrm{not}}
\newcommand{\ffcmp}{\mathrm{cmp}}
\newcommand{\ffcmpeq}{\mathrm{cmpeq}}
\newcommand{\ffsum}{\mathrm{sum}}
\newcommand{\ffssqrt}{\mathrm{ssqrt}}
\newcommand{\ffreverse}{\mathrm{reverse}}
\newcommand{\ffintcode}{\mathrm{code}}
\newcommand{\ffintlcode}{\mathrm{lcode}}

\newcommand{\qqm}{\mathrm{M}}
\newcommand{\llbit}{\mathrm{BIT}}
\newcommand{\llx}{\mathrm{X}}
\newcommand{\ccfom}{\mathrm{FOM}}
\newcommand{\ccfomn}{\mathrm{FOM}^N}
\newcommand{\ccffom}{\mathrm{FFOM}}
\newcommand{\ccffomvar}{\mathrm{FFOM^{var}}}
\newcommand{\ccffomalt}{\mathrm{FFOM^{alt}}}
\newcommand{\fflen}{\mathrm{len}}
\newcommand{\ffext}{\mathrm{ext}}
\newcommand{\ffcode}{\mathrm{CODE}}
\newcommand{\ffcodey}{\mathrm{CODE}^{\mathrm{alt}}}
\newcommand{\ffcodexd}{\mathrm{CODE}^{\mathrm{var}}}

\newcommand{\cckhs}{\mathrm{XS}}
\newcommand{\ccffomh}{\mathrm{FFOM}}
\newcommand{\cckhsplus}{\mathrm{XS}_+}
\newcommand{\ccexppoly}{\mathrm{P}}
\newcommand{\tth}{h}

\newcommand{\lln}{|X|}

\newcommand{\mains}{T}
\newcommand{\mainffoms}{T'}

\newcommand{\ffp}{c}
\newcommand{\ffpr}{c}
\newcommand{\ffrol}{\mathrm{rol}}
\newcommand{\ffall}{\mathrm{all}}
\newcommand{\ffmove}{\mathrm{move}}
\newcommand{\ffdelone}{\mathrm{del}}

\newcommand{\ffswapone}{\mathrm{swap_1}}
\newcommand{\ffswaptwo}{\mathrm{swap_2}}
\newcommand{\ffplace}{\mathrm{place}}
\newcommand{\ffs}{\mathrm{s}}
\newcommand{\perf}{\mathrm{p}}
\newcommand{\ffpx}{\mathrm{px}}

\newcommand{\gr}{\mathrm{Gr}}
\newcommand{\ccq}{\mathrm{Q}}

\newtheorem{theorem}{Theorem}
\newtheorem*{theorem1}{Theorem}
\newtheorem{stat}{Statement}[subsection]
\newtheorem{statse}{Statement}[section]
\newtheorem*{conseq}{Consequence}
\newtheorem*{comment}{Note}

\renewcommand{\thestat}{\arabic{chapter}.\arabic{section}.\arabic{subsection}.\arabic{stat}}

\renewcommand{\theequation}{\arabic{equation}}

\newcounter{sectionref}

\newcounter{subsectionref}

\newcounter{chapterref}

\newcommand{\Def}[1]{\vspace{2mm}\par{\bf Definition. }{#1\vspace{2mm}}}

\renewcommand{\leq}{\leqslant}
\renewcommand{\geq}{\geqslant}

\begin{document}
\begin{titlepage}
\begin{center}
\large LOMONOSOV MOSCOW STATE UNIVERSITY \\

FACULTY OF COMPUTATIONAL MATHEMATICS AND CYBERNETICS
\end{center}
\vspace{15 pt}
\begin{flushright}
Manuscript copyright
\end{flushright}
\vspace{15 pt}
\begin{center}
\large Volkov \\ Sergey
\end{center}
\vspace{15 pt}
\begin{center}
\large FINITE BASES WITH RESPECT TO THE SUPEROSITION IN CLASSES of ELEMENTARY
RECURSIVE FUNCTIONS
\end{center}
\vspace{15 pt}
\begin{center}
Specialty 01.01.09 --- Discrete Mathematics and Mathematical
Cybernetics
\end{center}
\vspace{15 pt}
\begin{center}
\large Dissertation \\ seeking the degree \\ of the Candidate of
Physical and Mathematical Sciences
\end{center}
\vspace{50 pt}
\begin{flushright}
\large Supervisor: \\ Doctor of Physical and Mathematical Sciences, \\
Professor S. S. Marchenkov
\end{flushright}
\vspace{100 pt}
\begin{center}
Moscow --- 2009
\end{center}
\end{titlepage}

%\section{test2}\refstepcounter{sectionref}\label{glava1}

\large

\tableofcontents

\chapter*{Introduction}
\addcontentsline{toc}{chapter}{Introduction}

\section{An Overview of the Results in the Thesis}

\subsection{Classification of Recursive Functions}

The concept of algorithmic computability (the recursive one) of a computable function
is prominent in modern Mathematics. The formalization of the concept of the computable function was accomplished in mid-1930 by the following famous mathematicians: A.~ Turing ~\cite{turing}, E.~ Post~\cite{post}, A.~ Church~\cite{church}, S.~Kleene~\cite{kleene}, and others. For every such formalization there exists the thesis (Church-Turing thesis) that claims that the whole class of algorithmically computable functions coincides with the class of functions that are computable in this type of formalization.

There is a number of approaches to formalizing the concept of a computable function. However, the simpler formalizations are most favorable. The most renowned approach is the one that uses various mechanical devices such as the Turing machine~\cite{turing,post}.

Another approach is based on generating functions from the set of
base functions and a number of generating operations. S. Kleene in
his work ~\cite{kleene} introduced the concept of a partial
recursive function. The partial recursive function is a partially
defined function $f:\nat^n\rightarrow\nat$, which can be derived
from initial functions $0$, $x+1$  with the help of a finite number
of operations using the superposition (superposition includes
substitution of functions into functions, permutations and
identification of variables, introduction of dummy variables),
primitive recursion, and minimization.

One should note that the class of all computable functions is not an
adequate formalization of the computability concept in practice. For
example, the function $f(n,x)$ that us defined in the following way
\begin{equation}\label{eqn_intro_f}
\left\{
\begin{array}{l}
f(0,x)=x+2, \\
f(n+1,0)=1, \\
f(n+1,x+1)=f(n,f(n+1,x)),
\end{array}
\right.
\end{equation}
is computable but it grows so fast that even $f(10,10)$ is
impossible to compute in practice. Besides, there is no method for
an arbitrary computable function and a set of incoming variables to
predict how much time it is going to take to compute the values of
the given function on a given set (and give a top down estimate of
resources needed for calculations.) Furthermore, not all listed
computable functions are defined everywhere (computable function is
not defined on such tuples where the computing algorithm does not
give an answer within a defined period of time).

In the light if these circumstances one considers narrower classes
that consist only of everywhere defined recursive functions as they
are closer to practical computability. The two above mentioned
methods work towards this direction: the machine one (consideration
of functions that are computable using different computational
devices with restrictions on resources that can be in use) and the
functional one (generation of classes based on some initial
functions and generating operations) as well as other approaches,
with which one can often encounter situations in which different
approaches give the same classes in the end.

The research of these kinds of classes and a search for equivalent
definitions for them can help to understand the nature of effective
computability and on top of that, possibly, this kind of definitions
will give an opportunity for specific functions, the computation of
which has an practical application, to figure out how they can be
computed effectively.

An example of this kind of class is the class of primitive recursive
functions (see ~\cite{maltsev}.) They say that the function
$f(\tilde{x},y)$  is obtained from functions $g(\tilde{x})$ and
$h(\tilde{x},y,z)$ with the help of the operation called
\emph{primitive recursion} if the following conditions hold:
\begin{equation}\label{eq_primrec}
\left\{
\begin{array}{l}
f(\tilde{x},0)=g(\tilde{x}), \\
f(\tilde{x},y+1)=h(\tilde{x},y,f(\tilde{x},y)).
\end{array}
\right.
\end{equation}
The function is called \emph{primitive recursive} if it can be
obtained from the functions $0$, $x+1$ with the help of the
superposition operation and primitive recursion. The class of
primitively recursive functions contains almost all everywhere
defined functions (at $\nat$) that are used in mathematical
practice. This class can be described in terms of complexity of
machine calculations: everywhere defined function
$g(x_1,\ldots,x_k)$ is primitively recursive if and only if it is
computable on Turing machine with a time constrain of the kind
$f(n,x_1+\ldots+x_k)$ for some $n$ where $f$ is defined by the
following relations (\ref{eqn_intro_f}) (see
~\cite{ritchie,muchnik}.)

Another example is the class of functions $\cck$ called Kalmar
elementary ~\cite{kalmar}. The class $\cck$ is a set of all
functions that can be obtained from functions
\begin{equation}\label{eq_basic_functions}
x+1,\quad x\dotminus y
\end{equation}
with the help of superposition, summation of the following kind:
$\sum_{x\leq y}$ and restricted multiplication $\prod_{x\leq y}$
(here and in the following it is $x\dotminus y=\max(0,x-y)$.) All
functions from the class $\cck$ are primitive recursive; however,
the opposite is not true (see ~\cite{kalmar}.) In terms of the
machine calculations, the class $\cck$ is the class of all functions
 $g(x_1,\ldots,x_n)$, Turing machine computable with a time input constrain of the type
$h_k(x_1+\ldots+x_n)$ with some $k$, where
$$
h_0(x)=x,\quad h_{k+1}(x)=2^{h_k(x)},
$$
an analogous statement is true for restrictions on the space (see
~\cite{ritchie}.) There are other equivalent definitions for the
class $\cck$. For example $f(x_1,\ldots,x_n)\in\cck$ if and only if
there exists such $k\in\nat$ and polynomials
$P(\tilde{x},y,\tilde{z})$, $Q(\tilde{x},y,\tilde{z})$ with
coefficients from $\nat$ that $f(\tilde{x})\leq m(\tilde{x})$ and
$$
(y=f(\tilde{x}))\equiv(\exists z_1)_{z_1\leq
m(\tilde{x})}\ldots(\exists z_l)_{z_l\leq
m(\tilde{x})}(P(\tilde{x},y,\tilde{z})=Q(\tilde{x},y,\tilde{z})),
$$
where $m(x_1,\ldots,x_n)=h_k(x_1+\ldots+x_n)$ (see ~\cite{vinkoss}.)

In his paper ~\cite{gzh}, A.~Grzegorczyk defined hierarchy of
classes $\bigeps^n$, $n\in\nat$. They say that $f(\tilde{x},y)$ is
obtained from functions $g(\tilde{x})$, $h(\tilde{x},y,z)$, and
$j(\tilde{x},y)$ with the help of \emph{restricted recursion} if it
satisfies the relations (\ref{eq_primrec}) and
$$
f(\tilde{x},y)\leq j(\tilde{x},y).
$$
$\bigeps^n$ is a minimal class of functions that contains functions
$x+1$, $f_n(x,y)$ and is closed with respect to superposition, and
restricted recursion where
$$
f_0(x,y)=y+1,
$$
$$
f_1(x,y)=x+y,
$$
$$
f_2(x,y)=(x+1)\cdot (y+1),
$$
with $n\geq 2$
$$
f_{n+1}(0,y)=f_n(y+1,y+1),
$$
$$
f_{n+1}(x+1,y)=f_{n+1}(x,f_{n+1}(x,y)).
$$
Grzegorczyk's hierarchy is strictly monotonous and exhausts the
whole class of primitive recursive functions. Besides,
$\bigeps^3=\cck$ (this is proved in ~\cite{gzh}.) For classes
$\bigeps^n$, $n\geq 2$, there exists a description as defined in
terms of complexity of calculations on the Turing machines. For
example $\bigeps^2$ is a set of all functions that can be computed
on the Turing machine with a linear space (of its input length; the
numbers are represented in the binary form), see ~\cite{ritchie}.

One more class that deserves attention is the class introduced by T.
Skolem in ~\cite{skolemone,skolemtwo}, the class of elementary
functions ("lower elementary functions"). This class will be defined
as $\ccs$ and will be called the class of functions that are Skolem
elementary. $\ccs$ is a minimal class that contains functions
(\ref{eq_basic_functions}) and is closed with respect to
superposition and restricted summation of the kind $\sum_{x\leq y}$.
For the class that is lacking any sort of description based on the
complexity of machine calculations. It is known that
$\ccs\subseteq\bigeps^2$, but the question of coincidence of these
classes is an open one at this time (see ~\cite{erf}.) Besides, it
is known that $\ccs$ contains $\mathrm{NP}$-hard functions (see the
proof for another class in ~\cite{wrathall}, connection with the
class $\ccs$ in ~\cite{erf}).

These classes differ in terms of the speed of growth of the
contained functions. For example all functions of the classes
$\bigeps^2$ and $\ccs$ constrained by polynomials and $\bigeps^3$
contain functions that grow exponentially. To distinguish the
computational complexity of functions in its pure form for every
class $Q$ one considers set $Q_*$  of all predicates with
characteristic functions from $Q$ (the characteristic function of a
predicate is the function that equals to $1$ for all sets, at which
the value of the predicate is true and equals to $0$ at all other
sets). It is known that ~\cite{gzh} the hierarchy $\bigeps^n_*$,
$n\geq 2$ is strictly monotonous. Besides, it is known that
~\cite{marchskolem},
$$
\ccs_*\subseteq\bigeps^0_*.
$$
The question is, which of the classes
 $\ccs_*$, $\bigeps^0_*$,
$\bigeps^1_*$, $\bigeps^2_*$ coincide and which ones do not, and it is an open one.

One will also note that in terms of approximating the class of
practically computable functions, it is convenient to consider the
class $\ccfp$ of functions Turing machine computable in polynomial
time (of the input length). The class $\ccfp$ also has equivalent
functional definitions (see ~\cite{cobham,belcook}). They are more
involved then the ones described above; other definitions, thus,
will not be described in this paper.

\subsection{Finite Superposition Bases}
For the first time the question of generating quite broad and
substantially interesting classes of recursive functions with the
help of only the superposition operation was considered by
Grzegorczyk in 1953 in his paper ~\cite{gzh}. The superposition
basis in a class will be defined as a complete system with respect
to superposition in this class. Traditionally, in theory of
recursive functions one does not need to satisfy the requirement of
minimality for such systems. In ~\cite{gzh}, it was demonstrated
that the class of primitive recursive functions does not have a
finite bases with respect to superposition. In this paper, one
considers the existence of such bases in classes $\bigeps^n$.

The interest to the problem of having finite bases with respect to
superposition in classes of recursive functions is due to a few
factors. Firstly, the operation of the superposition is a relatively
weak one, thus, one can expect the functions from similar bases
systems to a significant degree reflect the specifics of the class,
its arithmetic and algorithmic nature and perhaps complexity in one
of its aspects. A well defined bases can be the necessary grounds
for canonical representations that give a chance to compare and
evaluate different parameters of the class of functions. The
definition of the class of functions rooted in its bases is one of
the non-redundant (and in a sense effective) definitions of the
class.

The problem, as was stated by Grzegorczyk, is solved in two steps.
The existence of finite bases with respect to superposition in
classes $\bigeps^n$ ($n\geq 3$) was proved by D. R\"{o}dding in 1964
in his paper ~\cite{redding}. The Redding's proof was so cumbersome
that the bases were not written out in an explicit way. In his paper
~\cite{parsons}, in 1968 Ch. Parsons obtained easier bases with
respect to the superposition in classes $\bigeps^n$ ($n\geq 3$).
R\"{o}dding and Parsons used a method to build their bases, which,
in this paper, is referred to as the method of generating functions.
The gist of this method is that for the function one builds
generating functions that is the function, every value of which
contains information about the initial values in the initial
function. For example, the function $f(x)$ has its generating
function of the type
$$
g(x)=\prod_{i\leq x}p_i^{f(i)},
$$
where $p_i$ is the $i$-th prime number. If functions can be obtained
from other functions with the help of restricted recursion,
summation, etc., then their corresponding generating functions can
be obtained with the help of just the superposition (though with a
set of some helping functions). To use the method of generating
functions one must have functions that grow at least exponentially.
All functions from classes $\bigeps^0$, $\bigeps^1$, $\bigeps^2$ are
constrained by polynomials, thus for them the method of generating
functions does not work.

For the class $\bigeps^2$, the difficulties with building up the
bases were overcome in 1969 by S.S. Marchenkov in his work
~\cite{ss_ustr} (see also ~\cite{ogr_rec}). In this paper, one used
the method based on modeling a kind of Turing machine. An important
role, when building a bases with the help of the "machine method"\,
is played by the numerating functions (the functions that enumerate
tuples). All functions of the classes $\bigeps^0$ and $\bigeps^1$
are constrained from the top by linear functions and, thus, do not
contain numerating functions. The problem of superposition bases
existence in $\bigeps^0$ and $\bigeps^1$ still remains open. Also,
the question of a bases existence in $\ccs$ is still open too (in
it, there are numerating functions but the 'machine'\ method does
not work for it for other reasons).

In 1970 in the paper ~\cite{muchnik} A.A. Muchnik inspired by the
idea of S.S. Marchenkov ~\cite{ss_ustr} proposed a quite simple
method to build bases with the help of the superposition in some
classes of recursive functions. This method is based on using
special functions (called quasi-universal) grounded in numerating
Turing machines. In this paper the existence of bases was proved for
a big family of classes that are defined by using the complexity of
Turing computations as well as based on the results of
~\cite{ritchie} one obtains an alternative proof of existence of
finite bases in $\bigeps^n$, $n\geq 2$.

Of special interest is the problem of building bases that are as
simple as possible. In this direction, there were obtained a few
interesting results. The first salient advancement in this field was
the result ~\cite{ss1} accomplished by S.S. Marchenkov in the year
of 1980: the superposition basis in the class $\cck$ is the system
$$ \{x+1,\quad \left[\frac{x}{y}\right],\quad x^y,\quad
\varphi(x,y)\},
$$
where $\varphi(x,y)$ for $x>1$ equals to the least index of the zero
digit in the representation of the number $y$ in a positional number
system with base $x$, when $x\leq 1$ it equals zero. In the same
paper, it was shown that the superposition of functions
$$
x+1,\quad x\dotminus y,\quad \left[\frac{x}{y}\right],\quad x^y
$$
is the one for which one can obtain all the functions from $\cck$
that take a finite number values. In 1989, the work of ~\cite{ss2}
proved that the basis in the class $\cck$ is
$$
\{x\dotminus 1,\quad \left[\frac{x}{y}\right],\quad 2^{x+y},\quad
\sigma(x)\},
$$
where $\sigma(x)$ is the number of ones in the binary representation
of $x$. Note that in all the above mentioned bases in the class
$\cck$ in addition to standard arithmetic functions it contains a
one "bad"\ function that, although it has a very simple form, is not
in its pure form an arithmetic one. S. Mazzanti in 2002 in the work
of ~\cite{mazzanti} managed to get rid of the "bad"\ function. In
this work, he proved that
$$
\{x+y,\ x\dotminus y,\ \left[\frac{x}{y}\right],\ 2^x\}
$$
is the basis for the superposition in the class $\cck$.

As an example of application of the above results,
one gives the formula for the binomial coefficient:
$$
\binom{x}{y} =
\left[\frac{(2^{x+1}+1)^x}{2^{(x+1)y}}\right]\dotminus\left[\frac{(2^{x+1}+1)^x}{2^{(x+1)(y+1)}}\right]\cdot
2^{x+1}.
$$

\section{General Description of the Paper}

In the following one presents the goals of this dissertation:
\begin{itemize}
\item
description of classes of functions that can be obtained by the
superposition of basic arithmetic functions with different
restrictions imposed on their growth and the build-up of formulas;
\item
bulding up finite superposition bases of a simple form in classes
analogous to class $\bigeps^2$ of Grzegorczyk without using the
Turing machine numeration;
\item researching groups with recursive permutations that are
connected with the known classes of recursive functions concerning
the subject of finite generability.
\end{itemize}

The results ~\cite{ss1,ss2,mazzanti} about bases in class $\cck$,
regardless of its beauty, simplicity, and that fact that they are
simply amazing have a significant downfall: they cannot be applied
in real life due to the fact that the class $\cck$ contains
functions with very high computational complexity (for example,
$x^{y^z}$) and is, therefore, a very bad approximation for the class
of functions that "can be computed in practice". The technique
offered in these papers to a significant extend uses functions with
super exponential growth and, thus, does not allow to obtain
analogous results for classes significantly smaller than $\cck$.

The technique from papers ~\cite{ogr_rec,ss_ustr,muchnik} allows to
build bases in narrower classes of complexity than $\cck$ (such as
for example $\bigeps^2$ or $\ccfp$) but these bases turn out to be
quite 'bulky' (one of these bases functions is defined based on the
numeration of some types of Turing machines). No other ways for
obtaining simpler bases in congruent classes were known. Moreover,
there was a hypothesis claiming that they cannot be significantly
simpler than those built based on the Turing machine numeration.

As was mentioned, a goal of this dissertation is the one of
obtaining an "easy"\ bases that are analogous to those known for
$\cck$, for narrower classes that estimate the notion of practically
effective computability.

In chapter \ref{chapter_arithm} for some classes that can be
considered "generalized"\ complexity classes one builds bases that
consist only of the simplest arithmetic functions and functions that
are standard in the majority of programming languages such that in
some classes (not closed with respect to superposition) the bases
can be obtained by applying restrictions on formulas. In this
dissertation, for the first time one investigated the question of
describing classes of recursive functions that can be obtained by
the superposition of functions with restrictions on the skeleton of
the superposition.

Let one consider the following class $\ccs$. All functions of this
class can be computed over the course of an exponential time with
linear memory (of the length of the input), thus $\ccs$ gives a much
better approximation of practically computable functions than
$\cck$. The question of existence of finite bases in $\ccs$ remains
unanswered but nonetheless one managed to figure out a description
of this class in terms of the superposition. In section
\ref{section_exp_skolem} of the chapter \ref{chapter_arithm} one
introduces the class $\ccxs$. All functions of this class are
restricted by functions of the type $2^{p(\tilde{x})}$, where $p$ is
a polynomial. $\ccs$ coincides with the set of all functions from
$\ccxs$, bounded by polynomials, thus one calls $\ccxs$ the
exponential expansion $\ccs$. The class $\ccxs$ is not closed with
respect to superposition (thus, there cannot be a bases in it).
Regardless of that this class is a fairly natural one and has a few
equivalent definitions. The main result of the section
\ref{section_exp_skolem} of chapter \ref{chapter_arithm} is the fact
that $\ccxs$ there is a set of all functions that can be expressed
in terms of the superposition of functions
$$
x+1,\quad xy,\quad x\dotminus y,\quad x\wedge y,\quad
\left[x/y\right],\quad 2^x,
$$
where $x\wedge y$ is a bitwise conjunction of binary representations
of $x$ and $y$ with the following restriction on formulas: the
formula needs to have a height of no more than $2$. The height of
the formula is calculated with respect to exponent, i.e for example
the height of the formula $x2^{x+yz}+2^t$ equals to $2$, and the
formula $2^{2^x}$ it equals to $3$.

In section \ref{section_ffom} of the chapter \ref{chapter_arithm},
one considers the class $\ccffom$ that is a functional analog of the
class $\ccfom$ (in English that is "First Order with respect to
Majority," see ~\cite{uniformity}). The class $\ccfom$ is defined
based on the representation of dictionary-based predicates with the
help of first order logical formulas with generalized quantifiers
for majorizing. For example in ~\cite{uniformity} there are a few
equivalent definitions of class $\ccfom$, amongst which there are
those defined in terms of complexity theory. All functions from
$\ccffom$ are bounded by the functions of the type $2^{[\log_2
(x_1+\ldots+x_n)]^n}$, i.e. for any function
$f(\tilde{x})\in\ccffom$ the length of input $f(\tilde{x})$ is
bounded by a polynomial of input length $\tilde{x}$. Generally
speaking, the classes $\ccfom$ and $\ccffom$ can be called "very
small". All functions from $\ccffom$ are computed within a
polynomial time (and, moreover, with a logarithmic space, see
~\cite{uniformity}). Regardless, $\ccffom$ contains the majority of
effectively computable functions that can be encountered while
practicing mathematics that are suitable with respect to their
growth rate. Besides, $\ccffom$ has the property of computational
completeness, specifically all recursively enumerable sets can be
enumerated by functions from $\ccffom$ (see for example
~\cite{uniformity}), $\ccffom$ can be considered as a "generalized"
\ complexity class. The main result of section \ref{section_ffom} of
the chapter \ref{chapter_arithm} is the fact that the system of
functions
$$
\{x+y,\quad x\dotminus y,\quad x\wedge y,\quad
\left[x/y\right],\quad 2^{[\log_2 x]^2}\}
$$
is a basis with respect to the superposition in $\ccffom$. One can
note that $\ccffom$-reducibility is quite strong (see
~\cite{arythm,uniformity}). The analysis of proofs from
~\cite{garyjohnson,pcomplete} shows that the majority of known
$\mathrm{NP}$-complete, $\mathrm{PSPACE}$-complete, as well as
$\mathrm{P}$-complete (with respect for example to reducibility with
the logarithmic space) are of this kind also with respect to
$\ccffom$-reducibility \footnote{It is known that
$\ccfom\neq\mathrm{PSPACE}$, but the question of coincidence of the
classes $\ccfom$ and $\mathrm{NP}$ is at the moment an open one.}.
This means that what was obtained in section \ref{section_ffom} of
the chapter \ref{chapter_arithm} is the result that can be used to
build bases of a simple kind in many known classes. For example to
build a basis in $\ccfp$ it is sufficient to add to the basis in
$\ccffom$ any $\ccfp$-complete function with respect to
$\ccffom$-reducibility that can be constructed for example based on
$\mathrm{P}$-complete problems from ~\cite{pcomplete}.

In section \ref{section_hierarchy} of chapter \ref{chapter_arithm},
one considers classes of functions that can be represented by
formulas analogous to those that one considers in section
\ref{section_exp_skolem} of this same chapter with an arbitrary
height. Thus, this is a hierarchy of classes that is exhaustive of
the class $\cck$ (each height that is bigger or equal to  $2$
corresponds its own class). In section \ref{section_hierarchy} one
considers equivalent definitions of classes in this hierarchy that
are based on substituting into functions of classes $\ccs$ and
$\ccfom$ monotonous functions with corresponding speeds of growth.

One can note that for all bases described in chapter
\ref{chapter_arithm}, there is a function $x\wedge y$, which is a
bitwise conjunction that is not in and of itself an "arithmetic
one."\ (although it is in the set of standard arithmetic functions
of the majority of programming languages and CPUs). Unfortunately,
one cannot get rid of this function analogously to how it was done
in ~\cite{mazzanti} for the class $\cck$.
\ \\

The main task of the chapter \ref{chapter_gzhbasis} is building a
basis of a simple type in the class $\bigeps^2$ of Grzegorczyk
hierarchy ~\cite{gzh}. One can notice that (see
~\cite{arythm,uniformity}) functions of a simple type bounded by a
polynomial analogous to those that were used to build bases in
chapter \ref{chapter_arithm} (for example, $[\sqrt{x}]$, $[\log_x
y]$, $\min(x,y^z)$, various easy operations in binary notation or
other forms of representation and etc.) are in $\ccffom$ and,
therefore, are computable with a logarithmic space, that is they are
in $\bigcup_{C_1,C_2}\mathrm{FSPACE}(C_1\log n+C_2)$, where
$\mathrm{FSPACE}(f(n))$ is the set of all functions computable on
multitape Turing machines that do not record onto input tape and do
not read from the output one, with a restriction on space $f(n)$,
$n$ is the length of entrance (see ~\cite{uniformity,hartmanis}). On
the other hand, in agreement with ~\cite{ritchie}, $\bigeps^2$ is
the set of all functions that are Turing machine computable with a
linear space. Thus, from the theorem on hierarchy ~\cite{hartmanis}
it follows that if $\Phi$ is a basis in $\bigeps^2$ and $f(n)=o(n)$,
then in $\Phi$ there is a function that is not in
$\mathrm{FSPACE}(f(n))$. This means that the basis in $\bigeps^2$
must contain a function that is significantly more complex than the
ones considered above, those being "simple arithmetic"\ ones. In
chapter \ref{chapter_gzhbasis}, there is an example of a basis that
consists of simple arithmetic functions and a special function
$Q(x,p_1,p_2,c_1,c_2,t)$. Function $Q$ is defined with the help of
primitive recursion, its definition is a very simple one and it does
not contain in an explicit way any type of Turing machine
numeration. The function $Q$ in some sense is quasi-universal in
$\bigeps^2$ (the definition of quasi-universality is slightly
different from the one introduced by A. A. Muchnik in
~\cite{muchnik}). One can note that the function $Q$ is interesting
also as a very simple example of $\mathrm{PSPACE}$-equivalent
function (see ~\cite{garyjohnson}.)
\ \\

In chapter \ref{chapter_per}, one investigates special classes of
functions, the classes of permutations. More specifically one
considers groups of permutations $\gr(Q)=\{f:\ f,f^{-1}\in Q\}$ for
classes $Q$, closed with respect to the superposition and those that
contain an identity function. The main result of the chapter
\ref{chapter_per} is the proof of finite generability $\gr(Q)$ for a
big family of classes $Q$ (that satisfy specific requirements). The
requirements have a purely "functional"\ character, one makes no
assumptions about the computability of functions from $Q$. The proof
of this statement is constructive, the permutations of the
generating set are being built from functions of basis in $Q$,
numerating functions as well as some basic arithmetic functions. A
special interest is classes $Q$ being an estimation of the class of
functions that is computable in practice. In this case, $\gr(Q)$ is
the set of all effective non-redundant codes $\nat\rightarrow\nat$
that allow for effective decoding. Such codes are used both for
compressing information and encoding it. Searching for finite
generating sets of such groups gives a lot of information about the
structure of these groups as well as it gives an easy and effective
method for enumerating elements of these groups.

In chapter \ref{chapter_per} one proves finite generability of the
group $\gr(Q)$ for classes $\ccfp$, $\ccffom$, generalizations of
the Grzegorczyk classes and etc. One can note that the classes of
permutations $\gr(Q)$ are subclasses of classes of single valued
functions $Q^{(1)}$, the existence of bases in $Q^{(1)}$ for a big
family of classes $Q$ is proved in ~\cite{basis}. When proving
finite generability $Q^{(1)}$ in ~\cite{basis} one uses the fact
that the superposition allows to select from functions an
information that is required and to get rid of the unnecessary, for
example the value $f(g(x))$ is not dependent on the values of
function $f$ when using arguments that do not belong to the image of
$g$. This allows one to use quasi-universal functions that are
analogous to those that are used in papers
~\cite{ss_ustr,muchnik,ogr_rec}, i.e. functions that contain in a
sense the information about all functions from $Q^{(1)}$ (with the
help of auxiliary functions one can extract that part of information
from the quasi-universal one that corresponds to some specific
function and with the help of other axillary functions one can build
the needed function). For classes of permutations such method does
not work, to prove finite generbility $\gr(Q)$ one uses a new
method, that was generated specifically in the framework of this
dissertation.

Of a special interest is the problem of minimizing the generating
set. In chapter \ref{chapter_per} one proves that for the same
requirements for $Q$ the cardinality of the minimal generating set
equals to two (more specifically, only the upper bound is proved,
the lower bound follows for example from the non-commutativity).
Moreover, it is proved that there exists a two-element set that
generates $\gr(Q)$ in a functional sense, i.e. the basis with
respect to superposition of two functions.

The main results of this dissertation were presented at
international conferences "Discrete Models in Control Systems
Theory"\ (Moscow, 2006), "Problems of Theoretical Cybernetics"\
(Kazan, 2008), research seminar at the Institute for Information
Transmission Problems, research seminar at the Department of
Mathematical Logic and Theory of Algorithms at the Faculty of
Mechanics and Mathematics at Lomonosov Moscow State University,
research seminar at the Department of Mathematical Cybernetics at
the Faculty of Computational Mathematics and Cybernetics, Lomonosov
Moscow state University, published in papers
~\cite{perbasis,perbasiskonf,doklady,etwobasis,diplom}.

\section{Summary of the Main Results}

\subsection{Basic definitions}\refstepcounter{subsectionref}\label{subsection_definitions}

For the reader's convenience  some of the definitions here and in
other parts of the paper repeat the ones introduced in the review
part.

Let $\nat=\{0,1,2,\ldots\}$. One considers everywhere defined
functions (with an arbitrary number of arguments) on the set of
$\nat$. Under the operation of superposition one means the
substitution of functions into functions, permutation and
identification of variables, introduction of dummy variables.

Let $Q$ be an arbitrary class of functions over $\nat$. One will
denote through $\left[ Q \right]$ the closure over superposition of
the class $Q.$

Let there be some set $\Psi$ of functions closed under the
superposition and $\Phi\subseteq \Psi$. One considers that the
$\Phi$ set generated the set $\Psi$ if $[\Phi]=\Psi$. Finite sets
generating  $\Psi$ are called \emph{finite superposition basis} in
the set $\Psi.$

Let one assume that
$$
x\dotminus y=\max(x-y,0),
$$
$$
\ffsg(x)=\begin{cases} 1,\text{ if }x>0,\\0\text{
else,}\end{cases}
$$
$$
\overline{\sg}(x)=\begin{cases} 0,&\text{if}\ x>0, \\
1,&\text{if}\ x=0, \end{cases}
$$
$$
\ffrm(x,y)=\begin{cases}\text{remainder from dividing $x$ by $y,$ if
$y>0,$} \\ 0\text{ else},\end{cases}
$$
$$
\left[\log_2 x\right]=\begin{cases} \text{integer part of a binary logarithm $x,$ if $x>0,$ }\\0\text{ else,}\end{cases}
$$
$$
x\langle y\rangle=y\text{-th binary digit of $x$}
$$
$$
\text{(therefore, $x=\sum_{y=0}^{\infty}x\langle y\rangle\cdot
2^y$),}
$$
\begin{equation}\label{eq_len}
\fflen(x)=([\log_2 x]+1)\cdot\ffsg(x).
\end{equation}

One can see that $\fflen(x)$ equals to the length of the binary
notation for $x$ if $x>0$ and zero otherwise. Let one define the
function $x\wedge y$ as the bit-wise conjunction of the binary
representations of numbers $x$ and $y$. Let there be  $a_n
a_{n-1}\ldots a_0,\,b_n b_{n-1}\ldots b_0$ as binary representations
of  numbers $x$ and $y$ (if the lengths of the binary
representations are different, then the most significant bit of the
binary representation of the smaller number equals to zero). Thus,
the binary representation of a number $x\wedge y$ is
$$
(a_n\cdot b_n)(a_{n-1}\cdot b_{n-1})\ldots (a_0 \cdot b_0).
$$
By characteristic function of the predicate $\rho(x_1,\ldots,x_n)$
we call the function  $\chi_{\rho}(x_1,\ldots,x_n)$ such that for
any $x_1,\ldots,x_n$
$$
\chi_{\rho}(x_1,\ldots,x_n)=\begin{cases}1,\text{ if
}\rho(x_1,\ldots,x_n)\text{ is true},\\ 0,\text{ otherwise}.\end{cases}
$$

For the class of functions $Q$ by $Q_*$ one denotes the set of all
predicates, which characteristic functions lie in $Q.$

One claims that the function $f(x_1,\,\ldots,\,x_n,\,y)$ is obtained
from the functions $g(x_1,\,\ldots,\,x_n),$
$h(x_1,\,\ldots,\,x_n,\,y,\,z)$, $j(x_1,\,\ldots,\,x_n,\,y)$ with
the help of using the operation called \textit{bounded recursion} if
the following relations hold true
$$
\left\{
\begin{array}{l}
f(x_1,\,\ldots,\,x_n,\,0)=g(x_1,\,\ldots,\,x_n), \\
f(x_1,\,\ldots,\,x_n,\,y+1)=h(x_1,\,\ldots,\,x_n,\,y,\,f(x_1,\,\ldots,\,x_n,\,y)),
\\
f(x_1,\,\ldots,\,x_n,\,y)\leq j(x_1,\,\ldots,\,x_n,\,y).
\end{array}
\right.
$$

Let us define classes $\bigeps^n$ ($n\in\nat$) of the Grzegorczyk
hierarchy ~\cite{gzh}. $\bigeps^n$ is the minimal class of functions
which contains the functions $x+1$, $f_n(x,y)$ and is closed with
respect to superposition and limited recursion where
$$
f_0(x,y)=y+1,
$$
$$
f_1(x,y)=x+y,
$$
$$
f_2(x,y)=(x+1)\cdot (y+1),
$$
for $n\geq 2$
$$
f_{n+1}(0,y)=f_n(y+1,y+1),
$$
$$
f_{n+1}(x+1,y)=f_{n+1}(x,f_{n+1}(x,y)).
$$

For tuples of variables (and their parts), one uses abbreviations of
the form $\tilde{x}$, $\tilde{y}$, etc. (for example,
$(\tilde{x},t)$ is $(x_1,\ldots,x_n,t)$.)

\subsection{Main results of Chapter \ref{chapter_arithm}}\refstepcounter{subsectionref}\label{subsection_basic_arithm}

One can say that the function $f(x,z_1,\ldots,z_n)$ can be obtained from the function
$g(y,z_1,\ldots,z_n)$ with the help of an operation called
\emph{bounded summation} with respect to the $y$ variable if $$
f(x,z_1,\ldots,z_n)=\sum_{y\leq x}g(y,z_1,\ldots,z_n).
$$

The class $\ccclh$ of the \emph{Skolem elementary} functions (see
~\cite{erf,skolemone,skolemtwo}) is a minimal class of functions
that contains functions
\begin{equation}\label{eq_basic_arithm_srcf} 0,\quad x+1,\quad
x\dotminus y
\end{equation}
and is closed with respect to superposition and bounded summation.
One can note that $\ccclh$ coincides with the minimal class that
contains functions (\ref{eq_basic_arithm_srcf}) and is closed with
respect to superposition and summation of the form $\sum_{x<y}$
(summation over an empty set equals to zero), for convenience one
will use specifically this definition.

For every set of functions $Q$ one can define sequences of classes
$\left[ Q \right]_{2^x}^n$ and $\left[ Q \right]_{x^y}^n$
($n=0,1,2,\ldots$) inductively.
\begin{enumerate}
\item $\left[ Q \right]_{x^y}^0=\left[ Q \right]_{2^x}^0=[Q].$
\item If $f\in\left[ Q \right]_{2^x}^n$ ($\left[ Q
\right]_{x^y}^n$), then $f\in\left[ Q \right]_{2^x}^{n+1}$ ($\left[
Q \right]_{x^y}^{n+1}$).
 \item If $f\in\left[ Q \right]_{2^x}^n$ ($f\in\left[ Q \right]_{x^y}^n$) and $g$
is obtained from $f$ by permuting, identifying variables, or
introducing dummy variables, then $g\in\left[ Q \right]_{2^x}^n$
($g\in\left[ Q \right]_{x^y}^n$). \item If $f(y_1,\ldots,y_m)\in Q$
and $g_1(x_1,\ldots,x_k),\,\ldots,\,g_m(x_1,\ldots,x_k)\in\left[ Q
\right]_{2^x}^n$ ($\left[ Q \right]_{x^y}^n$), then
$$f(g_1(x_1,\ldots,x_k),\ldots,g_m(x_1,\ldots,x_k))\in\left[ Q
\right]_{2^x}^n\ (\left[ Q \right]_{x^y}^n).$$  \item If
$f\in\left[ Q \right]_{x^y}^{n+1},$ $g\in\left[ Q
\right]_{x^y}^n,$ $h\in\left[ Q \right]_{2^x}^n,$ then $2^h\in\left[
Q \right]_{2^x}^{n+1}$ and $f^g\in\left[ Q \right]_{x^y}^{n+1}.$
\end{enumerate}

For $[Q]^1_{2^x}$ and $[Q]^1_{x^y}$ one can use contracted notation
$[Q]_{2^x}$ and $[Q]_{x^y}$ respectively.

One can introduce a sequence of classes $\ccexppoly^n$
($n=0,1,2,\ldots$) inductively. \begin{enumerate} \item
$\ccexppoly^0$ is the class of all polynomials. \item $\ccexppoly^{n+1}$ is the class of all functions of the type $2^f,$ where $f\in\ccexppoly^n.$
\end{enumerate}

One can define the class $\cckhs^n$ ($n=0,1,2,\ldots$) as a class of all functions
$f(x_1,\ldots,x_n),$ for which they satisfy the following conditions:
\begin{enumerate}
\item $f$ is bounded by some function from $\ccexppoly^{n+1}.$
\item There exist functions $m(x_1,\ldots,x_n)\in \ccexppoly_n$ and
$g(x_1,\ldots,x_n,y,z)\in\ccclh,$ such that
$$
f(x_1,\ldots,x_n)\langle
y\rangle=g(x_1,\ldots,x_n,y,m(x_1,\ldots,x_n)).
$$
\end{enumerate}

$\ccxs$ one can define the class of all functions as $f(x_1,\ldots,x_n),$ for which the following conditions hold true:
\begin{enumerate}
\item There exists a polynomial $p(x_1,\ldots,x_n)$ with natural coefficients such that for any
$x_1,\ldots,x_n$ it is true that the inequality
$$
f(x_1,\ldots,x_n)<2^{p(x_1,\ldots,x_n)}.
$$
\item $f(x_1,\ldots,x_n)\langle y\rangle\in\ccclh.$
\end{enumerate}

It is obvious that $\ccclh\subseteq\ccxs$ and $\cckhs^0=\ccxs$ (this can be proved using the technique from ~\cite{erf}.)

One can assume that
$$
\mains=\{x+1,\quad xy,\quad x\dotminus y,\quad x\wedge y,\quad
\left[x/y\right]\}.
$$

\begin{theorem}\label{theorem_exp_skolem}
$\ccxs=\left[ \mains \right]_{2^x}=\left[ \mains \right]_{x^y}.$
\end{theorem}
This theorem is proved in section \ref{section_exp_skolem} of chapter
\ref{chapter_arithm}. \ \\

If $A$ is some alphabet, then one can denote $A^+$ as the set of all
finite non-empty words in the alphabet $A.$ If $X$ is a word in the
alphabet $A,$ then one can denote $|X|$ as the length of this word.

One can name $\ccfom$-\emph{term} over variables $x_1,\ldots,x_m$
the expression of form $x_1,\ldots,x_m,1,\lln.$

\Def{Expressions of the form $(t_1\leq t_2),$ $\llbit(t_1,t_2)$ or
$\llx\langle t_1\rangle,$ where $t_1,t_2$ are $\ccfom$-terms over
variables $x_1,\ldots,x_m,$ are called \emph{elementary
$\ccfom$-formulas} over variables $x_1,\ldots,x_m$.}

One can inductively define the notion of $\ccfom$-formula over
variables $x_1,\ldots,x_m.$\footnote{In the list of variables, there
are not only free but also bound variables, technically it is more
convenient.}
\begin{enumerate}
\item All elementary $\ccfom$-formulas over $x_1,\ldots,x_m$
are $\ccfom$-formulas over $x_1,\ldots,x_m.$ \item If
$\Phi_1$, $\Phi_2$ are $\ccfom$-formulas over variables
$x_1,\ldots,x_m$, $x_i\in\{x_1,\ldots,x_m\},$ then
$(\Phi_1\&\Phi_2),$ $(\Phi_1\vee\Phi_2),$ $(\neg\Phi_1),$
$(\exists x_i)(\Phi_1),$ $(\forall x_i)(\Phi_1),$ $(\qqm
x_i)(\Phi_1)$ are $\ccfom$-formulas over $x_1,\ldots,x_m.$
\end{enumerate}

To every $\ccfom$-term $t$ over variables $x_1,\ldots,x_m$ one will
match up the function $h_t(X,x_1,\ldots,x_m)$, which is defined over
the set of all arrays $(X,x_1,\ldots,x_m)$ such that $X\in\{0,1\}^+$
and $1\leq x_1,\ldots,x_m\leq |X|$, in the following way.
\begin{enumerate}
\item If $t$ is $1,$ then
$$h_t(X,x_1,\ldots,x_m)=1.$$
\item If $t$ is $\lln,$ then
$$h_t(X,x_1,\ldots,x_m)=|X|.$$
\item If $t$ is $x_i,$ then
$$h_t(X,x_1,\ldots,x_m)=x_i.$$
\end{enumerate}
For every elementary $\ccfom$-formula $\Phi$ over variables
$x_1,\ldots,x_m$ one can match up the predicate
$\rho_{\Phi}(X,x_1,\ldots,x_m)$, the domain of which coincides with
the domain of the function for $\ccfom$-terms over $x_1,\ldots,x_m$,
in the following way.
\begin{enumerate}
\item If $\Phi$ is of the type $(t_1\leq t_2),$ then
$$
\rho_{\Phi}(X,x_1,\ldots,x_m)\equiv(h_{t_1}(X,x_1,\ldots,x_m)\leq
h_{t_2}(X,x_1,\ldots,x_m)).
$$
\item If $\Phi$ is of the type $\llbit(t_1,t_2),$ then
$$
\rho_{\Phi}(X,x_1,\ldots,x_m)\equiv(h_{t_1}(X,x_1,\ldots,x_m)\langle
h_{t_2}(X,x_1,\ldots,x_m)-1\rangle=1).
$$
\item If $\Phi$ is of the type $\llx\langle t_1\rangle,$ then
$$\rho_{\Phi}(X,x_1,\ldots,x_m)\equiv(h_{t_1}(X,x_1,\ldots,x_m)\text{-s symbol of the word $X$ equals $1$}),$$ where the numeration of the symbols starts with one (and one numerates symbols left to right).
\end{enumerate}
Every $\ccfom$-formula $\Phi$ over variables  $x_1,\ldots,x_m$
matches the predicate $\rho_{\Phi}(X,x_1,\ldots,x_m),$ the domain of
which coincides with the domain of the function for $\ccfom$-terms
over $x_1,\ldots,x_m,$ in the following way.
\begin{enumerate}
\item If the formula is elementary $\ccfom$-formula,
then its corresponding predicate coincides with the predicate that is defined for the given elementary formula. \item If $\Phi$ is of the type
$(\Phi_1\&\Phi_2),$ then
$$
\rho_{\Phi}(X,x_1,\ldots,x_m)\equiv\rho_{\Phi_1}(X,x_1,\ldots,x_m)\&\rho_{\Phi_2}(X,x_1,\ldots,x_m).
$$
\item If $\Phi$ is of type $(\Phi_1\vee\Phi_2),$ then
$$
\rho_{\Phi}(X,x_1,\ldots,x_m)\equiv\rho_{\Phi_1}(X,x_1,\ldots,x_m)\vee\rho_{\Phi_2}(X,x_1,\ldots,x_m).
$$
\item If $\Phi$ is of type $(\neg\Phi_1),$ then
$$
\rho_{\Phi}(X,x_1,\ldots,x_m)\equiv\neg\rho_{\Phi_1}(X,x_1,\ldots,x_m).
$$
\item If $\Phi$ is of  type $(\exists x_i)(\Phi_1),$ then
$$
\rho_{\Phi}(X,x_1,\ldots,x_m)\equiv(\exists x)_{(1\leq x\leq
|X|)}\rho_{\Phi_1}(X,x_1,\ldots,x_{i-1},x,x_{i+1},\ldots,x_m).
$$
\item If $\Phi$ is of type $(\forall x_i)(\Phi_1),$ then
$$
\rho_{\Phi}(X,x_1,\ldots,x_m)\equiv(\forall x)_{(1\leq x\leq
|X|)}\rho_{\Phi_1}(X,x_1,\ldots,x_{i-1},x,x_{i+1},\ldots,x_m).
$$
\item If $\Phi$ is of type $(\qqm x_i)(\Phi_1),$ then
$\rho_{\Phi}(X,x_1,\ldots,x_m)$ is true if and only if when the
number of $x$ such that $1\leq x\leq |X|$ and
$\rho_{\Phi_1}(X,x_1,\ldots,x_{i-1},x,x_{i+1},\ldots,x_m)$ is truly
greater than $|X|/2$.
\end{enumerate}

The $\ccfom$ (see ~\cite{uniformity}) is defined at the set of
everywhere defined over the set $\{0,1\}^+$ predicates $\varphi(X),$
for which there exists $\ccfom$-formula that corresponds to the
predicate $\rho(X,x_1,\ldots,x_m)$ such that for any
$X,x_1,\ldots,x_m$ from its domain
$$
\varphi(X)\equiv\rho(X,x_1,\ldots,x_m).
$$

If $x_1,\ldots,x_m$ are natural numbers, then one can denote
$\ffcode(x_1,\ldots,x_m)$ word
$$
01s_101s_201\ldots01s_m01,
$$
where
$$
s_i=\begin{cases} \text{to an empty word if $x_i=0,$} \\
\text{to the word that one obtains from binary notation $x_i$} \\ \quad\quad\text{by substituting each one by 11,} \\
\quad\quad \text{and every zero by $00,$ if $x_i\neq 0.$}
\end{cases}
$$
The class $\ccfomn$ can be defined as the set of everywhere defined over the set $\nat$ predicates $\varphi(x_1,\ldots,x_n),$ for which there exist a predicate $\psi(X)\in\ccfom$ such that for any
$x_1,\ldots,x_n$ it satisfies
$$
\varphi(x_1,\ldots,x_n)\equiv\psi(\ffcode(x_1,\ldots,x_n)).
$$

The class $\ccffom$ can be defined as the set of everywhere defined over the set $\nat$ functions $f(x_1,\ldots,x_n)$ such that the following two conditions are satisfied.
\begin{enumerate}
\item There exists a polynomial $p(y_1,\ldots,y_n)$ such that for any
$x_1,\ldots,x_n$
$$
f(x_1,\ldots,x_n)\leq 2^{p([\log_2(x_1)],\ldots,[\log_2(x_n)])}.
$$
\item The predicate $\rho$ that can be defined by the following relation
$$
\rho(x_1,\ldots,x_n,y)\equiv(f(x_1,\ldots,x_n)\langle y\rangle=1),
$$
is in $\ccfomn.$
\end{enumerate}

\begin{theorem}\label{theorem_main_ffom}The following takes place
$$
\ccffom=\left[x+y,\quad x\dotminus y,\quad x\wedge y,\quad
\left[x/y\right],\quad 2^{[\log_2 x]^2} \right]=
$$
$$
=\left[x+y,\quad x\dotminus y,\quad x\wedge y,\quad
\left[x/y\right],\quad x^{[\log_2 y]} \right].
$$
\end{theorem}
This theorem is proved in section \ref{section_ffom} of the chapter
\ref{chapter_arithm}. \\

One can introduce the following sequence of classes $\ccexppoly^n$
($n=0,1,2,\ldots$) inductively. \begin{enumerate} \item
$\ccexppoly^0$ is the class of all polynomials. \item
$\ccexppoly^{n+1}$ is the class of all functions of the type $2^f,$
where $f\in\ccexppoly^n.$
\end{enumerate}

One can define the class $\ccffomh^n$ ($n=1,2,\ldots$) as a class of
all upper-bounded by functions from $\mathrm{P}^n$ functions
$f(x_1,\ldots,x_n),$ for which there exists a function
$m(x_1,\ldots,x_n)\in \mathrm{P}^n$ and a predicate
$\rho(x_1,\ldots,x_n,y,z)\in\ccfomn$ such that
$$(f(x_1,\ldots,x_n)\langle y\rangle=1)\equiv\rho(x_1,\ldots,x_n,y,m(x_1,\ldots,x_n)).$$

\begin{theorem}\label{theorem_main_h}
For any $n\geq 0$ there is
$$
\cckhs^n=[\mains]_{2^x}^{n+1}=[\mains]_{x^y}^{n+1}=\ccffomh^{n+1}.
$$
\end{theorem}
This theorem is a generalization of theorem \ref{theorem_exp_skolem} and
is proved in section \ref{section_hierarchy} of chapter
\ref{chapter_arithm}.

\subsection{Main results of Chapter \ref{chapter_gzhbasis}}\refstepcounter{subsectionref}\label{subsection_basic_gzhbasis}

One can define $R(x,\,y)$ as a cyclic shift of a binary
representation of the number $x$ by $y$ digits to the right. In
other words, let $R(0,\,y)=0,$ $R(1,\,y)=1,$ and if $x\geq 2$ and
$a_n a_{n-1} \ldots a_1 a_0$ is the binary representation $x,$ such
that $a_n=1,$ then the binary notation$R(x,\,y)$ is
$$
a_{n+y} a_{n+y-1} \ldots a_{1+y} a_y,
$$
where all additions go mod $n+1.$

From ~\cite{gzh} it follows that the functions
$$
\sg(x),\quad \overline{\sg}(x),\quad x\dotminus y,\quad \left[
\frac{x}{y} \right],\quad \left[ \log_2 x \right],\quad
\min(x,2^y),\quad \rrm(x,y)
$$
belong to Grzegorczyk class $\bigeps^2$.

With the help of for example the operation of bounded summation
~\cite{erf} it is not difficult to show that the function $x\wedge y,\,R(x,\,y)$
belongs to the class $\bigeps^2$.

Let one define the function $Q(x,\,p_1,\,p_2,\,c_1,\,c_2,\,t)$ by
the following primitive recursion:
$$
\left\{
\begin{array}{l}
Q(x,\,p_1,\,p_2,\,c_1,\,c_2,\,0)=x, \\
Q(x,\,p_1,\,p_2,\,c_1,\,c_2,\,t+1) = \\ \quad = \begin{cases}
Q(x,\,p_1,\,p_2,\,c_1,\,c_2,\,t),\,\text{if}\
Q(x,\,p_1,\,p_2,\,c_1,\,c_2,\,t)\wedge R(p_1,\,c_1\cdot t)\neq 0,
\\ Q(x,\,p_1,\,p_2,\,c_1,\,c_2,\,t)+R(p_2,\,c_2\cdot t)\ \text{в
otherwise.} \end{cases}
\end{array}
\right.
$$
Since $Q(x,\,p_1,\,p_2,\,c_1,\,c_2,\,t)\leq x+2p_2 t,$ then one has
a bounded recursion in the class $\bigeps^2$ and, thereby,
$Q\in\bigeps^2.$

A function $Q(x_1,\,x_2,\,\ldots,\,x_m)$ from the class $\bigeps^2$
one will name \emph{quasi-universal} in the class $\bigeps^2$
relative to the system of the functions $\Phi$ if for any function
$f(\tilde{y})$ from the class $\bigeps^2$ one can find functions
$h(x,\,\tilde{y}),\,g_1(\tilde{y}),\,g_2(\tilde{y}),\,\ldots,\,g_m(\tilde{y})$
from the set $\Phi,$ such that
$$
f(\tilde{y})=h(Q(g_1(\tilde{y}),\,g_2(\tilde{y}),\,\ldots,\,g_m(\tilde{y})),\tilde{y}).
$$

\begin{theorem}\label{theorem_gzh_quazi}
The function $Q(x,\,p_1,\,p_2,\,c_1,\,c_2,\,t)$ is a quasi-universal one in the class $\bigeps^2$
with respect to closure by the superposition of the system of functions
\begin{equation}\label{arithm}
x+1,\quad xy,\quad \min(x,\,2^y),\quad x\dotminus y,\quad \left[
\frac{x}{y} \right],\quad \left[ \log_2 x \right].
\end{equation}
\end{theorem}
\begin{conseq}
The system of functions
$$
x+1,\quad xy,\quad \min(x,\,2^y),\quad x\dotminus y,\quad \left[
\frac{x}{y} \right],\quad \left[ \log_2 x \right],\quad
Q(x,\,p_1,\,p_2,\,c_1,\,c_2,\,t)
$$
forms a basis with respect to superposition in the Grzegorczyk class
$\bigeps^2$.
\end{conseq}

\subsection{Main results of Chapter \ref{chapter_per}}\refstepcounter{subsectionref}\label{subsection_basic_per}

Under the term \emph{permutation} one assumes a permutation over the
set $\nat.$

For any class $Q,$ which is closed with respect to superposition and
 contains the function $I(x)=x,$ by $\gr(Q)$ one
can denote the group of permutations $(\{f:\ f,f^{-1}\in Q\},\
\circ)$.

\Def{An infinite set $A\subseteq\nat$ is \emph{regular} in the class
of functions $Q$ if it satisfies two conditions:
\begin{enumerate}
\item $\chi_A\in Q;$ \item One can enumerate elements of the set $A$
in such a way that $\mu(x)$ that calculates the number of the
element $x$ in this numeration (equals to zero for $x\notin A$) and
the function $\nu(x)$ that calculates an element with the number $x$
belong to $Q$ (enumeration starts with zero).
\end{enumerate}}

One considers classes of functions $Q$ that satisfy the following
requirements:
\renewcommand{\theenumi}{\Roman{enumi}}
\begin{enumerate}
\item\label{dembase} $Q$ contains functions
\begin{equation}\label{eq_base_functions}
1,\quad x+y,\quad x\dotminus y,\quad x\cdot \ffsg\ y,\quad [x/2];
\end{equation}
\item\label{demnumerate} $Q$ contains an enumerating function
$\ffp_2(x_1,x_2)$ that mutually exclusively maps the set $\nat^2$ to
$\nat$ and its inverse functions $\ffpr_{2,1}(x)$ and
$\ffpr_{2,2}(x)$ ($\ffpr_{2,1}(\ffp_2(x,y))=x,$
$\ffpr_{2,2}(\ffp_2(x,y))=y$ for any $x,$ $y$);
\item\label{demrazb}For any permutation $f\in\gr(Q)$ there
exist non-intersecting regular in $Q$ sets $A,$ $B$ such that
$f(A)\cap B=\varnothing$ and $\nat\backslash A,$ $\nat\backslash B$
are regular in $Q;$ \item\label{demzamkn} $Q$ is closed with respect
to superposition; \item\label{demunarbasis} $Q$ has a finite basis
with respect to superposition.
\end{enumerate}
\renewcommand{\theenumi}{\arabic{enumi}}
One can notice that the requirements are not independent (for example
\ref{demzamkn} follows from \ref{demunarbasis}). Nonetheless, one considers all requirements so that one can show for some statements that they hold true for quite weak restrictions on $Q$ (see chapter \ref{chapter_per}).

\begin{theorem}\label{theorem_per_basic}If the class $Q$ satisfies the
requirements  \textup{\ref{dembase}--\ref{demrazb},
\ref{demunarbasis}}, then there exist two permutations from
$\gr(Q)$, with compositions of which one can represent any
permutation from $\gr(Q)$.
\end{theorem}

Let FP be the set of all functions $\nat^n\rightarrow\nat$
computable on Turing machine and its running time is upper bounded
by a polynomial expression in the size of the input for the
algorithm (the number is expressed in binary code). Similarly, FL
the set of all functions computable with the space $O(\log n)$,
where $n$ is the input length (for multitape Turing machine not
recording on the input tape). Besides, one uses the definition of
the class $\ccffom$ from the section \ref{subsection_basic_arithm}
of Introduction.

The class of functions
 $Q$ is called \emph{$\bigeps^2$-closed} if it contains the
 following functions
$$
0,\quad x+1,\quad xy
$$
and is closed with regards to superposition and bounded recursion.
\begin{theorem}\label{theorem_per_concrete}
Classes $\ccfp$, $\ccfl$, $\ccffom$ as well as $\bigeps^2$-closed
classes that have a finite basis with respect to superposition
satisfy the requirements \textup{\ref{dembase}--\ref{demrazb},
\ref{demunarbasis}}.
\end{theorem}
\begin{conseq}
For classes $Q$ from the theorem \ref{theorem_per_concrete} the
group $\gr(Q)$ is generated by two permutations (also, in a
functional sense, i.e. with the help of using only composition).
\end{conseq}

\chapter{Generating Classes by Superposition of Simple
Arithmetic Functions}\refstepcounter{chapterref}\label{chapter_arithm}

\renewcommand{\theequation}{\arabic{chapter}.\arabic{equation}}

\section{Exponential Expansion of the Class of Skolem Elementary Functions
and a Formula of Height
two}\refstepcounter{sectionref}\label{section_exp_skolem}

   \renewcommand{\thestat}{\arabic{chapter}.\arabic{section}.\arabic{subsection}.\arabic{stat}}

\subsection{Definitions}

For basic definitions one can check sections \ref{subsection_definitions} and
\ref{subsection_basic_arithm} of the introduction.

\Def{Predicate $\rho(x_1,\ldots,x_n)$ is a \emph{correct} one if there exists a function $f(y)\in\left[\mains\right]_{2^x}$ such that for any $y\geq 1$ $$ f(y)=\Sum_{0\leq
x_1<y}\ldots\Sum_{0\leq
x_n<y}(\chi_{\rho}(x_1,\ldots,x_n)2^{x_1+x_2y+\ldots+x_ny^{n-1}}).
$$ } In this case the function $f$ is called the \emph{generating function} of the predicate $\rho.$

Further, the generating function of any predicate $\rho$ will be
denoted as $f_{\rho}.$

\Def{Function $f(x_1,\ldots,x_n)$ is called
\emph{$\mains$-polynomial} with respect to the set of variables
$\{x_{i_1},\ldots,x_{i_k}\}$ if for any functions
$g_1(\tilde{y}),\ldots,g_n(\tilde{y})$  from satisfying relations
$$
\begin{array}{l}g_i\in\left[\mains\right]_{2^x},\text{ if }i\in\{i_1,\ldots,i_k\}, \\
g_i\in\left[\mains\right],\text{ if
}i\notin\{i_1,\ldots,i_k\},\end{array} \quad(1\leq i\leq n)$$
it follows that
$$
f(g_1(\tilde{y}),\ldots,g_n(\tilde{y}))\in\left[\mains\right]_{2^x}.
$$}

By \emph{explicit transformations} one understands operations of
permutation and identification of variables, introduction of dummy
variables, and constants substitution (from the set $\nat$) in place
of variables.

One can say that the predicate $\varphi(x_1,\ldots,x_n,y)$ can be
obtained from the predicate $\psi(x_1,\ldots,x_n)$ with the help of
\emph{counting operation} with respect to the variable $x_i$ and a
polynomial $p(x_1,\ldots,x_n)$ if for any $x_1,\ldots,x_n,y\in \nat$
the value $\varphi(x_1,\ldots,x_n,y)$ holds true if and only if $y$
is the number of such $x$ that $x<p(x_1,\ldots,x_n)$ and
$\psi(x_1,\ldots,x_{i-1},x,x_{i+1},\ldots,x_n)$ hold true.

Let $\cccr$ be a minimal class of predicates that contains
predicates $x+y=z$ and $xy=z,$ closed with respect to explicit
transformations, logical operations and operations of counting.

The class $\ccr$ (see ~\cite{erf}) can be defined as a minimal class
of predicates that contains predicates $x+y=z,$ $xy=z$ and closed
with respect to explicit transformational logical operations and
bounded quantifications of the form $(\exists x)_{x<y}$ and
$(\forall x)_{x<y}.$

The \emph{graph} of the function $f(x_1,\ldots,x_n)$ is a predicate
of the type $y=f(x_1,\ldots,x_n).$

Let $\cccrf$ be the set of all functions upper-bounded by
polynomials, the graphs of which are in $\cccr.$

One can say that the function $f(x,z_1,\ldots,z_n)$ is obtained from the function $g(y,z_1,\ldots,z_n)$ with the help of the operation called \emph{narrowed bounded summation} if for any
$x,z_1,\ldots,z_n\in\nat$
$$
f(x,z_1,\ldots,z_n)=\begin{cases}
\Sum_{y<x}\ffsg(g(y,z_1,\ldots,z_n)),\text{ if }x>0, \\ 0,\text{
if }x=0. \end{cases}
$$

One can assume that
$$
(\mu x_i)_{x_i<y}(f(x_1,\ldots,x_n)=z)=$$
$$=\begin{cases}\text{minimal of these values $x_i,$ such that $x_i<y$ and} \\
\quad\quad\text{$f(x_1,\ldots,x_n)=z$ if such $x_i$
exists,}
\\ 0\text{ otherwise}\end{cases}
$$
The operation $\mu$ is called the operation \emph{of bounded minimization.}

\subsection{Inclusion $\left[ \mains\right]_{x^y}\subseteq\ccxs$}\refstepcounter{subsectionref}\label{subsection_expskolem_burjui}

\begin{stat}\label{statangles}
The function $x\langle y\rangle$ lies in $\ccclh.$ Besides, the predicate
$(x\langle y\rangle=1)$ is in $\ccclh_*.$
\end{stat}
\begin{proof}
Indeed, it is clear that
$$
(x\langle y\rangle=1)\equiv(\exists z)_{z\leq x}(\exists
t)_{t<z}(\exists u)_{u\leq x}((x=2uz+z+t)\&(z=2^y)).
$$
From ~\cite{erf} it is known that $(x=2uz+z+t)$ and $(z=2^y)$ are in
$\ccr.$ From this it follows that $(x\langle
y\rangle=1)\in\ccr\subseteq\ccclh_*.$ From this and from the fact that
$x\langle y\rangle$ takes values $0$ and $1,$ it follows that
is the statement that one wants to prove.
\end{proof}
\begin{stat}\label{statogrmin}The class $\ccclh$ is closed with respect to bounded minimization.
\end{stat}
\begin{proof}See ~\cite{erf}.\end{proof}

\begin{stat}\label{statlenxs}If $f(\tilde{x})\in\ccxs,$ then
$\fflen(f(\tilde{x}))\in\ccclh.$\end{stat}
\begin{proof}
Let $p(\tilde{x})$ be a polynomial that is upper bounded (strictly)
the length of binary notation $f(\tilde{x})$ (the existence of such
polynomial follows from the definition of $\ccxs$). Then it is
obvious that
$$
\fflen(f(\tilde{x}))=(\mu z)_{z<p(\tilde{x})}(f(\tilde{x})<2^z).
$$
From ~\cite{erf} it is known that
$$
(x<2^y)\in\ccclh_*.
$$
From this and from the statement \ref{statogrmin} follows the
statement that one is proving.
\end{proof}

\begin{stat}\label{stattermxs}Let $g_1(\tilde{z}),\ldots,g_k(\tilde{z})\in\ccxs,$
$t$ be a $\ccfom$-term over variables $x_1,\ldots,x_m,$ it corresponds to a function $h_t(X,x_1,\ldots,x_m).$ Then
$$
h_t(\ffcode(g_1(\tilde{z}),\ldots,g_k(\tilde{z})),x_1,\ldots,x_m)\in\ccclh.$$
\end{stat}
\begin{proof}
Everywhere defined
$$
h_t(\ffcode(g_1(\tilde{z}),\ldots,g_k(\tilde{z})),x_1,\ldots,x_m)
$$
follows from the fact that the domain of $\ffcode$ does not contain
an empty word. One can prove that the needed function belongs to
$\ccclh$. There can be the following cases.
\begin{enumerate}
\item $t$ is $1.$ Then it is obvious that
$$
h_t(\ffcode(g_1(\tilde{z}),\ldots,g_k(\tilde{z})),x_1,\ldots,x_m)=1.
$$
The $\ccclh$ affiliation is obvious. \item $t$ is $\lln.$ Then
$$
h_t(\ffcode(g_1(\tilde{z}),\ldots,g_k(\tilde{z})),x_1,\ldots,x_m)=$$
$$=2\cdot(\fflen(g_1(\tilde{z}))+\ldots+ \fflen(g_k(\tilde{z}))+k+1)
$$
(see the definitions of $\ffcode$ and $h_t$). Affiliation to the
class $\ccclh$ follows from the statement \ref{statlenxs}. \item $t$
looks like $x_i.$ Then
$$
h_t(\ffcode(g_1(\tilde{z}),\ldots,g_k(\tilde{z})),x_1,\ldots,x_m)=x_i.
$$
Affiliation to the class $\ccclh$ is obvious.
\end{enumerate}
The statement is proved.
\end{proof}
\begin{stat}\label{statelemformulaxs}Let $g_1(\tilde{z}),\ldots,g_k(\tilde{z})\in\ccxs,$
$\Phi$ is an elementary $\ccfom$-formula over the variables
$x_1,\ldots,x_m,$ it has a corresponding predicate
$\rho_{\Phi}(X,x_1,\ldots,x_m).$ Then
$$
\rho_{\Phi}(\ffcode(g_1(\tilde{z}),\ldots,g_k(\tilde{z})),x_1,\ldots,x_m)\in\ccclh_*.$$
\end{stat}
\begin{proof}
Everywhere defined predicate
$$\rho_{\Phi}(\ffcode(g_1(\tilde{z}),\ldots,g_k(\tilde{z})),x_1,\ldots,x_m)$$
follows from the fact that the domain of $\ffcode$ doesn not contain
an empty word. One can prove the affiliation to the class
$\ccclh_*.$ There can be different cases.
\begin{enumerate}
\item $\Phi$ is of the type $(t_1=t_2)$. Then
$$
\rho_{\Phi}\equiv(h_{t_1}=h_{t_2}).
$$
Affiliation to the class $\ccclh_*$ follows from the statement
\ref{stattermxs} and from the fact that  $(x=y)\in\ccr\subseteq\ccclh_*$
(see ~\cite{erf}). \item $\Phi$ look like $(t_1\leq t_2)$.
Analogously. \item $\Phi$ looks like $\llbit(t_1,t_2).$ From the statement
 \ref{statangles} it follows that $(x\langle
y\rangle=1)\in\ccclh_*.$ From the definition of $\rho_{\Phi}$
follows the representation
$$
\rho_{\Phi}(\ffcode(g_1(\tilde{z}),\ldots,g_k(\tilde{z})),x_1,\ldots,x_m)\equiv$$
$$
\equiv(h_{t_1}(\ffcode(g_1(\tilde{z}),\ldots,g_k(\tilde{z})),\tilde{x})$$
$$\langle
h_{t_2}(\ffcode(g_1(\tilde{z}),\ldots,g_k(\tilde{z})),\tilde{x})-1\rangle=1).
$$
From here, from the statement \ref{stattermxs} and from the fact that $(x\langle
y\rangle=1)\in\ccclh_*,$ it follows that
$$
\rho_{\Phi}(\ffcode(g_1(\tilde{z}),\ldots,g_k(\tilde{z})),x_1,\ldots,x_m)\in\ccclh_*.
$$
\item $\Phi$ looks like $\llx\langle t_1\rangle.$ Briefly
$h_{t_1}(\ffcode(g_1(\tilde{z}),\ldots,g_k(\tilde{z})),x_1,\ldots,x_m)$
and
$\rho_{\Phi}(\ffcode(g_1(\tilde{z}),\ldots,g_k(\tilde{z})),x_1,\ldots,x_m)$
бwill be denoted simply as $h_{t_1}$ and $\rho_{\Phi},$ while
expressions $2\cdot\fflen(g_i(\tilde{z}))$ will be denoted $l_i$
($1\leq i\leq k$). Then from definitions of $\ffcode$ and
$\rho_{\Phi}$ it follows that
$$
\rho_{\Phi}\equiv\begin{cases}
(g_{i}(\tilde{z})\left\langle\left[\frac{h_{t_1}\dotminus(2i+l_1+\ldots+l_{i-1}+1)}{2}\right]\right\rangle=1),\text{ if }\\
\quad\quad 2i+l_1+\ldots+l_{i-1}+1\leq
h_{t_1}< \\ \quad\quad <2i+l_1+\ldots+l_{i-1}+l_i+1,\ 1\leq i\leq k, \\
\text{true},\text{ if }h_{t_1}=2i+l_1+\ldots+l_i+2,\ 0\leq i\leq k, \\
\text{false}\text{ else}. \\
\end{cases}
$$
Based on induction proposal, $h_{t_1}\in\ccclh.$ From here, from
inclusions $g_i\in\ccxs,$ from the definition of $\ccxs$ and from
the simplest features of the class $\ccclh$ (see ~\cite{erf}) it
follows that for any $i$ ($1\leq i\leq k$)
$$
g_{i}(\tilde{z})\left\langle
\left[\frac{h_{t_1}\dotminus(2i+l_1+\ldots+l_{i-1}+1)}{2}\right]\right\rangle\in\ccclh.
$$
Besides, according to the statement \ref{statlenxs}, $l_i\in\ccclh$
($1\leq i\leq k$). From here and from the closeness of $\ccclh$
relative to breaking down cases with respect to predicates from
$\ccclh_*$ (see ~\cite{erf}) it follows that $\rho_{\Phi}\in\ccclh.$
\end{enumerate}
The statement is proved.
\end{proof}
\begin{stat}\label{statformulaxs}Let $g_1(\tilde{z}),\ldots,g_k(\tilde{z})\in\ccxs,$
$\Phi$ is a $\ccfom$-formula over variables $x_1,\ldots,x_m,$ to which there is a corresponding predicate $\rho_{\Phi}(X,x_1,\ldots,x_m).$ Then
$$
\rho_{\Phi}(\ffcode(g_1(\tilde{z}),\ldots,g_k(\tilde{z})),x_1,\ldots,x_m)\in\ccclh_*.$$
\end{stat}
\begin{proof}
Everywhere defined predicate
$\rho_{\Phi}(\ffcode(g_1(\tilde{z}),\ldots,g_k(\tilde{z})),x_1,\ldots,x_m)$
it follows from the fact that the domain of $\ffcode$ does not
contain an empty word. Let $l$ be a contracted notation for
$$
2\cdot(\fflen(g_1(\tilde{z}))+\ldots+\fflen(g_k(\tilde{z}))+k+1).
$$
From the statement  \ref{statlenxs} it follows that $l\in\ccclh.$
Affiliation to the class $\ccclh_*$ can be proved by inducting on the construction of the formula.
\begin{enumerate}
\item $\Phi$ is an elementary $\ccfom$-formula. Then this formula follows from the statement \ref{statelemformulaxs}. \item
$\Phi$ is of the form $(\Phi_1\&\Phi_2),$ $(\Phi_1\vee\Phi_2)$ or
$(\neg\Phi_1).$ The statement follows from the close of $\ccclh_*$
with respect to logical operations (see ~\cite{erf}).
\item $\Phi$ is of the form $(\exists x_i)\Phi_1.$  Then
$$
\rho_{\Phi}(\ffcode(g_1(\tilde{z}),\ldots,g_k(\tilde{z})),x_1,\ldots,x_m)
\equiv$$ $$\equiv(\exists x)_{(1\leq x\leq
l)}\rho_{\Phi_1}(\ffcode(g_1(\tilde{z}),\ldots,g_k(\tilde{z})),x_1,\ldots,x_{i-1},x,x_{i+1},\ldots,x_m).
$$
Affiliation with the class $\ccclh_*$ follows from the fact that
$l\in\ccclh,$
 and from the closeness of $\ccclh_*$ with respect to bounded quantification (see ~\cite{erf}). \item
$\Phi$ is of the type $(\forall x_i)(\Phi_1).$ Then the given
statement is a consequence from items $2$ and $3.$ \item $\Phi$ is
of the form $(\qqm x_i)(\Phi_1).$ Let
$$
r(\tilde{z},x_1,\ldots,x_m)= $$ $$=\Sum_{1\leq x\leq
l}\chi(\ffcode(g_1(\tilde{z}),\ldots,g_k(\tilde{z})),x_1,\ldots,x_{i-1},x,x_{i+1},\ldots,x_m),
$$
where $\chi$ is the characteristic function of the predicate
$\rho_{\Phi_1}.$ It is obvious that $r(\tilde{z},x_1,\ldots,x_m)$ is
the number of $x$ such that  $1\leq x\leq l$ and
$$\rho_{\Phi_1}(\ffcode(g_1(\tilde{z}),\ldots,g_k(\tilde{z})),x_1,\ldots,x_{i-1},x,x_{i+1},\ldots,x_m)$$
holds true. From the definition of $\ccclh,$ the induction step and
the fact that $l\in\ccclh,$ it follows that $r\in\ccclh.$ From the
definition of $\rho_{\Phi}$ it follows that
$$
\rho_{\Phi}(\ffcode(g_1(\tilde{z}),\ldots,g_k(\tilde{z})),x_1,\ldots,x_m)\equiv$$
$$\equiv(r(\tilde{z},x_1,\ldots,x_m)>l\dotminus
r(\tilde{z},x_1,\ldots,x_m)).
$$
Therefore from $(x\dotminus y)\in\ccclh$ and from $(x>y)\in\ccclh_*$ it follows that
 $\rho_{\Phi}\in\ccclh_*.$
\end{enumerate}
The statement is proved.
\end{proof}
\begin{stat}\label{statfomn}
Let $g_1(\tilde{z}),\ldots,g_k(\tilde{z})\in\ccxs,$
$\rho(y_1,\ldots,y_n)\in\ccfomn.$ Then the predicate
$$
\varphi(\tilde{z})=\rho(g_1(\tilde{z}),\ldots,g_k(\tilde{z}))
$$
is in $\ccclh_*.$
\end{stat}
\begin{proof}
From the definition of $\ccfomn$ it follows that there exists such
predicate $\psi(X)$ from $\ccfom$ that
$$
\rho(y_1,\ldots,y_n)\equiv\psi(\ffcode(y_1,\ldots,y_n)).
$$
From the definition of $\ccfom$ it follows that there exists
$\ccfom$-formula $\Phi,$ to which there is a corresponding predicate
$\rho_{\Phi}(X,x_1,\ldots,x_m)$ such that
$$
\rho_{\Phi}(X,x_1,\ldots,x_m)\equiv\psi(X).
$$
Thereby,
$$
\rho(g_1(\tilde{z}),\ldots,g_k(\tilde{z}))\equiv\rho_{\Phi}(\ffcode(g_1(\tilde{z}),\ldots,g_k(\tilde{z})),
x_1,\ldots,x_m).
$$
From the statement \ref{statformulaxs} it follows that
$$
\rho_{\Phi}(\ffcode(g_1(\tilde{z}),\ldots,g_k(\tilde{z})),
x_1,\ldots,x_m)\in\ccclh_*.
$$
The statement is proved.
\end{proof}
\begin{stat}\label{statffom}
Let $g_1(\tilde{z}),\ldots,g_k(\tilde{z})\in\ccxs,$
$f(y_1,\ldots,y_n)\in\ccffom.$ Then
$$
h(\tilde{z})=f(g_1(\tilde{z}),\ldots,g_k(\tilde{z}))\in\ccxs.
$$
\end{stat}
\begin{proof}
The boundedness $h(\tilde{z})$ by functions of the form
$2^{p(\tilde{z})},$ where $p$ is a polynomial, follows from
restrictions in definitions of $\ccxs$ and $\ccffom.$ One can prove
that $(h(\tilde{z})\langle t\rangle=1)\in\ccclh_*.$ Indeed, it is
obvious that
$$
(h(\tilde{z})\langle
t\rangle=1)\equiv\xi(g_1(\tilde{z}),\ldots,g_k(\tilde{z}),t),
$$
where
$$
\xi(y_1,\ldots,y_m,t)\equiv(f(y_1,\ldots,y_m)\langle t\rangle=1).
$$
From the definition of $\ccffom$ it follows that $\xi\in\ccfomn.$
From that and the statement \ref{statfomn} it follows that
$$
\xi(g_1(\tilde{z}),\ldots,g_k(\tilde{z}),t)\in\ccclh_*.
$$
And this is equivalent to the fact that
$$
h(\tilde{z})\langle t\rangle\in\ccclh
$$
(because $h(\tilde{z})\langle t\rangle$ takes only values $0$ and
$1$). From this and the definition of $\ccxs$ it follows that
 $h\in\ccxs.$ The statement is proved.
\end{proof}
\begin{stat}\label{statarythmffom}
Functions
$$
x+1,\quad x\dotminus y,\quad xy,\quad x\wedge y,\quad
\left[x/y\right],\quad x^{\fflen(y)}
$$
are in the class $\ccffom.$
\end{stat}
\begin{proof}
For the function $x\wedge y$ it obviously follows from equivalent
definitions of the class $\ccfom$ (for example through the boolean
circuits, see ~\cite{uniformity}). For the remaining functions the
proof is in ~\cite{arythm}.
\end{proof}
\begin{stat}\label{statexpsxs}
If $f(\tilde{x})\in\ccclh,$ then $2^{f(\tilde{x})}\dotminus
1\in\ccxs.$
\end{stat}
\begin{proof}
The verity of the upper bound on the speed of growth is clear. From
the simplistic features of binary notation of numbers it follows
that the following holds
$$
((2^{f(\tilde{x})}\dotminus 1)\langle y\rangle=1) \equiv
(y<f(\tilde{x})).
$$
It is obvious that the predicate $(x<y)$ is in $\ccclh_*$ (see
~\cite{erf}). Therefore
$$
(y<f(\tilde{x}))\in\ccclh_*.
$$
From this one sees the validity of the statement that one was proving.
\end{proof}
\begin{theorem}\label{statburjui}
The following inclusion takes place $\left[ \mains\right]_{x^y}\subseteq\ccxs.$
\end{theorem}
\begin{proof}
One can prove this statement by induction on constructing functions in the class
$[\mains]_{x^y}.$ Let $h\in[\mains]_{x^y}.$ Then there can be different cases.
\begin{enumerate}
\item $h\in\mains.$ Then obviously (see for example ~\cite{erf}) that
$h\in\ccclh.$ From $\ccclh\subseteq\ccxs$ it follows that $h\in\ccxs.$
\item $h$ is obtained from $f$ by permuting, identifying variables or
introducing dummy variables, $f\in\ccxs$. In this case the inclusion
$h\in\ccxs$ follows from the fact that the class $\ccclh$ is closed
with respect to superposition (specifically, permutation,
identification of variables, introduction of dummy variables).
\item
$h(\tilde{x})=f(g_1(\tilde{x}),\ldots,g_m(\tilde{x})),$ where
$f\in\mains,$ $g_1,\ldots,g_m\in\ccxs.$ In this case from statements
 \ref{statarythmffom} and \ref{statffom} follows that
$h\in\ccxs.$ \item $h=f^g,$ where $f\in\ccxs,$ $g\in[\mains].$ Then
this can be written as
$$
h=f^{\fflen(2^g\dotminus 1)}.
$$
It is obvious that $g\in\ccclh$ (see ~\cite{erf}), thus $2^g\dotminus
1\in\ccxs$ (the statement \ref{statexpsxs}). From this, from the statement
$x^{\fflen(y)}\in\ccffom$ (the statement \ref{statarythmffom}) and the statement \ref{statffom} it follows that $h\in\ccxs.$
\end{enumerate}
The theorem is proved.
\end{proof}

\subsection{The Classes $\cccr,$ $\cccrf,$ and $\mains$- polynomiality}

\begin{stat}\label{statcccrquant}
The class $\cccr$ is closed with respect to bounded quantifications
of the form $(\exists x)_{x<p(\tilde{y})}$ and $(\forall
x)_{x<p(\tilde{y})},$ where $p$ is a polynomial with coefficients
from $\nat$.
\end{stat}
\begin{proof}
Let
$$
\varphi(\tilde{y})\equiv(\exists
x)_{x<p(\tilde{y})}\psi(x,\tilde{y}),
$$
$\psi\in\cccr.$ Let $\rho(x,\tilde{y},z)$ is obtained from
$\psi(x,\tilde{y})$ with the help of the counting operation with
respect to the variable $x$ and the polynomial $p(\tilde{y}).$ Then
$\rho(x,\tilde{y},z)$ is true if and only if $z$ is a number of such
$t<p(\tilde{y})$ that $\psi(t,\tilde{y})$ is true. From this it
follows that $\rho(0,\tilde{y},0)$ is true if and only if there is
no such $t<p(\tilde{y})$ that $\psi(t,\tilde{y})$ is true. For this
it follows that for any $\tilde{y}$ it holds true that
$$
\varphi(\tilde{y})\equiv\neg\rho(0,\tilde{y},0).
$$
Since $\cccr$ is closed with respect to the counting operation, the
constant substitution operation and logical operations, one obtains
that $\varphi\in\cccr.$ The closeness of the class $\cccr$ with
respect to $(\forall x)_{x<p(\tilde{y})}$ follows from the closeness
of $\cccr$ with respect to $(\exists x)_{x<p(\tilde{y})}$ and
logical operations. The statement is proved.
\end{proof}
\begin{stat}\label{statccrincccr}
The following inclusion takes place $\ccr\subseteq\cccr.$ \end{stat}
\begin{proof}
Indeed, for this one needs to prove that $\cccr$ is closed with respect to quantification of the type $(\exists x)_{x<y}$ and $(\forall
x)_{x<y},$ and this follows from the statement \ref{statcccrquant}.
The statement is proved.
\end{proof}
\begin{stat}\label{statcccrfsubst}
The class $\cccrf$ is closed with respect to superposition.
\end{stat}
\begin{proof}
Closeness with respect to permutation, identification of variables
and introduction of dummy variables follows from the fact that
$\cccr$ is closed with respect to explicit transformations.

One can prove this closeness with respect to substitution of a function into function.
Let
$$
h(\tilde{x})=f(g_1(\tilde{x}),\ldots,g_m(\tilde{x})),
$$
$f,g_1,\ldots,g_m\in\cccrf.$ One can claim that $h\in\cccrf.$
Polynomial boundedness of the function $h$ follows from
the polynomial boundedness of $f,g_1,\ldots,g_m.$ Let
$p(\tilde{x})$ be a polynomial  that strictly upper-bounds functions $g_1(\tilde{x}),\ldots,g_m(\tilde{x}).$ Then
$$
(z=h(\tilde{x}))\equiv(\exists
y_1)_{y_1<p(\tilde{x})}\ldots(\exists y_m)_{y_m<p(\tilde{x})}$$
$$((y_1=g_1(\tilde{x}))\&\ldots\&(y_m=g_m(\tilde{x}))\&
(z=f(y_1,\ldots,y_m))).
$$
From this, from the closeness of $\cccr$ with respect to logical
operations and explicit transformations and from the statement
\ref{statcccrquant} it follows that
 $(z=h(\tilde{x}))\in\cccr.$ The statement is proved.
\end{proof}
\begin{stat}\label{statcccrfsum}
The class $\cccrf$ is closed with respect to the operation of narrowed bounded summation.
\end{stat}
\begin{proof}
Let
$$
f(x,z_1,\ldots,z_n)=\Sum_{y<x}\ffsg(g(y,z_1,\ldots,z_n)),
$$
$g\in\cccrf.$ One can prove that $f\in\cccrf.$ Let
$$
\rho(x,z_1,\ldots,z_n)\equiv\neg(0=g(x,z_1,\ldots,z_n)).
$$
It is obvious that $\rho$ is obtained from the graph of $g$ with the
help of logical operations and substitution of constants, thus
$\rho\in\cccr.$ Let $\varphi(x,z_1,\ldots,z_n,u)$ is obtained from
$\rho(x,z_1,\ldots,z_n)$ with the help of counting operations with
respect to the variable $x$ and the polynomial $x.$ Then
$\varphi(x,z_1,\ldots,z_n,u)$ holds treu if and only if $u$ is the
number of $y<x$ such that $\rho(y,z_1,\ldots,z_n)$ is true (i.e.
$\ffsg(g(y,z_1,\ldots,z_n))=1$). Thereby, for any
$x,z_1,\ldots,z_n,u$ it is true that
$$
(u=f(x,z_1,\ldots,z_n))\equiv\varphi(x,z_1,\ldots,z_n,u).
$$
From $\rho\in\cccr$ and the closeness of $\cccr$ with respect to explicit transformations and counting operations it follows that $\varphi\in\cccr,$
i.e. the graph of the function $f$ is in $\cccr.$ The polynomial boundedness
of $f$ obviously follows from the polynomial boundedness of $g.$ The statement is proved.
\end{proof}
\begin{stat}\label{statcccrfelem}
Functions $0,\ x+1,\ x\dotminus y,\ xy$ are in $\cccrf.$
\end{stat}
\begin{proof}
It is known that ~\cite{erf} the predicates
$$
x=0,\quad y=x+1,\quad z=x\dotminus y,\quad z=xy
$$
are in $\ccr.$ From this and the statement \ref{statccrincccr} it
follows the following statement.
\end{proof}
\begin{stat}\label{statscolemincccr}
It satisfies $\ccclh_*\subseteq\cccr.$
\end{stat}
\begin{proof}
From ~\cite{erf} it is known that $\ccclh$ coincides with the
minimal class of functions that contains functions $0,\ x+1,\
x\dotminus y,\ xy$ and closed with respect to the superposition and
the narrowed bounded summation. From this and the statements
\ref{statcccrfelem}, \ref{statcccrfsubst}, \ref{statcccrfsum} it
follows that
$$
\ccclh\subseteq\cccrf.
$$
From ~\cite{erf} is known that $\ccclh_*$ is the set of all graphs of functions from  $\ccclh.$ Thereby,
$$
\ccclh_*\subseteq\cccr.
$$
The statement is proved.
\end{proof}

It is easy to see that the following five statements hold true.
\begin{stat}If the function $f$ is $\mains$-polynomial with respect to some set of variables, then $f\in[\mains]_{2^x}.$
\end{stat}
\begin{stat}If the function $f$ is $\mains$-polynomial with respect to
the set of variables $X$ and $Y\subseteq X,$ then $f$ is
$\mains$-polynomial with respect to variables $Y.$\end{stat}
\begin{stat}\label{spoly}
If the function $f(x_1,\ldots,x_n)$ is $\mains$-polynomial with
respect to the set of variables  $\{x_{i_1},\ldots,x_{i_k}\},$ the
function $g(y_1,\ldots,y_m)$ is $\mains$-polynomial with respect to
the set of variables $\{y_{j_1},\ldots,y_{j_p}\},$ then
$$
f(x_1,\ldots,x_{i_1-1},g(y_1,\ldots,y_m),x_{i_1+1},\ldots,x_n)
$$
is $\mains$-polynomial with respect to the set of variables
$\{x_{i_2},\ldots,x_{i_k},y_{j_1},\ldots,y_{j_p}\}.$
\end{stat}
\begin{stat}\label{spolybad}
Let the function $f(x_1,\ldots,x_n)$ be $\mains$-polynomial with
respect to the set of variables $\{x_{i_1},\ldots,x_{i_k}\},$ $1\leq
i\leq n,$ $i\notin\{i_1,\ldots,i_k\}$ and
$g(y_1,\ldots,y_m)\in\left[\mains\right].$ Then
$$
f(x_1,\ldots,x_{i-1},g(y_1,\ldots,y_m),x_{i+1},\ldots,x_n)
$$
is $\mains$-polynomial with respect to the set
$\{x_{i_1},\ldots,x_{i_k}\}.$
\end{stat}
\begin{stat}\label{spolyswap}
Let the function $g(y_1,\ldots,y_m)$ $\mains$ be a
$\mains$-polynomial one with respect to the set of variables
$\{y_{j_1},\ldots,y_{j_n}\},$
$$
f(x_1,\ldots,x_k)=g(x_{i_1},\ldots,x_{i_m}).
$$
Then $f(x_1,\ldots,x_k)$ is $\mains$-polynomial with respect to the
set of  all variables $x_i,$ for which the set of all $y_j$ such
that $i_j=i,$ is in $\{y_{j_1},\ldots,y_{j_n}\}.$
\end{stat}

\begin{stat}\label{spolyfinite}
Let $f(x_1,\ldots,x_n)$ be $\mains$-polynomial with respect to the
set of variables $X,$ $g(x_1,\ldots,x_n)$ differs from $f$ in finite
number of points. Then $g(x_1,\ldots,x_n)$ is $\mains$-polynomial
with respect to $X.$
\end{stat}
\begin{proof}
It is obvious that one can just prove this statement for one point.
Let
$$
f(a_1,\ldots,a_n)=b,\ g(a_1,\ldots,a_n)=c,
$$
in other points $f$ and $g$ coincide. Then if $b<c,$ then
$$
g(x_1,\ldots,x_n)=$$
$$=f(x_1,\ldots,x_n)+(c-b)\cdot(1\dotminus((x_1\dotminus
a_1)+(a_1\dotminus x_1)))
\cdot\ldots\cdot(1\dotminus((x_n\dotminus a_n)+(a_n\dotminus
x_n))).
$$
An analogous formula holds true for $b>c.$ From these formulas and from statements
 \ref{spoly}, \ref{spolybad}, \ref{spolyswap} it follows that
 $g(x_1,\ldots,x_n)$ is $\mains$-polynomial with respect to $X.$
\end{proof}

\subsection{The Inclusion $\ccxs\subseteq\left[ \mains\right]_{2^x}$}

\begin{stat}\label{statrm}
The function $\ffrm(x,y)$ is $\mains$-polynomial over the set of
variables $\{x,y\}.$
\end{stat}
\begin{proof}
One has $\ffrm(x,y)=x\dotminus[x/y]\cdot y.$ Thus,
$\ffrm(x,y)\in[\mains].$ From this it follows $\mains$-polynomiality
of the function $\ffrm.$
\end{proof}

Let
$$
\langle x_0,\ldots,x_{n-1};\ l\rangle=\Sum_{i=0}^{n-1}x_i2^{il}.
$$
One can notice that if the condition $x_0,x_1,\ldots,x_{n-1}<2^l$ is
satisfied then for any $i$ ($0\leq i<n$) the binary digits of the
number $\langle x_0,x_1,\ldots,x_{n-1};\ l\rangle$ from $(il)$-th up
to $(il+l-1)$-th make up the binary notation of the number $x_i.$

Let
$$
\ffrep(x,n,l)=x\cdot\left[\frac{2^{nl}\dotminus 1}{2^l\dotminus
1}\right]\footnote{In this dissertation, it was convenient to define
functions through formulas rather than by conditions to which these
formulas satisfy. Thus, for the reader it might be easier to read
what follows the formula with statements on features of functions
and the proofs of those statements rather than trying to parse the
formula straight away.}.
$$
\begin{stat}\label{statrep}
If $n,l\geq 1,$ then
$$
\ffrep(x,n,l)=\langle\underbrace{x,x,\ldots,x}_{n\text{ times }};\
l\rangle.
$$
Besides, $\ffrep(x,n,l)$ is $\mains$-polynomial with respect to
$\{x\}.$
\end{stat}
\begin{proof}
By using the formula for geometric progression sum, one obtains
$$
\ffrep(x,n,l)=x\cdot\Sum_{i=0}^{n-1}2^{li}=\Sum_{i=0}^{n-1}x2^{li}=\langle\underbrace{x,x,\ldots,x}_{n\text{
times}};\ l\rangle.
$$
$\mains$-polynomiality with respect to $\{x\}$ follows from the form
of the formula ($x$ is excluded from power exponents). The statement
is proved.
\end{proof}
Let
$$
\ffincrx(x,n,l_1,l_2)=\ffrep(x,n,l_2\dotminus
l_1)\wedge\ffrep(2^{l_1}\dotminus 1,n,l_2).
$$
\begin{stat}\label{statincrx}
If $n,l_1\geq 1,$ $l_2\geq(n+1)l_1,$ $x=\langle
x_0,\ldots,x_{n-1};\ l_1\rangle,$ $0\leq
x_0,\ldots,x_{n-1}<2^{l_1},$ then
$$
\ffincrx(x,n,l_1,l_2)=\langle x_0,\ldots,x_{n-1};\ l_2\rangle.
$$
Besides, $\ffincrx(x,n,l_1,l_2)$ is $\mains$-polynomial with respect
to $\{x\}.$
\end{stat}
\begin{proof}
One has
$$
x=\Sum_{i=0}^{n-1}x_i2^{il_1}\leq\Sum_{i=0}^{n-1}(2^{l_1}-1)2^{il_1}<2^{nl_1}.
$$
From $l_2\geq(n+1)l_1$ it follows that
$$
l_2-l_1\geq nl_1.
$$
From this it follows that the binary digits of the number
$\ffrep(x,n,l_2\dotminus l_1)$ from $i(l_2-l_1)$-th up to
$(i(l_2-l_1)+l_2-l_1-1)$-th generate the binary notation of the
number $x$ ($0\leq i\leq n-1$). Besides, binary digits of the number
$x$ from $(il_1)$-th up to $(il_1+l_1-1)$-th make up binary notation
of the number $x_i$ ($0\leq i\leq n-1$). From this it follows that
the binary notation of the number
 $\ffrep(x,n,l_2\dotminus l_1)$ from $(il_2)$-th up to
$(il_2+l_1-1)$-th make up binary notation of the number $x_i$
($0\leq i\leq n-1$).

One can note that the binary notation of the number
$\ffrep(2^{l_1}\dotminus 1,n,l_2)$ is $n$ blocks of ones, besides
$i$-th block ($0\leq i\leq n-1$) takes up digits from $(il_2)$-th up
to $(il_2+l_1-1)$-th. Thereby, one obtains that
$$
\ffrep(x,n,l_2\dotminus l_1)\wedge\ffrep(2^{l_1}\dotminus
1,n,l_2)=\langle x_0,\ldots x_{n-1};\ l_2\rangle.
$$
$\mains$-polynomiality of $\ffincrx(x,n,l_1,l_2)$ with respect to
$\{x\}$ follows from the form of the formula, from $\mains$-
polynomiality of $\ffrep(x,n,l)$ with respect to $\{x\}$ and from
statements \ref{spoly}, \ref{spolyswap}. The statement is proved.
\end{proof}

One can define families of functions $\ffsp_n,$ $\ffsa_n$ ($n\geq
1$) in the following way:
$$
\ffsp_n(q,m_1,\ldots,m_n)=q^{2n}+q\cdot(m_1+\ldots+m_n+1),
$$
$$
\ffsa_n(q,k_1,\ldots,k_n,m_1,\ldots,m_n)=q\cdot\ffsp_n(q,m_1,\ldots,m_n)\cdot(k_1+\ldots+k_n+1).
$$
Further for brevity we will replace the expressions
$\ffsp_n(q,m_1,\ldots,m_n)$ and
$\ffsa_n(q,k_1,\ldots,k_n,m_1,\ldots,m_n)$ with $\ffsp$ and $\ffsa$
respectively. One can define the family of functions $\ffswap_n$
($n\geq 1$) in the following way:
$$
\ffswap_n(x,q,k_1,\ldots,k_n,m_1,\ldots,m_n)=$$
$$=\ffrm\left(\left[\frac{\ffincrx(x,q^n,1,\ffsp)\cdot\prod_{r=1}^{n}\left[
\frac{2^{\ffsa+k_r\ffsp}\dotminus 2^{\ffsa+k_r\ffsp\dotminus
q(k_r\ffsp\dotminus m_r)}}{2^{k_r\ffsp}\dotminus 2^{m_r}}
\right]}{2^{n\cdot\ffsa}}\right],2^{\ffsp}\right).
$$
\begin{stat}\label{statswap}
Let $n,q\geq 1,$ $f(i_1,\ldots,i_n)$ be some function that takes up values
 $0$ and $1.$ Besides, let numbers
$k_1,\ldots,k_n\geq 1$ be such that for any different vectors
$(i_1',\ldots,i_n')$ and $(i_1'',\ldots,i_n'')$ $(0\leq
i_1',\ldots,i_n',i_1'',\ldots,i_n''<q)$ the following inequality
holds
$$
k_1i_1'+\ldots+k_ni_n'\neq k_1i_1''+\ldots+k_ni_n''.
$$
Then if \begin{equation}\label{swapx} x=\Sum_{0\leq
i_1,\ldots,i_n<q}f(i_1,\ldots,i_n)2^{k_1i_1+\ldots+k_ni_n},
\end{equation}
then for any $m_1,\ldots,m_n$ it is true that
$$
\ffswap_n(x,q,k_1,\ldots,k_n,m_1,\ldots,m_n)=\Sum_{0\leq
i_1,\ldots,i_n<q}f(i_1,\ldots,i_n)2^{m_1i_1+\ldots+m_ni_n}.
$$
Besides, $\ffswap_n(x,q,k_1,\ldots,k_n,m_1,\ldots,m_n)$ is
$\mains$-polynomial with respect to $\{x\}$.
\end{stat}
\begin{proof}
From the definition of $\ffsp$ and from the fact that $k_r,q\geq 1,$
it follows that for any $r$ ($1\leq r\leq n$) it holds true that
$k_r\ffsp\geq m_r.$ Besides, from the definition of $\ffsa$ it
follows that for any $r$ ($1\leq r\leq n$) it is true that
$\ffsa\geq qk_r\ffsp.$ From these two inequalities it follows that
$$
\left[ \frac{2^{\ffsa+k_r\ffsp}\dotminus
2^{\ffsa+k_r\ffsp\dotminus q(k_r\ffsp\dotminus
m_r)}}{2^{k_r\ffsp}\dotminus 2^{m_r}} \right] = \left[
\frac{2^{\ffsa+k_r\ffsp}-2^{\ffsa+k_r\ffsp- q(k_r\ffsp-
m_r)}}{2^{k_r\ffsp}- 2^{m_r}} \right]
$$
By using the formula for the sum of a geometric sequence, one obtains
$$
\left[ \frac{2^{\ffsa+k_r\ffsp}\dotminus
2^{\ffsa+k_r\ffsp\dotminus q(k_r\ffsp\dotminus
m_r)}}{2^{k_r\ffsp}\dotminus 2^{m_r}}
\right]=\Sum_{j=0}^{q-1}2^{\ffsa+j(m_r-k_r\ffsp)}.
$$
Thereby, one has
$$
\Prod_{r=1}^{n}\left[ \frac{2^{\ffsa+k_r\ffsp}\dotminus
2^{\ffsa+k_r\ffsp\dotminus q(k_r\ffsp\dotminus
m_r)}}{2^{k_r\ffsp}\dotminus 2^{m_r}} \right] =
\Prod_{r=1}^{n}\Sum_{j=0}^{q-1}2^{\ffsa+j(m_r-k_r\ffsp)}=
$$
$$
=\Sum_{0\leq
j_1,\ldots,j_n<q}2^{n\ffsa+j_1(m_1-k_1\ffsp)+\ldots+j_n(m_n-k_n\ffsp)}.
$$
From the definition of $\ffsp$ and from $q\geq 1,$ it follows that
$\ffsp\geq q^n+1.$ From this and from the fact that $f$ only takes
up values from $\{0,1\},$ and from (\ref{swapx}) and the statement
\ref{statincrx} it follows that
$$
\ffincrx(x,q^n,1,\ffsp)=\Sum_{0\leq
i_1,\ldots,i_n<q}f(i_1,\ldots,i_n)2^{\ffsp\cdot(k_1i_1+\ldots+k_ni_n)}.
$$
Thereby,
$$
\ffincrx(x,q^n,1,\ffsp)\cdot\Prod_{r=1}^{n}\left[
\frac{2^{\ffsa+k_r\ffsp}\dotminus 2^{\ffsa+k_r\ffsp\dotminus
q(k_r\ffsp\dotminus m_r)}}{2^{k_r\ffsp}\dotminus 2^{m_r}} \right]=
$$
$$
= \left(\Sum_{0\leq
i_1,\ldots,i_n<q}f(i_1,\ldots,i_n)2^{\ffsp\cdot(k_1i_1+\ldots+k_ni_n)}\right)
\cdot$$ $$\cdot\left(\Sum_{0\leq
j_1,\ldots,j_n<q}2^{n\ffsa+j_1(m_1-k_1\ffsp)+\ldots+j_n(m_n-k_n\ffsp)}\right)=
$$
$$
=\Sum_{\substack{0\leq i_1,\ldots,i_n<q \\ 0\leq
j_1,\ldots,j_n<q}}f(i_1,\ldots,i_n)2^{n\ffsa+j_1m_1+\ldots+j_nm_n+\ffsp\cdot(k_1(i_1-j_1)+\ldots
+k_n(i_n-j_n))}.
$$
Let one divide all parts of this sum into three groups.
\begin{enumerate}
\item The terms for which
$k_1(i_1-j_1)+\ldots+k_n(i_n-j_n)\leq-1$. The sum of these terms one
defines as $A.$ It is clear that for such terms the following
inequalities hold true
$$
n\ffsa+j_1m_1+\ldots+j_nm_n+\ffsp\cdot(k_1(i_1-j_1)+\ldots
+k_n(i_n-j_n))\leq $$
$$\leq n\ffsa+q\cdot(m_1+\ldots+m_n)-\ffsp\leq n\ffsa-q^{2n}.
$$
The last inequality follows from the definition of $\ffsp.$ From these inequalities it follows that every term of this type is not bigger then $2^{n\ffsa-q^{2n}}$. From this and from the fact that the total number of terms equals $q^{2n},$ one can conclude that
$$
A<2^{n\ffsa}.
$$
\item The terms for which $k_1(i_1-j_1)+\ldots+k_n(i_n-j_n)\geq
1.$ Let the sum of these terms equal to $B.$ It is plain that for
such terms
$$
n\ffsa+j_1m_1+\ldots+j_nm_n+\ffsp\cdot(k_1(i_1-j_1)+\ldots
+k_n(i_n-j_n))\geq n\ffsa+\ffsp.
$$
Thus, each of such terms is divided by $2^{n\ffsa+\ffsp}.$
Thereby, one has
$$
B=2^{n\ffsa+\ffsp}B_0,
$$
where $B_0\in\nat$. \item The terms for which
$k_1(i_1-j_1)+\ldots+k_n(i_n-j_n)=0,$ i.e. $$
k_1i_1+\ldots+k_ni_n=k_1j_1+\ldots+k_nj_n. $$ Let the sum of these terms equal $C.$ As required, in this case $i_r=j_r$ for any $r=1,2,\ldots,n.$ Since for every vector
$(j_1,\ldots,j_n)$ there exists only one vector
$(i_1,\ldots,i_n),$ for which this condition is satisfied, it holds that
$$
C=\Sum_{0\leq
j_1,\ldots,j_n<q}f(j_1,\ldots,j_n)2^{n\ffsa+j_1m_1+\ldots+j_nm_n}=2^{n\ffsa}\cdot
C_0,
$$
where
$$
C_0=\Sum_{0\leq
j_1,\ldots,j_n<q}f(j_1,\ldots,j_n)2^{j_1m_1+\ldots+j_nm_n}.
$$
The number of terms in a sum is $q^n,$ each of which is no bigger
than $2^{q\cdot(m_1+\ldots+m_n)}.$ From this and from the definition
of $\ffsp$ it follows that
$$
C_0\leq q^n\cdot 2^{q\cdot(m_1+\ldots+m_n)}<2^{q^{2n}}\cdot
2^{q\cdot(m_1+\ldots+m_n)}\leq 2^{\ffsp}.
$$
\end{enumerate}
Thereby, one has
$$
\ffswap_n(x,q,k_1,\ldots,k_n,m_1,\ldots,m_n)=
\ffrm\left(\left[\frac{A+2^{n\ffsa+\ffsp}B_0+2^{n\ffsa}C_0}{2^{n\ffsa}}\right],2^{\ffsp}\right)=
C_0.
$$
The last equality follows from the fact that $A<2^{n\ffsa}$ and
$C_0<2^{\ffsp}.$ $\mains$- polynomiality of
$\ffswap_n(x,q,k_1,\ldots,k_n,m_1,\ldots,m_n)$ with respect to
$\{x\}$ follows from the make up of the formula and the statements
\ref{spoly}, \ref{spolyswap}, \ref{statincrx}, \ref{statrm}.
\end{proof}

Let
$$
\ffincr(x,q,l)=\ffswap_1(x,q,1,l),
$$
$$
\ffdecr(x,q,l)=\ffswap_1(x,q,l,1).
$$
\begin{stat}\label{statincr}
Let $q,l\geq 1,$ $0\leq x_0,\ldots,x_{q-1}\leq 1,$
$$
x=\langle x_0,\ldots,x_{q-1};\ 1\rangle.
$$
Then
$$
\ffincr(x,q,l)=\langle x_0,\ldots,x_{q-1};\ l\rangle.
$$
Besides, $\ffincr(x,q,l)$ is $\mains$-polynomial with respect to the
variable $x.$
\end{stat}
\begin{proof}
Let $k_1=1,$ $m_1=l,$ $n=1,$ $f(i)=x_i,$ if $i<q,$ $f(i)=0$
otherwise. Then for the numbers $q,n,k_1,m_1,x$ and the function $f$
all conditions of the statement \ref{statswap} are satisfied. Thus,
$$
\ffswap_1(x,q,1,l)=\Sum_{i=0}^{q-1}2^{lx_i}=\langle
x_0,\ldots,x_{n-1};\ l\rangle.
$$
$\mains$-polynomiality with respect to $\{x\}$ follows from the statememts
\ref{statincr} and \ref{spoly}. The statement is proved.
\end{proof}
\begin{stat}\label{statdecr}
Let $q,l\geq 1,$ $0\leq x_0,\ldots,x_{q-1}\leq 1,$
$$
x=\langle x_0,\ldots,x_{q-1};\ l\rangle.
$$
Then
$$
\ffdecr(x,q,l)=\langle x_0,\ldots,x_{q-1};\ 1\rangle.
$$
Besides, $\ffdecr(x,q,l)$ is $\mains$-polynomial with respect to the
set of variables $\{x\}.$
\end{stat}
\emph{The proof} is completely analogous to the proof of the
statement \ref{statincr}.

Let
$$
\ffnot(x,n)=(2^n\dotminus 1)\dotminus x,
$$
$$
\ffor(x,y,n)=\ffnot(\ffnot(x,n)\wedge\ffnot(y,n),n),
$$
$$
\ffxor(x,y,n)=\ffor(x,y,n)\wedge\ffnot(x\wedge y,n).
$$
%%% ----------------------------------------- Утвержденandе для not, or, xor --------------------------------------
\begin{stat}\label{statlogicfunc}
Let $n\geq 1,$
$$x=\langle x_0,\ldots,x_{n-1};\ 1\rangle,\ y=\langle y_0,\ldots,y_{n-1};\ 1\rangle,$$
$$0\leq x_0,\ldots,x_{n-1},y_0,\ldots,y_{n-1}\leq 1.$$
Then
$$
\ffnot(x,n)=\langle 1-x_0,\ldots,1-x_{n-1};\ 1\rangle,
$$
$$
\ffor(x,y,n)=\langle
x_0+y_0-x_0y_0,\ldots,x_{n-1}+y_{n-1}-x_{n-1}y_{n-1};\ 1\rangle,
$$
$$
\ffxor(x,y,n)=\langle
\ffrm(x_0+y_0,2),\ldots,\ffrm(x_{n-1}+y_{n-1},2);\ n\rangle.
$$
Besides, $\ffnot(x,n)$ is $\mains$-polynomial with respect to
$\{x\},$ $\ffor(x,y,n)$ and $\ffxor(x,y,n)$ are $\mains$-polynomial
with respect to $\{x,y\}.$
\end{stat}
\begin{proof}
One has
$$
\ffnot(x,n)=\Sum_{i=0}^{n-1}2^i\dotminus\Sum_{i=0}^{n-1}x_i2^i=
\Sum_{i=0}^{n-1}(1-x_i)2^i=\langle 1-x_0,\ldots,1-x_{n-1};\
1\rangle.
$$
The statements for $\ffor$ and $\ffxor$ follows from the fact that the corresponding equalities of the algebraic logic: $$
\alpha\vee\beta=\neg(\neg\alpha\wedge\neg\beta),
$$
$$
\alpha\oplus\beta=(\alpha\vee\beta)\wedge\neg(\alpha\wedge\beta).
$$
$\mains$-polynomiality follows from the definition and statements
\ref{spoly}, \ref{spolybad}, \ref{spolyswap}. The statement is proved.
\end{proof}
\begin{stat}\label{logicright}
The set of all correct predicates is closed with respect to operations of propositional logic.
\end{stat}
\begin{proof}
Let $\rho(x_1,\ldots,x_n),$ $\varphi(x_1,\ldots,x_n)$ be correct predicates and $f_{\rho},f_{\varphi}$ be their generating functions
$$
\psi_1(x_1,\ldots,x_n)\equiv\neg\rho(x_1,\ldots,x_n),
$$
$$
\psi_2(x_1,\ldots,x_n)\equiv\rho(x_1,\ldots,x_n)\&\varphi(x_1,\ldots,x_n),
$$
$f_{\psi_1},f_{\psi_2}$ are generating functions of the predicates
$\psi_1,\psi_2$ respectively. From the statement \ref{statlogicfunc}
and the definition of generating function, it follows that for any $x\geq
1$ it satisfies
$$
f_{\psi_1}(x)=\ffnot(f_{\rho}(x),x^n),\quad\quad
f_{\psi_2}(x)=f_{\rho}(x)\wedge f_{\varphi}(x).
$$
From this and the statements \ref{statlogicfunc}, \ref{spoly},
\ref{spolyfinite} it follows that
$$
f_{\psi_1},f_{\psi_2}\in\left[\mains\right]_{2^x}.
$$
The statement is proved.
\end{proof}
%%% ----------------------------------- Формула для cmp and cmpeq -----------------------------------------------------
Let
\begin{normalsize}
$$
\ffcmp(x,y,n,l)=$$
$$=\ffdecr\left(\left[\frac{((\ffrep(2^{2l\dotminus
1},n,2l)+\ffincr(x,nl,2))\dotminus\ffincr(y,nl,2))\wedge
\ffrep(2^{2l\dotminus 1},n,2l)}{2^{2l\dotminus
1}}\right],n,2l\right).
$$
\end{normalsize}
\begin{stat}\label{statcmp}
Let $n,l\geq 1,$
$$
x=\langle x_0,\ldots,x_{n-1};\ l\rangle,\ y=\langle
y_0,\ldots,y_{n-1};\ l\rangle,
$$
$$
0\leq x_0,\ldots,x_{n-1},y_0,\ldots,y_{n-1}<2^l.
$$
Then
$$
\ffcmp(x,y,n,l)=\langle \sigma_0,\ldots,\sigma_{n-1};\ 1\rangle,
$$
where
$$
\sigma_i=\begin{cases}1,\text{ if }x_i\geq y_i,\\ 0\text{
otherwise.}\end{cases}
$$
Besides, $\ffcmp(x,y,n,l)$ is $\mains$-polynomial with respect to
$\{x,y\}.$
\end{stat}
\begin{proof}
Let $x_{i,j}$ signify $j$-th binary digit of the number $x_i,$
$y_{i,j}$ is $j$-th digit of $y_i.$ Then for any $i$ ($0\leq i\leq
n-1$) it satisfies
$$
x_i=\langle x_{i,0},\ldots,x_{i,l-1};\ 1\rangle,\quad y_i=\langle
y_{i,0},\ldots,y_{i,l-1};\ 1\rangle.
$$
Besides,
$$
x=\langle
x_{0,0},\ldots,x_{0,l-1},x_{1,0},\ldots,x_{1,l-1},\ldots,x_{n,0},\ldots,x_{n,l-1};\
1\rangle,
$$
$$
y=\langle
y_{0,0},\ldots,y_{0,l-1},y_{1,0},\ldots,y_{1,l-1},\ldots,y_{n,0},\ldots,y_{n,l-1};\
1\rangle.
$$
From the statement \ref{statincr} and simple properties of numbers it follows that
$$
\ffincr(x,nl,2)=\langle x'_0,\ldots,x'_{n-1};\ 2l\rangle,\quad
\ffincr(y,nl,2)=\langle y'_0,\ldots,y'_{n-1};\ 2l\rangle,
$$
where for all $i$ ($0\leq i\leq n-1$)
$$
x'_i=\langle x_{i,0},\ldots,x_{i,l-1};\ 2\rangle,\quad
y'_i=\langle y_{i,0},\ldots,y_{i,l-1};\ 2\rangle.
$$
From this and the statement \ref{statrep} it follows that
$$(\ffrep(2^{2l\dotminus
1},n,2l)+\ffincr(x,nl,2))\dotminus\ffincr(y,nl,2)=$$ $$=\langle
2^{2l-1}+x'_0-y'_0,\ldots,2^{2l-1}+x'_{n-1}-y'_{n-1};\ 2l\rangle.
$$
One can notice that for any $i$ ($0\leq i<n$) it holds that
$$
0\leq 2^{2l-1}+x'_i-y'_i<2^{2l}.
$$
From this, the statement \ref{statrep}, and simple properties of
binary number notation it follows that
$$
((\ffrep(2^{2l\dotminus
1},n,2l)+\ffincr(x,nl,2))\dotminus\ffincr(y,nl,2))\wedge\ffrep(2^{2l\dotminus
1},n,2l)=
$$
$$
\langle (2^{2l-1}+x'_0\dotminus y'_0)\wedge 2^{2l-1},\ldots,
(2^{2l-1}+x'_{n-1}\dotminus y'_{n-1})\wedge 2^{2l-1};\ 2l\rangle.
$$
It is obvious that for any $i$ ($0\leq i<n$)
$$
(2^{2l-1}+x'_i-y'_i)\wedge 2^{2l-1}=\begin{cases}2^{2l-1},\text{
if }x'_i\geq y'_i, \\ 0\text{ otherwise}.\end{cases}
$$
Besides, for any $i$ ($0\leq i<n$) it is true that
$$
(x_i\geq y_i)\Leftrightarrow (x'_i\geq y'_i).
$$
Thereby,
$$
((\ffrep(2^{2l\dotminus
1},n,2l)+\ffincr(x,nl,2))\dotminus\ffincr(y,nl,2))\wedge\ffrep(2^{2l\dotminus
1},n,2l)=
$$
$$
\langle \sigma_02^{2l-1},\ldots,\sigma_{n-1}2^{2l-1};\ 2l\rangle.
$$
From this and from the statement \ref{statdecr} it follows that
$$
\ffcmp(x,y,n,l)=\langle \sigma_0,\ldots,\sigma_{n-1};\ 1\rangle.
$$
$\mains$- polynomiality follows from \ref{statincr},
\ref{statdecr}, \ref{statrep}, \ref{spoly}, \ref{spolyswap}.
\end{proof}
Let
$$
\ffcmpeq(x,y,n,l)=\ffcmp(x,y,n,l)\wedge\ffcmp(y,x,n,l).
$$
\begin{stat}
\label{statcmpeq} Let $n,l\geq 1,$
$$
x=\langle x_0,\ldots,x_{n-1};\ l\rangle,\quad y=\langle
y_0,\ldots,y_{n-1};\ l\rangle,
$$
$$
0\leq x_0,\ldots,x_{n-1},y_0,\ldots,y_{n-1}<2^l.
$$
Then
$$
\ffcmpeq(x,y,n,l)=\langle \sigma_0,\ldots,\sigma_{n-1};\ 1\rangle,
$$
where
$$
\sigma_i=\begin{cases}1,\text{ if }x_i=y_i,\\ 0\text{
otherwise.}\end{cases}
$$
Besides, $\ffcmpeq(x,y,n,l)$ is $\mains$-polynomial with respect to
$\{x,y\}.$
\end{stat}
\begin{proof}
From the statement \ref{statcmp} it follows that
$$
\ffcmp(x,y,n,l)=\langle \sigma'_0,\ldots,\sigma'_{n-1};\
1\rangle,\quad \ffcmp(y,x,n,l)=\langle
\sigma''_0,\ldots,\sigma''_{n-1};\ 1\rangle,
$$
where
$$
\sigma'_i=\begin{cases}1,\text{ if }x_i\geq y_i,\\ 0\text{
otherwise,}\end{cases} \quad \sigma''_i=\begin{cases}1,\text{ if
}x_i\leq y_i,\\ 0\text{ otherwise.}\end{cases}
$$
From here it follows the first part of the statement that is being
proved. The second part follows from the statement \ref{statcmp}.
\end{proof}
\begin{stat}\label{statsumpolyexp} For any $n\geq 0$ the function
$g_n(y,z)$, which is defined by the following relation
$$
g_n(y,z)=\Sum_{x<y}2^{xz}x^n,
$$
belongs to $\left[\mains\right]_{2^x}.$
\end{stat}
\begin{proof}
This statement follows from known formulas for summation.
\end{proof}
\begin{conseq}
If $r(z_1,\ldots,z_n)$ is a polynomial with natural coefficients, then
$$
\Sum_{0\leq
z_1,\ldots,z_n<x}r(z_1,\ldots,z_n)2^{y(z_1+z_2x+\ldots+z_nx^{n-1})}\in[\mains]_{2^x}.
$$
\end{conseq}
\begin{proof}
Indeed, it is obvious that one needs to consider the case where $r$
is a monomial,
$$
r(z_1,\ldots,z_n)=C\cdot z_1^{m_1}\ldots z_n^{m_n}.
$$
Then
$$
\Sum_{0\leq
z_1,\ldots,z_n<x}r(z_1,\ldots,z_n)2^{y(z_1+z_2x+\ldots+z_nx^{n-1})}=$$
$$=\Sum_{0\leq z_1,\ldots,z_n<x}C\cdot z_1^{m_1}\ldots
z_n^{m_n}2^{y(z_1+z_2x+\ldots+z_nx^{n-1})}=
$$
$$
=C\cdot \left(\Sum_{0\leq z<x}z^{m_1}2^{zy}\right)\cdot
\left(\Sum_{0\leq
z<x}z^{m_2}2^{zyx}\right)\cdot\ldots\cdot\left(\Sum_{0\leq
z<x}z^{m_n}2^{zyx^{n-1}}\right).
$$
Thus, from the statement \ref{statsumpolyexp} it follows the claim
that one was proving.
\end{proof}
\begin{stat}\label{cmppolyright}
Let $p(x_1,\ldots,x_n)$ and $q(x_1,\ldots,x_n)$ be polynomials with
coefficients from $\nat$. Then the predicate
$$
\varphi(x_1,\ldots,x_n)\equiv(p(x_1,\ldots,x_n)\geq
q(x_1,\ldots,x_n))
$$
is the correct one.
\end{stat}
\begin{proof}
For any function $r(z_1,\ldots,z_n)$ one denotes
$$
g_r(x,y)=\Sum_{0\leq
z_1,\ldots,z_n<x}r(z_1,\ldots,z_n)2^{y(z_1+z_2x+\ldots+z_nx^{n-1})}.
$$

From the consequence from the statement \ref{statsumpolyexp} it
follows that $g_p(x,y)\in[\mains]_{2^x}$ and
$g_q(x,y)\in[\mains]_{2^x}.$

Let
$$
f(x)=\ffcmp(g_p(x,p+q+1),g_q(x,p+q+1),x^n,p+q+1),
$$
where $p,q$ are contracted notations for
$p(\underbrace{x,\ldots,x}_{n\text{ times}})$ and
$q(\underbrace{x,\ldots,x}_{n\text{ times}})$ respectively. One can
prove that for any $x\geq 1$ it is true that
$$
f(x)=f_{\varphi}(x),
$$
where $f_{\varphi}$ is the generating function of the predicate
$\varphi.$ Indeed, let $x\geq 1.$ One can notice that
$$
g_p(x,p+q+1)=\langle
p(0,\ldots,0),p(1,0,\ldots,0),\ldots,p(x-1,\ldots,x-1);\
p+q+1\rangle,
$$
$$
g_q(x,p+q+1)=\langle
q(0,\ldots,0),q(1,0,\ldots,0),\ldots,q(x-1,\ldots,x-1);\
p+q+1\rangle
$$
(vectors are ordered in reverse lexicographical order). It is
obvious that for any $z_1,\ldots,z_n$ such that $0\leq
z_1,\ldots,z_n<x,$ it holds true
$$
p(z_1,\ldots,z_n),q(z_1,\ldots,z_n)<2^{p+q+1}.
$$
From this and from the stastement \ref{statcmp} one can conclude that
$$
f(x)=\langle\sigma(0,\ldots,0),\sigma(1,0,\ldots,0),\ldots,\sigma(x-1,\ldots,x-1);\
1\rangle,
$$
where
$$
\sigma(z_1,\ldots,z_n)=\begin{cases}1,\text{ if
}p(z_1,\ldots,z_n)\geq q(z_1,\ldots,z_n), \\ 0\text{
otherwise.}\end{cases}
$$
Thereby, for $x\geq 1$ $f(x)=f_{\varphi}(x).$ From statements
\ref{statcmp}, \ref{spoly}, \ref{spolyfinite} it follows that
$f_{\varphi}(x)\in\left[\mains\right]_{2^x}.$ Thereby,
$\varphi$ is a correct predicate. The statement is proved.
\end{proof}
\begin{conseq}
For polynomials $p$ and $q$ the predicates $p=q,$ $p\neq q,$ $p>q$ are correct.
\end{conseq}
\begin{proof}
Indeed, it follows from the statement \ref{logicright} and relations
$$
(p=q)\equiv(p\geq q)\&(q\geq p),\quad (p\neq q)\equiv\neg(p=q),$$
$$(p>q)\equiv(p\geq q)\&\neg(q\geq p).
$$
\end{proof}
\begin{stat}\label{explicitright}
The set of all correct predicates is closed with respect to explicit
transformations.
\end{stat}
\begin{proof}
It is obvious that to prove the statements one needs to establish
the fact that the set of all correct predicates is closed with
respect to variables permutation, substituting constants instead of
 last variable, identification of last two variables, introduction
of the dummy variable at the last place.
\begin{enumerate}
\item \emph{Permutation of variables.} let
$$
\varphi(x_1,\ldots,x_n)=\psi(x_{i_1},\ldots,x_{i_n}),
$$
where $(i_1,\ldots,i_n)$ is some permutation of numbers
$1,2,\ldots,n,$ $f_{\varphi}(y),f_{\psi}(y)$ are the generating
functions of predicates $\varphi$ and $\psi$ respectively, the
predicate $\psi$ is a correct one. Then
$$
f_{\psi}(y)=\Sum_{0\leq
x_1,\ldots,x_n<y}\chi_{\psi}(x_1,\ldots,x_n)2^{x_1+x_2y+\ldots+x_ny^{n-1}}=$$
$$ = \Sum_{0\leq
x_1,\ldots,x_n<y}\chi_{\psi}(x_1,\ldots,x_n)2^{k_1x_1+\ldots+k_nx_n},
$$
$$
f_{\varphi}(y)=\Sum_{0\leq
x_1,\ldots,x_n<y}\chi_{\psi}(x_{i_1},\ldots,z_{i_n})2^{x_1+x_2y+\ldots+x_ny^{n-1}}
=$$ $$=\Sum_{0\leq
z_1,\ldots,z_n<y}\chi_{\psi}(z_1,\ldots,z_n)2^{m_1z_1+\ldots+m_nz_n},
$$
where
$$
k_i=y^{i-1},\ m_i=y^{j_i-1},\ 1\leq i\leq n,
$$
$(j_1,\ldots,j_n)$ is the inverse permutation to $(i_1,\ldots,i_n),$
$\chi_{\varphi},\,\chi_{\psi}$ are the characteristic functions of
the predicates $\varphi$ and $\psi$ respectively. It is easy to see
that
 for numbers $y,k_1,\ldots,k_n,m_1,\ldots,m_n$ the
 conditions of the statement \ref{statswap} are satisfied.
 From this it follows that
$$
f_{\varphi}(y)=\ffswap_n(f_{\psi}(y),y,k_1,\ldots,k_n,m_1,\ldots,m_n)=
$$
$$
=\ffswap_n(f_{\psi}(y),y,1,y,\ldots,y^{n-1},y^{j_1-1},\ldots,y^{j_n-1}).
$$
From this, from the statements \ref{spoly}, \ref{spolyswap},
\ref{statswap}, and from $f_{\psi}\in\left[\mains\right]_{2^x}$ it
follows that $f_{\varphi}\in\left[\mains\right].$ Thereby, the
predicate $\varphi$ is correct. \item \emph{Substituing of a
constant in the place of the last variable.} Let
$$
\varphi(x_1,\ldots,x_n)=\psi(x_1,\ldots,x_n,a),
$$
where $\psi$ is a correct predicate, $a\in N$ is a constant. Let one assume that
$$
\rho(x_1,\ldots,x_{n+1})\equiv\psi(x_1,\ldots,x_{n+1})\&(x_{n+1}=a).
$$
From the statement \ref{logicright}, the consequenc from the
statement \ref{cmppolyright}, and from the fact that $\psi$ is
correct, it follows that $\rho$ is also correct. Let
$\chi_{\varphi},$ $\chi_{\psi},$ $\chi_{\rho}$ be the characteristic
functions of the predicates $\varphi,$ $\psi,$ $\rho$ respectively,
$f_{\varphi},$ $f_{\psi},$ $f_{\rho}$ are their generating
functions. Then for $y>a$ one has
$$
f_{\rho}(y)=\Sum_{0\leq
x_1,\ldots,x_{n+1}<y}\chi_{\rho}(x_1,\ldots,x_{n+1})2^{x_1+x_2y+\ldots,x_{n+1}y^n}=
$$
$$
\Sum_{0\leq
x_1,\ldots,x_n<y}\chi_{\psi}(x_1,\ldots,x_n,a)2^{x_1+x_2y+\ldots,x_ny^{n-1}+ay^n}=
$$
$$
=2^{ay^n}\cdot\Sum_{0\leq
x_1,\ldots,x_n<y}\chi_{\varphi}(x_1,\ldots,x_n)2^{x_1+x_2y+\ldots,x_ny^{n-1}}=2^{ay^n}\cdot
f_{\varphi}(y).
$$
Thereby, for any $y>a$ it satisfies
$$
f_{\varphi}(y)=\left[\frac{f_{\rho}(y)}{2^{ay^n}}\right].
$$
From here, from $f_{\rho}\in\left[\mains\right]_{2^x}$, and from
statements \ref{spoly}, \ref{spolyfinite} it follows that
$f_{\varphi}\in\left[\mains\right]_{2^x}.$ Therefore, the predicate
$\varphi$ is correct. \item \emph{Identification of the last two
variables.} Let
$$
\varphi(x_1,\ldots,x_n)=\psi(x_1,\ldots,x_n,x_n),
$$
where $\psi$ is a correct predicate. Let one contend that
$$
\rho(x_1,\ldots,x_{n+1})\equiv\psi(x_1,\ldots,x_{n+1})\&(x_n=x_{n+1}).
$$
From the statement \ref{logicright}, consequence of the statement
\ref{cmppolyright} and from the fact that $\psi$ is correct, it follows that $\rho$ is also correct. Let $\chi_{\varphi},$ $\chi_{\rho}$ be characteristic functions of predicates $\varphi,$ $\rho$
respectively, $f_{\varphi},$
 $f_{\rho}$ are their generating functions. Then one has
$$
f_{\rho}(y)=\Sum_{0\leq
x_1,\ldots,x_{n+1}<y}\chi_{\rho}(x_1,\ldots,x_{n+1})2^{x_1+x_2y+\ldots,x_{n+1}y^n}
$$
$$
=\Sum_{0\leq
x_1,\ldots,x_n<y}\chi_{\varphi}(x_1,\ldots,x_n)2^{x_1+x_2y+\ldots+x_{n-1}y^{n-2}+x_n(y^{n-1}+y^n)}
$$
$$
=\Sum_{0\leq
x_1,\ldots,x_n<y}\chi_{\varphi}(x_1,\ldots,x_n)2^{k_1x_1+\ldots+k_nx_n},
$$
where
$$
k_i=y^{i-1},\,1\leq i\leq n-1,\quad k_n=y^{n-1}+y^n.
$$
On the other hand,
$$
f_{\varphi}(y)=\Sum_{0\leq
x_1,\ldots,x_n<y}\chi_{\varphi}(x_1,\ldots,x_n)2^{m_1x_1+\ldots+m_nx_n},
$$
where
$$
m_i=y^{i-1},\ 1\leq i\leq n.
$$
It is easy to check that for numbers
$y,k_1,\ldots,k_n,m_1,\ldots,m_n$ the conditions of the statement
\ref{statswap} are satisfied. Thereby, one obtains
$$
f_{\varphi}(y)=\ffswap_n(f_{\rho}(y),y,k_1,\ldots,k_n,m_1,\ldots,m_n)=
$$
$$
=\ffswap_n(f_{\rho}(y),y,1,y,\ldots,y^{n-2},y^{n-1}+y^n,1,y,\ldots,y^{n-1}).
$$
From this and from the statements \ref{statswap}, \ref{spoly},
\ref{spolyswap} it follows that
$f_{\varphi}\in\left[\mains\right]_{2^x}.$ Thereby, the predicate
$\varphi$ is correct. \item \emph{Introduction of a dummy variable in place of the last variable.} Let
$$
\varphi(x_1,\ldots,x_n,x_{n+1})=\psi(x_1,\ldots,x_n),
$$
$\psi$ is a correct predicate, $\chi_{\varphi},$ $\chi_{\psi}$ are
the characteristic functions of the predicates $\varphi,$ $\psi$
respectively, $f_{\varphi},$
 $f_{\psi}$ are their generating functions. For $y\geq 1$ one has
$$
f_{\varphi}(y)=\Sum_{0\leq
x_1,\ldots,x_{n+1}<y}\chi_{\psi}(x_1,\ldots,x_n)2^{x_1+x_2y+\ldots+x_{n+1}y^n}
$$
$$
=\left(\Sum_{0\leq
x_1,\ldots,x_n<y}\chi_{\psi}(x_1,\ldots,x_n)2^{x_1+x_2y+\ldots+x_ny^{n-1}}\right)\cdot
\left(\Sum_{0\leq x<y}2^{xy^n}\right)
$$
$$
=f_{\psi}(y)\cdot\left[\frac{2^{y^{n+1}\dotminus
1}}{2^{y^n}\dotminus 1}\right].
$$
From this and from statements \ref{spoly} and \ref{spolyfinite} it follows that  $f_{\varphi}\in\left[\mains\right]_{2^x},$ i.e. $\varphi$ is a correct predicate.
\end{enumerate}
The statement is proved.
\end{proof}

Let
$$
\ffsum(x,n,l,k)=\left[\frac{(x\cdot\left[\frac{2^{kl}\dotminus
1}{2^l\dotminus
1}\right])\wedge\ffrep(\ffrep(1,l,1),n,kl)}{2^{(k\dotminus
1)l}}\right].
$$
\begin{stat}\label{statsum}
Let $n,l,k\geq 1,$ $k<2^l,$
$$
x=\langle
x_{0,0},x_{0,1},\ldots,x_{0,k-1},\ldots,x_{n-1,0},x_{n-1,1},\ldots,x_{n-1,k-1};\
l\rangle,
$$
where for all $i,j$ ($0\leq i<n,$ $0\leq j<k$) it satisfies $0\leq
x_{i,j}\leq 1.$ Then
$$
\ffsum(x,n,l,k)=\langle s_0,\ldots,s_{n-1};\ lk\rangle,
$$
where
$$
s_i=\Sum_{j=0}^{k-1}x_{i,j},\quad 0\leq i<n.
$$
Besides, $\ffsum(x,n,l,k)$ is $\mains$-polynomial with respect to
$\{x\}.$
\end{stat}
\begin{proof}
Let one represent $x$ in the following way:
$$
x=\langle y_0,\ldots,y_{kn-1};\ l\rangle,
$$
such that
$$
0\leq y_i\leq 1,\quad 0\leq i<kn.
$$
Then
$$
x\cdot\left[\frac{2^{kl}\dotminus 1}{2^l\dotminus 1}\right] =
\left(\Sum_{0\leq i<kn}y_i2^{il}\right)\cdot \left(\Sum_{0\leq
j<k}2^{jl}\right) = \Sum_{\substack{0\leq i<kn \\ 0\leq
j<k}}y_i2^{(i+j)l}
$$
$$
=\Sum_{p=0}^{k(n+1)-2}\Sum_{\substack{0\leq i<kn \\
0\leq p-i<k}}y_i2^{pl} = \langle z_0,\ldots,z_{k(n+1)-2};\
l\rangle,
$$
where
$$
z_p=\Sum_{\substack{0\leq i<kn \\ 0\leq p-i<k}}y_i\quad\quad
(0\leq p\leq k(n+1)-2).
$$
One can note that the binary notation of the number
$\ffrep(\ffrep(1,l,1),n,kl)$ consists of $n$ blocks with ones, such
that  $r$-th block occupies digits from $l(rk+k-1)$-th up to
$(l(rk+k)-1)$-th ($0\leq r<n$). From this, from the fact that for
any $p$ ($0\leq p\leq k(n+1)-2$) it satisfies $z_p\leq k<2^l,$ and
from the fact that for any $i$ ($0\leq i<n$) it is true that
$s_i=z_{ik+k-1},$ it follows that
$$
\left(x\cdot\left[\frac{2^{kl}\dotminus 1}{2^l\dotminus
1}\right]\right)\wedge\ffrep(\ffrep(1,l,1),n,kl)
$$
$$
=\langle \underbrace{0,\ldots,0}_{k-1\text{
times}},z_{k-1},\underbrace{0,\ldots,0}_{k-1\text{ times}},
z_{2k-1},\ldots, \underbrace{0,\ldots,0}_{k-1\text{ times}},
z_{nk-1};\ l\rangle
$$
$$
=\langle \underbrace{0,\ldots,0}_{k-1\text{
times}},s_0,\underbrace{0,\ldots,0}_{k-1\text{ times}}, s_1,\ldots,
\underbrace{0,\ldots,0}_{k-1\text{ times}}, s_{n-1};\ l\rangle
$$
$$
=2^{(k-1)l}\cdot\langle s_0,\ldots,s_{n-1};\ kl\rangle.
$$
From this it follows that
$$
\ffsum(x,n,l,k)=\langle s_0,\ldots,s_{n-1};\ kl\rangle.
$$
$\mains$-polynomiality with respect to $\{x\}$ follows from the statements
\ref{statrep}, \ref{spoly}, \ref{spolyswap}. The statement is proved.
\end{proof}

\begin{stat}\label{countright}
The set of all correct predicates is closed with respect to the
operation of counting.
\end{stat}
\begin{proof}
Due to statement \ref{explicitright} it is sufficient to prove the
closeness in terms of counting with respect to the first variable.
Let $\varphi(x_1,\ldots,x_n,y)$ be obtained from
$\psi(x_1,\ldots,x_n)$ by counting operation with respect to the
variable $x_1$ and the polynomial $p(x_1,\ldots,x_n),$ $\psi$ be a
correct polynomial. Let one introduce the predicate $\rho$ in the
following way:
$$
\rho(x,x_1,\ldots,x_n,y)\equiv\psi(x,x_2,\ldots,x_n)\&(x<p(x_1,\ldots,x_n)).
$$
From the correctness of $\psi,$ statements \ref{logicright},
\ref{explicitright}, and a consequence from the statement
\ref{cmppolyright} it follows that $\rho$ is a correct predicate.
Let
$$
q(z)=p(\underbrace{z,\ldots,z}_{n\text{ раз}})+z+1.
$$
One can assume that
$$
f_{\rho}'(z)=\ffincr(f_{\rho}(q(z)),q(z)^{n+2},q(z)).
$$
Let $z\geq 1.$ From statements \ref{statincr}, \ref{spoly}
it follows that $f_{\rho}'\in\left[\mains\right]_{2^x}.$ Besides,
$$
f_{\rho}(q(z))=\langle
\chi_{\rho}(0,\ldots,0),\chi_{\rho}(1,0,\ldots,0), \ldots,
\chi_{\rho}(q(z)-1,\ldots,q(z)-1);\ 1\rangle.
$$
Thereby, from the statement \ref{statincr} it follows that
$$
f_{\rho}'(z)=\langle
\chi_{\rho}(0,\ldots,0),\chi_{\rho}(1,0,\ldots,0), \ldots,
\chi_{\rho}(q(z)-1,\ldots,q(z)-1);\ q(z)\rangle.
$$
One can assume that
$$
u(z)=\ffsum(f_{\rho}'(z),q(z)^{n+1},q(z),q(z)).
$$
From the statement \ref{statsum} it follows that
\begin{equation}\label{statcountu}
u(z)=\langle g(0,\ldots,0,z),
g(1,0,\ldots,0,z),\ldots,g(q(z)-1,\ldots,q(z)-1,z);\
q(z)^2\rangle,
\end{equation}
 where $g(x_1,\ldots,x_n,y,z)$ is the number of
 $x<q(z)$ such that $\rho(x,x_1,\ldots,x_n,y)$ holds true.
Besides, from the statements \ref{statsum} and \ref{spoly} it follows that
$u\in\left[\mains\right]_{2^x}.$ Let
\begin{equation}\label{statcountv}
v(z)=\langle h(0,\ldots,0),
h(1,0,\ldots,0),\ldots,h(q(z)-1,\ldots,q(z)-1);\ q(z)^2\rangle,
\end{equation} where $h(x_1,\ldots,x_n,y)=y$ for any $x_1,\ldots,x_n,y\in\nat.$ It is obvious that
$$
v(z)=\left(\Sum_{y=0}^{q(z)-1}y2^{yq(z)^{n-2}}\right)\cdot
\left(\Sum_{i=0}^{q(z)^n-1}2^{q(z)^2i}\right).
$$
From the statement \ref{statsumpolyexp} it follows that
$v(z)\in\left[\mains\right]_{2^x}.$ From (\ref{statcountu}),
(\ref{statcountv}), and the statement \ref{statcmpeq} it follows
that
\begin{equation}\label{statcountcmpeq}
\begin{array}{c}
\ffcmpeq(u(z),v(z),q(z)^{n+1},q(z)^2)= \\ =\langle
\sigma(0,\ldots,0,z),\sigma(1,\ldots,0,z),\ldots,
\sigma(q(z)-1,\ldots,q(z)-1,z) ;\ 1\rangle,
\end{array}
\end{equation}
where
$$
\sigma(x_1,\ldots,x_n,y,z)=\begin{cases} 1,\text{if $y$ is the number of $x<q(z)$ such that} \\ \quad \quad \text{$\rho(x,x_1,\ldots,x_n,y)$ holds true,} \\ 0\text{ otherwise.}
\end{cases}
$$
Let
$$
w(z)=\Sum_{0\leq
i_1,\ldots,i_{n+1}<z}2^{i_1+i_2q(z)+\ldots+i_{n+1}q(z)^n}=
\prod_{j=1}^{n+1}\left(\Sum_{i=0}^{z-1}2^{iq(z)^{j-1}}\right).
$$
From the statement \ref{statsumpolyexp} it follows that
$w\in\left[\mains\right].$ From the fact that $q(z)>z,$ it follows
that
$$
w(z)=\langle
\xi(0,\ldots,0,z),\xi(1,\ldots,0,z),\ldots,\xi(q(z)-1,\ldots,q(z)-1,z);
\ 1\rangle,
$$
where for all $x_1,\ldots,x_n,y,z\in\nat$ it satisfies
$$
\xi(x_1,\ldots,x_n,y,z)=\begin{cases} 1,\text{ if
}(x_1<z)\&\ldots\&(x_n<z)\&(y<z)\text{ is true}, \\ 0\text{
otherwise.}\end{cases}
$$
From this and (\ref{statcountcmpeq}) it follows that
$$
\ffcmpeq(u(z),v(z),q(z)^{n+1},q(z)^2)\wedge w(z)
$$
$$
=\langle \sigma'(0,\ldots,0,z),\sigma'(1,\ldots,0,z),\ldots,
\sigma'(q(z)-1,\ldots,q(z)-1,z) ;\ 1\rangle,
$$
where
$$
\sigma'(x_1,\ldots,x_n,y,z)=\begin{cases} 1,\text{if
$(x_1<z)\&\ldots\&(x_n<z)\&(y<z)$ is satisfied and} \\ \quad\quad \text{$y$ is the number $x<q(z)$ such that} \\
\quad \quad \text{
$\rho(x,x_1,\ldots,x_n,y)$ holds true,} \\
0\text{ otherwise.}
\end{cases}
$$
From this it follows that
$$
\ffcmpeq(u(z),v(z),q(z)^{n+1},q(z)^2)\wedge w(z)
$$
$$
=\Sum_{0\leq
x_1,\ldots,x_n,y<z}\sigma(x_1,\ldots,x_n,y,z)2^{x_1+x_2q(z)+\ldots+x_nq(z)^{n-1}+yq(z)^n}.
$$
From this, from the fact that $q(z)>z,$ and from the statement
\ref{statswap} it follows that
\begin{equation}\label{statcountmain}
\begin{array}{c}
\ffswap_{n+1}(\ffcmpeq(u(z),v(z),q(z)^{n+1},q(z)^2)\wedge
w(z), \\ z,1,q(z),\ldots,q(z)^n,1,z,\ldots,z^n)= \\
=\Sum_{0\leq
x_1,\ldots,x_n,y<z}\sigma(x_1,\ldots,x_n,y,z)2^{x_1+x_2z+\ldots+x_nz^{n-1}+yz^n}=
\\ =\langle \sigma(0,\ldots,0,z),\sigma(1,\ldots,0,z),\ldots,
\sigma(z-1,\ldots,z-1,z) ;\ 1\rangle.
\end{array}
\end{equation}
It is obvious that for any $x_1,\ldots,x_n<z$ it holds true that
$p(x_1,\ldots,x_n)<q(z).$ From this and from the definitions of
$\sigma$ and $\rho$ it follows that for any $x_1,\ldots,x_n,y,z$
such that $x_1,\ldots,x_n<z,$ it is true that
$$
\sigma(x_1,\ldots,x_n,y,z)=\begin{cases} 1,\text{ if $y$ is the
number  $x<p(x_1,\ldots,x_n)$ such that} \\ \quad\quad \text{
$\psi(x,x_2,\ldots,x_n)$ is true,} \\ 0 \text{
otherwise.}\end{cases}
$$
From this and (\ref{statcountmain}) it follows that for any $z\geq
1$ it satisfies
$$
f_{\varphi}(z)=\ffswap_{n+1}(\ffcmpeq(u(z),v(z),q(z)^{n+1},q(z)^2)\wedge
w(z),$$ $$z,1,q(z),\ldots,q(z)^n,1,z,\ldots,z^n).
$$
Thus, from inclusions $u(z),v(z),w(z)\in\left[\mains\right]$ and from statements \ref{statswap}, \ref{statcmpeq}, \ref{spolyfinite},
\ref{spoly}, \ref{spolyswap} it follows that
$f_{\varphi}\in\left[\mains\right].$ The statement is proved.
\end{proof}
\begin{stat}\label{cccrright}
Any predicate from $\cccr$ is correct.
\end{stat}
\begin{proof}Indeed, according to the consequence from the statement \ref{cmppolyright}
the predicates $x+y=z$ and $xy=z$ are correct. From here and from
the statements \ref{explicitright}, \ref{logicright} and
\ref{countright} it follows that any predicate from $\cccr$ is
correct. The statement is proved.
\end{proof}
\begin{theorem}\label{mytheorem}There is an inclusion
$$
\ccxs\subseteq\left[\mains\right]_{2^x}.
$$
\end{theorem}
\begin{proof}
Let $f(x_1,\ldots,x_n)\in\ccxs.$ Then $$
g(x_1,\ldots,x_n,y)=f(x_1,\ldots,x_n)\langle y\rangle\in\ccclh.
$$
From the statement \ref{statscolemincccr} it follows that
$g(x_1,\ldots,x_n,y)$ is the characteristic function of a predicate
$\psi\in\cccr.$ From the statement \ref{cccrright} it follows that
 $\psi$ is the correct predicate. From the definition of the generating function it follows that
$$
f_{\psi}(z)=\Sum_{0\leq
x_1,\ldots,x_n,y<z}g(x_1,\ldots,x_n,y)2^{x_1+x_2z+\ldots+x_nz^{n-1}+yz^n}.
$$
Thereby, for any $x_1,\ldots,x_n,y,z$ such that $x_1,\ldots,x_n,y<z$
and $z\geq 1,$ $y$-th binary digit of the number $f(x_1,\ldots,x_n)$
equals to the binary digit of the number $f_{\psi}(z)$ with the
number $x_1+x_2z+\ldots+x_nz^{n-1}+yz^n.$ If additionally the length
of the binary notation of $f(x_1,\ldots,x_n)$ does not exceed $t,$
then from the statement \ref{statdecr} it follows that
$$
f(x_1,\ldots,x_n)=\ffdecr\left(\left[\frac{f_{\psi}(z)}{2^{x_1+x_2z+\ldots+x_nz^{n-1}}}\right],tz^n,z^n\right).
$$
By plugging in $z$ the expression $x_1+\ldots+x_n+1,$ and instead of $t$
a polynomial $t(x_1,\ldots,x_n)$ such that for any $x_1,\ldots,x_n$
it holds that
$$
f(x_1,\ldots,x_n)<2^{t(x_1,\ldots,x_n)},
$$
one obtains an expression for $f.$ From the correctness of the
predicate $\psi$and from statements \ref{statdecr} and \ref{spoly}
it follows that $f\in\left[\mains\right]_{2^x}.$ The statement is
proved.
\end{proof}

\subsection{Proof of Theorem  \ref{theorem_exp_skolem}}
There is the following inclusion
$$
[\mains]_{2^x}\subseteq[\mains]_{x^y}.
$$
Besides, from the theorem \ref{statburjui} it follows that
$[\mains]_{x^y}\subseteq\ccxs,$ and from the theorem \ref{mytheorem}
it follows that $\ccxs\subseteq[\mains]_{2^x}.$ Thereby,
$$
[\mains]_{x^y}\subseteq\ccxs\subseteq[\mains]_{2^x}\subseteq[\mains]_{x^y}.
$$
Thus
$$
\ccxs=[\mains]_{2^x}=[\mains]_{x^y}.
$$
The main theorem is proved.

\section{Basis by Superposition of $\ccffom$}\refstepcounter{sectionref}\label{section_ffom}

\subsection{Definitions}

If $\alpha$ is the line of symbols, $n$ is a number, then put
$$
\ffext(\alpha,n)=\alpha\underbrace{00\ldots 0}_{n\dotminus
|\alpha|\text{ нулей}}.
$$

If $x_1,\ldots,x_n,k\in\nat$, $k\geq 1$, then let $$
\ffcodexd_k(x_1,\ldots,x_n) =
\ffext(\ffcode(x_1,\ldots,x_n),2|\ffcode(x_1,\ldots,x_n)|^k).$$
Besides, if additionally
$$
y<2^{|\ffcode(x_1,\ldots,x_n)|^k},
$$
then by $\ffcodey_k(x_1,\ldots,x_n;\ y)$ one can denote the line
$\alpha_1\beta_1\alpha_2\beta_2\ldots\alpha_l\beta_l$, where
$\alpha_1\ldots\alpha_l=\ffext(\ffcode(x_1,\ldots,x_n),|\ffcode(x_1,\ldots,x_n)|^k)$,
$\beta_l\ldots\beta_2\beta_1$ is the binary notation of $y$ (as
completed by zeroes from the left, if needs be, that is $\beta_1$ is
the lowest order digit, $\beta_2$ follows it, etc.).

The class $\ccffomalt$ can be defined as the set of everywhere defined over the set $\nat$ functions $f(x_1,\ldots,x_n)$ such that the following two conditions are satisfied

\begin{enumerate}
\item There exists a $k$ such that for any $x_1,\ldots,x_n$
it satisfies \begin{equation}\label{eq_codek}
\fflen(f(x_1,\ldots,x_n))\leq|\ffcode(x_1,\ldots,x_n)|^k.
\end{equation}
\item For any $k\geq 1$ that satisfies the condition above there exists a predicate $\rho\in \ccfom$
such that for any $x_1,\ldots,x_n,y$
from $y<|\ffcodexd_k(x_1,\ldots,x_n)|$ it follows that
$$
\rho(\ffcodey_k(x_1,\ldots,x_n;\
y))\equiv(f(x_1,\ldots,x_n)\langle y\rangle=1).
$$
\end{enumerate}

Let $\ccffomvar$ be the set of all everywhere defined over the set
$\nat$ functions $f(x_1,\ldots,x_n)$, for which the following two
conditions are satisfied.
\begin{enumerate}
\item There exists a $k$ such that for any $x_1,\ldots,x_n$ it holds that  (\ref{eq_codek}).
\item For any $k\geq 1$ that satisfies the condition above there
exists a $\ccfom$-formula $\Phi$ over the variables
$z_1,\ldots,z_m,y$, with a corresponding predicate
$\rho_{\Phi}(X,z_1,\ldots,z_m,y)$, where
$$
X\in\{0,1\}^+,\quad 1\leq x_1,\ldots,z_m,y\leq |X|,
$$ such that for any $x_1,\ldots,x_n,y,z_1,\ldots,z_m\in\nat$ from
$$
1\leq z_1,\ldots,z_m,y\leq |\ffcodexd_k(x_1,\ldots,x_n)|
$$ it follows that
$$
\rho_{\Phi}(\ffcodexd_k(x_1,\ldots,x_n),z_1,\ldots,z_m,y)\equiv(f(x_1,\ldots,x_n)\langle
y-1\rangle=1).
$$
\end{enumerate}

Let
$$\mainffoms = \{x+y,\quad x\dotminus y,\quad x\wedge
y,\quad \left[x/y\right],\quad 2^{[\log_2 x]^2}\}.$$

If $Q$ is some class of functions, then by $Q^{\lllog}$ one can
denote the set of all functions of the form $[\log_2 f]$, where
$f\in Q$.

\Def{The function $h(X,y_1,\ldots,y_m)$, which is defined for
$X\in\{0,1\}^+$, $1\leq y_1,\ldots,y_n\leq |X|$ and takes values
from $\nat$, is called \emph{$\tth$-function} if for any
 $X,y_1,\ldots,y_m$ from its domain it satisfies
$$
h(X,y_1,\ldots,y_m)<2^{|X|}.
$$}

\Def{$\tth$-function $h(X,y_1,\ldots,y_m)$ is called \emph{correct}
if there exists a function $f(x,z)\in[\mainffoms]$ such that for any
$X$, $y_1,\ldots,y_m$ from the domain of $h$ it holds that
$$
f(c(X),2^{|X|})=\sum_{1\leq y_1,\ldots,y_m\leq
|X|}(2^{(y_1-1)|X|+(y_2-1)|X|^2+\ldots+(y_m-1)|X|^m}h(X,y_1,\ldots,y_m)),
$$
where $c(X)$ is the number the binary notation of which (perhaps as
completed by zeroes from the left) is $X$ (for example if $X=00101$,
then $c(X)=5$).}

\Def{The predicate $\rho(X,y_1,\ldots,y_m)$ that is defined for
$X\in\{0,1\}^+$, $1\leq y_1,\ldots,y_m\leq |X|$ is called
\emph{correct}\footnote{this definition differs from the one used it
the section \ref{section_exp_skolem}}, if its characteristic
function $\chi_{\rho}(X,y_1,\ldots,y_m)$ is a correct
$\tth$-function.}

Let
$$
\langle x_0,\ldots,x_{n-1};\ l\rangle=\Sum_{i=0}^{n-1}x_i2^{il}.
$$
One can note that if the condition $x_0,x_1,\ldots,x_{n-1}<2^l$ is
satisfied for any $i$ ($0\leq i<n$) then the binary digits of the
number $\langle x_0,x_1,\ldots,x_{n-1};\ l\rangle$ from $(il)$-th up
to $(il+l-1)$-th generate the binary notation of the number $x_i.$

\subsection{Coincidence of classes $\ccffom$, $\ccffomalt$ and $\ccffomvar$}

\begin{stat}\label{stat_ffom_ffomalt}
$\ccffom=\ccffomalt$.
\end{stat}
\begin{proof}
It is easy to notice that the definitions of classes $\ccffom$ and
$\ccffomalt$ differ only with respect to how one encodes number
arrays by strings of symbols, such that the equivalency of these
encodings is obvious (see the equivalent definitions of the class
$\ccfom$ from ~\cite{uniformity}; for example based on the sequence
of boolean circuits that are generated by a Turing machine).
\end{proof}

\begin{stat}\label{stat_ffomalt_ffomvar}
$\ccffomalt=\ccffomvar$.
\end{stat}
\begin{proof}
The inclusion $\ccffomvar\subseteq\ccffomalt$ can be proved
analogously to the statement
 \ref{stat_ffom_ffomalt}. One can prove the inclusion
$\ccffomalt\subseteq\ccffomvar$.

Let $f\in\ccffomalt$, $k$ is a number that satisfies
(\ref{eq_codek}) for all $x_1,\ldots,x_n$. Besides, let
$\rho\in\ccfom$ be the predicate from the definition of $\ccffomalt$
for $f$ and $k$, $\Phi$ is a $\ccfom$-formula over the variables
$z_1,\ldots,z_m$ from the definition of $\ccfom$ for $\rho$.

By $\Psi$ one can denote the $\ccfom$-formula over variables
$z_1,\ldots,z_m,u,v,w,y$, that one obtains from $\Phi$ by
substituting every subformula of the form
 $X\langle t\rangle$, where $t$
is a $\ccfom$-term, to
\begin{equation}\label{eq_ffomvar_subst} \exists u(t=2u\ \&\ \llbit(y,u))\vee \exists v\exists
w(w=2v\ \&\ w=t+1\ \&\ X\langle v\rangle),
\end{equation}
with this, the auxiliary sub-formulas are being replaced based on equalities
$$(x>y)\equiv(x\geq y)\&\neg(y\geq x),$$
$$(y=x+1)\equiv y>x\ \&\ \neg\exists u(y>u\ \&\ u>x),$$
$$(y=2x)\equiv \forall u\forall
v(u=v+1\rightarrow(\llbit(y,u)\leftrightarrow\llbit(x,v)))
$$ $$\&\ \neg\llbit(y,1)\ \&\ \neg\llbit(x,|X|).$$ Let
$\psi(X,z_1,\ldots,z_m,u,v,w,y)$ be the corresponding to the formula
$\Psi$ predicate.

One can note that for any $x_1,\ldots,x_n,y$ such that $1\leq y\leq
|\ffcodexd_k(x_1,\ldots,x_n)|$, it holds that
$$
|\ffcodey(x_1,\ldots,x_n;\ y-1)|=|\ffcodexd(x_1,\ldots,x_n)|,
$$
i.e. when calculating
$$\rho(\ffcodey_k(x_1,\ldots,x_n;\ y-1))$$
and
$$\psi(\ffcodexd_k(x_1,\ldots,x_n),z_1,\ldots,z_m,u,v,w,y)$$
the quantifiers will have the same variable ranges.

Based on this and the fact that the expression
(\ref{eq_ffomvar_subst}) exhibits the needed re-coding (obviously),
one can easily see that for all arrays $(x_1,\ldots,x_n,y)$ such
that $1\leq y\leq |\ffcodexd(x_1,\ldots,x_n)|$, the following
statement holds true: for any $u,v,w$ it is true that
$$
\psi(\ffcodexd(x_1,\ldots,x_n),z_1,\ldots,z_m,u,v,w,y)\equiv(f(x_1,\ldots,x_n)\langle
y-1\rangle=1).
$$
Thereby, $f\in\ccffomvar$. The statement is proved.
\end{proof}

\subsection{Overview of Some Functions That Belong to the Class
$[\mainffoms]$}

\begin{stat}
All constants as well as functions $\ffsg(x)$, $\ffrm(x,y)$, $xy$
belong to $[\mainffoms]$.
\end{stat}
\begin{proof}
Let $f(x)=2^{[\log_2 x]^2}$. One can note that $0=x\dotminus x$,
$1=f(0)$, the remaining constants can be obtained from these with the help of the function $x+y$.

It is obvious that $\ffsg(x)=1\dotminus (1\dotminus x)$,
$\ffrm(x,y)=(x\dotminus[x/y]\cdot y)\cdot\ffsg(y)$.

Let $g(x)=f(x+x)$. One can note that for all $x$ it is true that
$g(x)\geq x^2$, thus $g(g(x+y))\geq(x+y)^4>2x^2y^2\geq
x^2y^2+xy$ for all $x,y\geq 1$. One can prove that
$$
xy=\left[\frac{g(g(x+y))}{[g(g(x+y))/x]/y}\right].
$$
For $xy=0$ it is obvious. If $x,y\geq 1$, then it follows from the fact that $[A/x]/y=[A/(xy)]$, where $A=g(g(x+y))$, and the chain of relationships
$$
xy\leq\frac{A}{\left[\frac{A}{xy}\right]}<\frac{A}{\frac{A}{xy}-1}<xy+1,
$$
where the last inequality follows from the fact that $A>x^2y^2+xy$.
The statement is proved.
\end{proof}

Let
$$
\ffssqrt(y)=y\dotminus\left(\left[
\frac{\left(\left[\frac{f(y)^4\dotminus 1}{y^2\dotminus
1}\right]\wedge(f(y)\dotminus
1)\right)\cdot\left[\frac{\left[f(y)^4/2\right]\dotminus
1}{\left[y^2/2\right]\dotminus 1}\right]}{\left[2f(y)/y^3\right]}
\right]\wedge(y\dotminus 1)\right),
$$
where $f(y)=2^{[\log_2 y]^2}$. One can note that
$\ffssqrt\in[\mainffoms]$.

\begin{stat}\label{stat_ssqrt}For any $x$ it holds that
$$
\ffssqrt(2^{2x})=2^x.
$$
\end{stat}
\begin{proof}
If $x=0$, then the statement is obvious. Let $x\geq 1$. One can assume that
$y=2^{2x}$. Then $f(y)=2^{4x^2}$.

By using the geometric progression sum formula, one obtains
$$
A=\frac{f(y)^4\dotminus 1}{y^2\dotminus
1}=\frac{2^{16x^2}-1}{2^{4x}-1}=\sum_{0\leq i\leq 4x-1}2^{4xi}.
$$

Analogously, one obtains
$$
B=\left[\frac{\left[f(y)^4/2\right]\dotminus
1}{\left[y^2/2\right]\dotminus
1}\right]=\frac{2^{16x^2-1}-1}{2^{4x-1}-1}=\sum_{0\leq i\leq
4x}2^{(4x-1)i}.
$$

One can note that the binary notation of the number $f(y)\dotminus 1$ is
actually $4x^2$ of consecutive ones, one obtains that
$$
C=A\wedge(f(y)\dotminus 1)=\sum_{0\leq i\leq x-1}2^{4xi}.
$$

Thereby,
$$
B\cdot C=\left(\sum_{0\leq i\leq
4x}2^{(4x-1)i}\right)\cdot\left(\sum_{0\leq i\leq
x-1}2^{4xi}\right)=\sum_{\substack{0\leq i\leq 4x \\ 0\leq j\leq
x-1}}2^{(4x-1)(i+j)+j}.
$$

One can note that in the last sum all powers are different, thus the
ones in the binary notation $B\cdot C$ stay on positions of type
$(4x-1)(i+j)+j$ ($0\leq i\leq 4x$, $0\leq j\leq x-1$) and only in
those.

It is obvious that $D=[BC/[2f(y)/y^3]]\wedge(y\dotminus 1)$ is a
number the binary notation of which is obtained from the binary
notation of $BC$ with a shift to the right by $(4x-1)(x-1)-x$ digit
places and removal of all digit places but the $2x$ lowest order
ones. From this it follows that $D=2^{2x}-2^x$ or $y-D=2^x$. The
statement is proved.
\end{proof}

\begin{stat}\label{stat_expmullog_mainffoms}
It holds that
$$
2^{[\log_2 x]},\ 2^{[\log_2 x]\cdot [\log_2 y]},\ [\log_2
x]\in[\mainffoms].
$$
\end{stat}
\begin{proof}
Let $f(x)=2^{[\log_2 x]^2}$. Then using the statement
\ref{stat_ssqrt} it is easy to obtain that
$$
2^{[\log_2
x]}=\ffssqrt\left(\left[\frac{f(2x)}{2f(x)}\right]\right)\cdot\ffsg(x)+(1\dotminus\ffsg(x)),
$$
$$
2^{[\log_2 x]\cdot [\log_2
y]}=\ffssqrt\left(\left[\frac{f(2^{[\log_2 x]}\cdot 2^{[\log_2
y]})}{f(x)\cdot f(y)}\right]\right).
$$
From these formulas follows the statement that one proved for the first two functions.

Let $l=[\log_2 x]$. One can note that
$$
\left(\frac{2{l^2}\dotminus 1}{2^l\dotminus 1}\right)^2 =
\left(\sum_{i=0}^{l-1}2^{il}\right)^2 =
\sum_{i=0}^{l-1}2^{il}(i+1)+\sum_{i=l}^{2(l-1)}2^{il}(2l-1-i).
$$
By considering the binary notation of the obtained number, one arrives at the following result
$$
l=\ffrm\left(\left[\frac{2{l^2}\dotminus 1}{2^l\dotminus
1}\right]^2/\left[\frac{2^{l^2}}{2^l}\right],2^l\right).
$$
From this and from the fact that $2^l,2^{l^2}\in[\mainffoms]$
follows the statement that one was trying to prove.
\end{proof}

\subsection{The Correctness of Predicates that Correspond to $\ccfom$-formulas}

\begin{stat}\label{stat_mainffoms_log}
If $f_1,f_2\in[\mainffoms]^{\lllog}$, then $f_1+f_2,\ f_1\dotminus
f_2,\ f_1\cdot f_2\in[\mainffoms]^{\lllog}$. Besides, if
$f\in[\mainffoms]^{\lllog}$, then $f,2^f\in[\mainffoms]$.
\end{stat}
\begin{proof}
Indeed, if $f_1=[\log_2 g_1]$, $f_2=[\log_2 g_2]$, where
$g_1,g_2\in[\mainffoms]$, then
$$
f_1+f_2=[\log_2(2^{[\log_2 g_1]}\cdot 2^{[\log_2 g_2]})],
$$
$$
f_1\dotminus f_2=[\log_2[2^{[\log_2 g_1]}/2^{[\log_2 g_2]}]],
$$
$$
f_1\cdot f_2=[\log_2 2^{[\log_2 g_1][\log_2 g_2]}],
$$
Minding the statement \ref{stat_expmullog_mainffoms} these functions are in $[\mainffoms]^{\lllog}$.

The last part of the statement follows directly from the statement
\ref{stat_expmullog_mainffoms}. The statement is proved.
\end{proof}

\begin{stat}\label{stat_exppolysum}
If
$$
g(y,z)=\sum_{x<y}2^{xz}x,
$$
then $g([\log_2 y],[\log_2 z])\in[\mainffoms]$.
\end{stat}
\begin{proof}
This follows from well known formulas of summation and the statement
\ref{stat_expmullog_mainffoms}.
\end{proof}

\begin{stat}\label{stat_term_right}
If $t$ is a $\ccfom$-term, then its corresponding function
$h_t(X,y_1,\ldots,y_m)$ is correct together with the function $2^{h_t-1}$.
\end{stat}
\begin{proof}
There are three possible cases.
\begin{itemize}
\item $h_t(X,y_1,\ldots,y_m)=1$. In this case for it and for
$2^{h_t-1}$ the following function is appropriate (as the function
from the definition of correctness)
$$
f(x,z)=\sum_{1\leq y_1,\ldots,y_m\leq
l}2^{(y_1-1)l+(y_2-1)l^2+\ldots+(y_m-1)l^m}=\left[\frac{2^{l^{m+1}}\dotminus
1}{2^l\dotminus 1}\right],
$$
here and further in the proof of this statement $l$ is a contracted notation for $[\log_2 z]$. From the statement
\ref{stat_expmullog_mainffoms} it follows that
$f(x,z)\in[\mainffoms]$. \item $h_t(X,y_1,\ldots,y_m)=|X|$. In this case one can use the function
$$
f(x,z)=\sum_{1\leq y_1,\ldots,y_m\leq
l}2^{(y_1-1)l+(y_2-1)l^2+\ldots+(y_m-1)l^m}l=\left[\frac{2^{l^{m+1}}\dotminus
1}{2^l\dotminus 1}\right]\cdot l,
$$
for $2^{h_t-1}$ the suitable function is
$$
f'(x,z)=\left[\frac{2^{l^{m+1}}\dotminus 1}{2^l\dotminus
1}\right]\cdot \left[\frac{2^l}{2}\right].
$$
Further reasoning is analogous to the previous point. \item
 $h_t(X,y_1,\ldots,y_m)=y_i$. For $h_t$ the suitable function is
$$
f(x,z)=\sum_{1\leq y_1,\ldots,y_m\leq
l}2^{(y_1-1)l+\ldots+(y_m-1)l^m}y_i=$$
$$
=\left(\sum_{y_1=1}^{l}2^{(y_1-1)l}\right)
\cdot\ldots\cdot\left(\sum_{y_{i-1}=1}^{l}2^{(y_{i-1}-1)l^{i-1}}\right)\cdot\left(\sum_{y_i=1}^{l}2^{(y_i-1)l^i}y_i
\right)\cdot $$ $$\cdot \left(\sum_{y_{i+1}=1}^{l}
2^{(y_{i+1}-1)l^{i+1}}\right)\cdot \ldots \cdot
\left(\sum_{y_m=1}^{l}2^{(y_m-1)l^m}\right),
$$
with this for $2^{h_t-1}$ a suitable function is $f'(x,z)$ that can
be expressed by an analogous formula with a substitution of the
$i$-th factor by $\sum_{y_i=1}^{l}2^{(y_i-1)l^i+(y_i-1)}$. From the
statements \ref{stat_exppolysum}, \ref{stat_expmullog_mainffoms} and
the formula of the geometric progression sum, it follows that
$f,f'\in[\mainffoms]$.
\end{itemize}
The statement is proved.
\end{proof}

\begin{stat}\label{stat_log_into_functions}
The following functions are in $[\mainffoms]$ (definitions can be looked up in the section \ref{section_exp_skolem}):
$$
\ffrep(x,[\log_2 n],[\log_2 l]),
$$
$$
\ffincrx(x,[\log_2 n],[\log_2 l_1],[\log_2 l_2]),
$$
$$
\ffswap_n(x,[\log_2 q],[\log_2 k_1],\ldots,[\log_2 k_n],[\log_2
m_1],\ldots,[\log_2 m_n]),
$$
$$
\ffincr(x,[\log_2 q],[\log_2 l]),
$$
$$
\ffdecr(x,[\log_2 q],[\log_2 l]),
$$
$$
\ffnot(x,[\log_2 n]),
$$
$$
\ffor(x,y,[\log_2 n]),
$$
$$
\ffcmp(x,y,[\log_2 n],[\log_2 l]),
$$
$$
\ffcmpeq(x,y,[\log_2 n],[\log_2 l]),
$$
$$
\ffsum(x,[\log_2 n],[\log_2 l],[\log_2 k]).
$$
\end{stat}
\begin{proof}
It follows from the formulas for these functions (see section
\ref{section_exp_skolem}) and statement \ref{stat_mainffoms_log}.
\end{proof}

One can assume that
$$
\ffreverse(x,n)=\ffdecr\left(\left[\frac{\ffrep(x,n,n)\wedge
\left[\frac{2^{n^2\dotminus 1}\dotminus 2^{n\dotminus
1}}{2^{n\dotminus 1}\dotminus 1}\right]}{2^{n\dotminus
1}}\right],n,n\dotminus 1\right).
$$
One can notice that $\ffreverse(x,[\log_2 n])\in[\mainffoms]$ (it follows from the statements \ref{stat_mainffoms_log} and
\ref{stat_log_into_functions}).
\begin{stat}\label{stat_reverse}
If $a_0\ldots a_{n-1}$ is a binary notation of $x$, then
$a_{n-1}\ldots a_0$ is a binary notation of $\ffreverse(x,n)$ (the
binary notation can have any number of zeroes on the right).
\end{stat}
\begin{proof}
One can contend that (statement \ref{statrep}):
$$
\ffrep(x,n,n)=\langle
\underbrace{a_{n-1},\ldots,a_0,\ldots,a_{n-1},\ldots,a_0}_{\text{$n$
times}};\ 1\rangle,
$$
using the geometric progression sum formula
$$
\left[\frac{2^{n^2\dotminus 1}\dotminus 2^{n\dotminus
1}}{2^{n\dotminus 1}\dotminus
1}\right]=\sum_{i=0}^{n-1}2^{(n-1)(i+1)}.
$$
From this it follows that
$$
\left[\left(\ffrep(x,n,n)\wedge \left[\frac{2^{n^2\dotminus
1}\dotminus 2^{n\dotminus 1}}{2^{n\dotminus 1}\dotminus
1}\right]\right)/2^{n\dotminus 1}\right]=\langle
a_0,a_1,\ldots,a_{n-1};\ n-1\rangle.
$$
Based on this and the claim \ref{statdecr} one obtains the claim that one was trying to prove.
\end{proof}

\begin{stat}\label{stat_elementary_right}
If $\Phi$ is an elementary $\ccfom$-formula, then its corresponding predicate  $\rho_{\Phi}(X,y_1,\ldots,y_m)$ is correct.
\end{stat}
\begin{proof}
There can be three cases.
\begin{itemize}\item $\Phi$ is $t_1\leq t_2$, where $t_1$, $t_2$
are $\ccfom$-terms. Let those terms correspond to $\tth$-functions
$h_1(X,y_1,\ldots,y_m)$ and $h_2(X,y_1,\ldots,y_m)$ respectively.
From the statement \ref{stat_term_right} it follows that the functions
$h_1,h_2$ are correct. Let the functions
$f_1(x,z),f_2(x,z)\in[\mainffoms]$ correspond to them. Let
$$
f(x,z)=\ffcmp(f_2(x,z),f_1(x,z),[\log_2 z]^m,[\log_2 z]).
$$
From the statements \ref{stat_mainffoms_log} and
\ref{stat_log_into_functions} it follows that $f\in[\mainffoms]$.

Besides, for any $X\in\{0,1\}^+$, $l=|X|$ if $x$ is a number in binary notation $X$, then it holds that
$$
f_i(x,2^l)=\langle
h_i(1,1,\ldots,1),h_i(2,1,\ldots,1),\ldots,h_i(l,\ldots,l);\
l\rangle,
$$
$i=1,2$ (vectors in reverse lexicographical order), thus (the
statement \ref{statcmp}) it holds true that
$$
\ffcmp(f_2(x,2^l),f_1(x,2^l),l^m,l)$$
$$=\langle\sigma(1,1,\ldots,1),\sigma(2,1,\ldots,1),\ldots,\sigma(l,\ldots,l);\
l\rangle,
$$
where
$$
\sigma(y_1,\ldots,y_m)=\begin{cases}1,\text{ if
$h_1(X,y_1,\ldots,y_m)\leq h_2(X,y_1,\ldots,y_m)$,} \\
0\text{ otherwise.}\end{cases}
$$
From this it follows that $f(x,z)$ complies with the definition of
correctness for a predicate $(h_1\leq h_2)$. \item $\Phi$ is
$\llbit(t_1,t_2)$. Let the function $h_i(X,y_1,\ldots,y_m)$
correspond to a term $t_i$ ($i=1,2$), $f_1(x,z)$, $f_2(x,z)$ are the
functions from the definitions of correctness for $h_1$ and
$2^{h_2-1}$ respectively. Put
$$
f'(x,z)=\ffcmpeq(f_1(x,z)\wedge f_2(x,z),0,[\log_2 z]^m,[\log_2
z]).
$$
Analogously to the previous point one obtains the fact that when
satisfying similar conditions over $X\in\{0,1\}^+$ and $l,x\in\nat$
the following takes place
$$
f'(x,2^l)=\langle\sigma(1,1,\ldots,1),\sigma(2,1,\ldots,1),\ldots,\sigma(l,\ldots,l);\
l\rangle,
$$
where
$$
\sigma(y_1,\ldots,y_m)=\begin{cases}1,\text{ if
$h_1(X,y_1,\ldots,y_m)\wedge 2^{h_2(X,y_1,\ldots,y_m)-1}=0$,} \\
0\text{ otherwise.}\end{cases}
$$
It is easy to notice that $\sigma(y_1,\ldots,y_m)=1$ if and only if
$$h_1(X,y_1,\ldots,y_m)\langle
h_2(X,y_1,\ldots,y_m)-1\rangle=0.$$ Consequently,
$$
f(x,z)=\ffrep(1,[\log_2 z]^m,[\log_2 z])\dotminus f'(x,z)
$$
fits the definition of correctness for $\rho_{\Phi}$
($f\in[\mainffoms]$ follows from statements
\ref{stat_mainffoms_log}, \ref{stat_log_into_functions}). \item
$\Phi$ is $X\langle t\rangle$, term $t$ corresponds to
$h_t(X,y_1,\ldots,y_m)$, $f_0(x,z)$ is a function from the definition of correctness for $2^{h_t-1}$. Function $f(x,z)$ is defined analogously to the previous point with a substitution $f_1(x,z)$ to
$$\ffrep(\ffreverse(x,[\log_2 z]),[\log_2 z]^m,[\log_2 z]),$$
$f_2(x,z)$ to $f_0(x,z)$. The fact that $f(x,z)$ fits the definition of correctness for $\rho_{\Phi}$, follows from the statements
\ref{statrep}, \ref{stat_reverse} and reasoning analogous to the previous point.
\end{itemize}
The statement is proved.
\end{proof}

\begin{stat}\label{stat_permvar_right}
Let $h(X,y_1,\ldots,y_m)$ be a correct $\tth$-function,
$h'(X,y_1,\ldots,y_m)$ that one obtains from $h$ by permuting variables
$y_1,\ldots,y_m$. Then $h'$ is also a correct function.
\end{stat}
\begin{proof}
Let
$$
h(X,y_1,\ldots,y_m)=h'(X,y_{j_1},\ldots,y_{j_m}),
$$
where $(j_1,\ldots,j_m)$ is some permutation of numbers
$1,\ldots,m$. Besides, let $f(x,z)\in[\mainffoms]$ be a function from the definition of correctness for
$h$,
$$
f'(x,z)=\ffswap(f(x,z),l,l,l^2,\ldots,l^m,l^{j_1},l^{j_2},\ldots,l^{j_m}),
$$
where $l=[\log_2 z]$. One will prove that $f'(x,z)$ complies with
the definition of correctness for $h'$. Let $X\in\{0,1\}^+$, $x$ be
a number in binary notation $X$, $l=|X|$. Thus,
\begin{normalsize}
$$
f'(x,2^l)=$$ $$=\ffswap\left(\sum_{1\leq y_1,\ldots,y_m\leq
l}2^{(y_1-1)l+\ldots+(y_m-1)l^m}h(X,y_1,\ldots,y_m),\ l,\
l,l^2,\ldots,l^m,l^{j_1},\ldots,l^{j_m}\right)
$$
$$
=\sum_{1\leq y_1,\ldots,y_m\leq
l}2^{(y_1-1)l^{j_1}+\ldots+(y_m-1)l^{j_m}}h(X,y_1,\ldots,y_m)$$
$$=\sum_{1\leq u_1,\ldots,u_m\leq
l}2^{(u_1-1)l+\ldots+(u_m-1)l^m}h'(X,u_1,\ldots,u_m),
$$
\end{normalsize}
the second equality follows from the statement \ref{statswap}. From statements
 \ref{stat_mainffoms_log} and
\ref{stat_log_into_functions} it follows that $f'\in[\mainffoms]$.
The statement is proved.
\end{proof}

\begin{stat}\label{stat_formula_right}
If $\Phi$ is a $\ccfom$-formula, then its corresponding predicate
$\rho_{\Phi}(X,y_1,\ldots,y_m)$ is correct.
\end{stat}
\begin{proof}
One can prove this by induction on construction of a formula. Let
this formula $\Phi$ have a corresponding predicate
$\varphi(X,y_1,\ldots,y_m)$. There are the following cases.
\begin{itemize}
\item $\Phi$ is an elementary $\ccfom$-formula. Then this follows from the statement \ref{stat_elementary_right}. \item $\Phi$ is
$(\qqm y_i)(\Psi)$, to formula $\Psi$ it corresponds the correct
predicate $\psi(X,y_1,\ldots,y_m)$, $f_{\psi}(x,z)$ is the function
from the definition of correctness for $\psi$ (more specifically,
from the definition of correctness from the characteristic function
of $\psi$). Minding statement \ref{stat_permvar_right} one can
suppose that $i=1$. One can assume that
$$
g(x,z)=\ffsum(f_{\psi}(x,z),[\log_2 z]^{m-1},[\log_2 z],[\log_2
z]),
$$
$$
p(x,z)=\ffrep([[\log_2 z]/2]+1,[\log_2 z]^{m-1},[\log_2 z]^2),
$$
$$
r(x,z)=\ffcmp(g(x,z),p(x,z),[\log_2 z]^{m-1},[\log_2 z]^2),
$$
$$
f_{\varphi}(x,z)=\left[\frac{2^{[\log_2 z]^2}\dotminus
1}{2^{[\log_2 z]}\dotminus 1}\right]\cdot r(x,z).
$$
One can prove that $f_{\varphi}$ fits the definition of correctness
for $\varphi$. Let $X\in\{0,1\}^+$, $l=|X|$, $x$ is the number with
binary notation $X$ (perhaps completed with zeroes from the left).
From the statement \ref{statsum} it follows that
$$
g(x,2^l)=\langle
s(1,1,\ldots,1),s(2,1,\ldots,1),\ldots,s(l,l,\ldots,l);\
l^2\rangle,
$$
where $s(y_2,\ldots,y_m)$ is the number of $y_1$ such that $1\leq
y_1\leq l$ and $\psi(X,y_1,\ldots,y_m)$ hold true. From this and the statements \ref{statrep}, \ref{statcmp} it follows that
$$
r(x,2^l)=\langle
v(1,1,\ldots,1),v(2,1,\ldots,1),\ldots,v(l,l,\ldots,l);\
l^2\rangle,
$$
where
$$
v(y_2,\ldots,y_m)=\begin{cases}1,\text{ if $s(y_2,\ldots,y_m)>l/2$,}
\\ 0\text{ otherwise}.\end{cases}
$$
Using the geometric progression sum formula, one obtains
$$
f_{\varphi}(x,2^l)=$$ $$=\left(\sum_{y_1=1}^{l}2^{(y_1-1)
l}\right)\cdot\left(\sum_{1\leq y_2,\ldots,y_m\leq
l}2^{(y_2-1)l^2+\ldots+(y_m-1)l^m}v(y_2,\ldots,y_m)\right)
$$
$$
=\sum_{1\leq y_1,y_2,\ldots,y_m\leq
l}2^{(y_1-1)l+\ldots+(y_m-1)l^m}\chi_{\varphi}(X,y_1,y_2,\ldots,y_m).
$$
$f_{\varphi}\in[\mainffoms]$ follows from the statements
\ref{stat_mainffoms_log} and \ref{stat_log_into_functions}. \item
$\Phi$ looks like $(\exists y_i)(\Psi)$ or $(\forall y_i)(\Psi)$. It
is considered analogously to the previous point. \item $\Phi$ looks
like $(\Psi_1\vee\Psi_2)$. Let $\psi_i$ be the correct predicate
that corresponds to $\Psi_i$, $f_i(x,z)$ is the function from the
definition of correctness for $\psi_i$ ($i=1,2$). It is easy to
notice that
$$
f(x,z)=\ffor(f_1(x,z),f_2(x,z),[\log_2 z]^{m+1})
$$
complies with the definition of correctness for $\varphi$ (see statements
\ref{stat_mainffoms_log},
\ref{stat_log_into_functions}, \ref{statlogicfunc}). \item $\Phi$
looks like $(\Psi_1\&\Psi_2)$ or $(\neg\Psi)$. Should have treatment analogous to the previous point, except instead of the function $\ffor$ one needs to use
 $x\wedge y$ or $\ffnot$ respectively.
\end{itemize}
The statement is proved.
\end{proof}

\subsection{Proof of Theorem \ref{theorem_main_ffom}}

\begin{stat}\label{stat_code_mainffoms}
For any $n,k\geq 1$ the functions $\ffintcode_{n,k}$ and
$\ffintlcode_{n,k}$, such that for any $x_1,\ldots,x_n$ it satisfies
$$
\ffintlcode_{n,k}(x_1,\ldots,x_n)=|\ffcodexd_k(x_1,\ldots,x_n)|
$$
and the binary notation of $\ffintcode_{n,k}(x_1,\ldots,x_n)$
(possibly completed by zeroes from the left) is
$\ffcodexd_k(x_1,\ldots,x_n)$, belong to $[\mainffoms]$ and
$[\mainffoms]^{\lllog}$ respectively.
\end{stat}
\begin{proof}
From $\ffsg(x)=[\log_2(2\ffsg(x))]$, (\ref{eq_len}) and the statement
\ref{stat_mainffoms_log} it follows that
$\fflen(x)\in[\mainffoms]^{\lllog}$.

Based on this and the statement \ref{stat_mainffoms_log} and by noticing that
$$
\ffintlcode_{n,k}(x_1,\ldots,x_n)=2^{k+1}(\fflen(x_1)+\ldots+\fflen(x_n)+n+1)^k,
$$
one obtains that $\ffintlcode_{n,k}\in[\mainffoms]^{\lllog}$.

Let $f(x)=3\cdot \ffincr(x,[\log_2 x]+1,2)$. One can notice that for
any $x$ it is true that if the binary notation of $x$ is
$a_1,\ldots,a_m$, then the binary notation of $f(x)$ is $a_1 a_1 a_2
a_2 \ldots a_m a_m$ (this follows from the statement
\ref{statincr}).

Further, one can notice that
$$
\ffintcode_{n,k}(x_1,\ldots,x_n)=\sum_{i=1}^{n}(2^{l\dotminus(2i+l_1+\ldots+l_{i-1})}+2^{l\dotminus(2i+l_1+\ldots+l_i)}f(x_i))+1,
$$
where $l$ is a contracted notation for
$\ffintlcode_{n,k}(x_1,\ldots,x_n)$ and $l_i$ is for
$2\cdot\fflen(x_i)$ ($1\leq i\leq n$).

From the statements \ref{stat_mainffoms_log} and
\ref{stat_log_into_functions} it follows that
$\ffintcode_{n,k}\in[\mainffoms]$.
\end{proof}

\begin{stat}\label{stat_ffomvar_mainffoms}$\ccffomvar\subseteq[\mainffoms]$.
\end{stat}
\begin{proof}
Let $f(x_1,\ldots,x_n)\in\ccffomvar$, $k$
 is the number from the definition $\ccffomvar$ for $f$, $\rho(X,z_1,\ldots,z_m,y)$
is the predicate from the definition that corresponds to the number
$k$. From the statement \ref{stat_formula_right} it follows that
$\rho$ is correct. Let $g(x,z)$ be the function from the definition
of correctness for $\rho$ (specifically, from the definition of
correctness for the characteristic function of $\rho$).

Let one assume that
$$
r(\tilde{x})=\left[\frac{g(\ffintcode_{n,k}(\tilde{x}),2^l)}{[(2^{l^{m+1}}\dotminus
1)/(2^l\dotminus 1)]}\right],
$$
where $l$ is a contraction for $\ffintlcode_{n,k}(x_1,\ldots,x_n)$,
the functions $\ffintcode_{n,k}$ and $\ffintlcode_{n,k}$ were taken
from the statement \ref{stat_code_mainffoms}.

One can notice that (see definitions of $\ccffomvar$ and the one
concerning correctness)
$$
g(\ffintcode_{n,k}(\tilde{x}),2^l)=$$ $$=\sum_{1\leq
z_1,\ldots,z_m,y\leq
l}2^{(z_1-1)l+\ldots+(z_m-1)l^m+(y-1)l^{m+1}}\chi_{\rho}(\ffcodexd_k(\tilde{x}),z_1,\ldots,z_m,y)=
$$
$$
=\sum_{1\leq z_1,\ldots,z_m,y\leq
l}2^{(z_1-1)l+\ldots+(z_m-1)l^m+(y-1)l^{m+1}}f(\tilde{x})\langle
y-1\rangle =$$ $$=\left[\frac{2^{l^{m+1}}\dotminus 1}{2^l\dotminus
1}\right]\cdot \sum_{y=0}^{l-1}2^{l^{m+1}y}f(\tilde{x})\langle
y\rangle.
$$
Consequently,
$$
r(\tilde{x})=\sum_{y=0}^{l-1}2^{l^{m+1}y}f(\tilde{x})\langle
y\rangle=\langle f(\tilde{x})\langle
0\rangle,\ldots,f(\tilde{x})\langle l-1 \rangle;\ l^{m+1}\rangle,
$$
that is
$$
\ffdecr(r(\tilde{x}),l,l^{m+1})=\langle f(\tilde{x})\langle
0\rangle,\ldots,f(\tilde{x})\langle l-1 \rangle;\
1\rangle=f(\tilde{x}).
$$
From the statements \ref{stat_mainffoms_log},
\ref{stat_log_into_functions} and \ref{stat_code_mainffoms} it follows that
 $f(\tilde{x})\in[\mainffoms]$. The statement is proved.
\end{proof}
Based on equivalent definitions of class $\ccfom$ from
~\cite{uniformity} one can easily obtain the following statement.
\begin{stat}\label{stat_ffom_closed}
$\ccffom$ is closed with respect to superposition.
\end{stat}
\noindent \textbf{Proof of theorem \ref{theorem_main_ffom}}.
Based on statements \ref{stat_ffom_ffomalt} and
\ref{stat_ffomalt_ffomvar} one obtains that $\ccffom=\ccffomvar$.
Besides from the statements \ref{statarythmffom} and
\ref{stat_ffom_closed} it follows that
$$
\left[x+y,\quad x\dotminus y,\quad x\wedge y,\quad
\left[x/y\right],\quad x^{[\log_2 y]} \right]\subseteq\ccffom.
$$
One can note that $2^{[\log_2 x]^2}$=$((x^{[\log_2 (x\dotminus
x)]}+x^{[\log_2 (x\dotminus x)]})^{[\log_2 x]})^{[\log_2 x]}$. Based on all of that and the statement
 \ref{stat_ffomvar_mainffoms},
one has the following consequence
$$
\ccffom=\ccffomvar\subseteq[\mainffoms]\subseteq $$
$$\subseteq\left[x+y,\quad x\dotminus y,\quad x\wedge y,\quad
\left[x/y\right],\quad x^{[\log_2 y]} \right]\subseteq\ccffom.
$$
Theorem is proved.

\section{Hierarchies of Classes that are Exhaustive with Regard to Kalmar Elementary Functions and Formulas of an Arbitrary Height}\refstepcounter{sectionref}\label{section_hierarchy}

\subsection{Definitions}

Let one define the class $\cckhsplus^n$ as a set of all functions
$f(\tilde{x}),$ for each of which there exist functions
$g(\tilde{x},y)\in\ccxs$ and $m(\tilde{x})\in\ccexppoly^n$ such that
$$
f(\tilde{x})=g(\tilde{x},m(\tilde{x})).
$$

\subsection{Coincidence of Classes $\cckhs^n$ and $\cckhsplus^n$}

\begin{stat}\label{statmedin}
For any functions
$f_1(\tilde{x}),\ldots,f_k(\tilde{x})\in\cckhs^n$ there exist functions $m(\tilde{x})\in\ccexppoly^n$ and
$g_1(\tilde{x},y,z),\ldots,g_k(\tilde{x},y,z)\in\ccclh$ such that
$$
f_i(\tilde{x})\langle y\rangle = g_i(\tilde{x},y,m(\tilde{x})),
$$
$$
f_i(\tilde{x})<2^{m(\tilde{x})}
$$
for any $\tilde{x},$ $y,$ $i=1,2,\ldots,k,$ the functions $g_i$ only
takes values $0$ and $1$, and $g_i(\tilde{x},y,z)=0$ for $y\geq z.$
\end{stat}
\begin{proof}
From the definitions of $\ccexppoly^{n+1}$ and $\cckhs^n$ it follows
that there exist functions $g'_i(\tilde{x},y,z)\in\ccclh,$
$m_1^i(\tilde{x}),m_2^i(\tilde{x})\in\ccexppoly^n$ such that
$f_i(\tilde{x})\langle y\rangle=g'_i(\tilde{x},y,m_1^i(\tilde{x})),$
$f_i(\tilde{x})<2^{m_2^i(\tilde{x})}$ ($i=1,\ldots,k$). Let one pick
a function $m(\tilde{x})\in\ccexppoly^n,$ dominating $m_1^i$ and
$m_2^i$ for all $i.$ Let one assume
$$
g_i(\tilde{x},y,z)=\min(g'_i(\tilde{x},y,\min(m_1^i(\tilde{x}),z)),1)\cdot\chi_{y<z}(y,z),\quad
i=1,\ldots,k.
$$
It is obvious that $g_i\in\ccclh$ and functions $g_i,$ $m$ satisfy the conditions. The statement is proved.
\end{proof}

\begin{stat}$\cckhs^n\subseteq\cckhsplus^n$ for any $n.$\end{stat}
\begin{proof}
Let $f(\tilde{x})\in\cckhs^n.$ One can prove that
$f(\tilde{x})\in\cckhsplus^n.$ According to the statement
\ref{statmedin} one can choose functions $m(\tilde{x})$ and
$g(\tilde{x},y,z).$

Let one assume
$$
h(\tilde{x},z)=\Sum_{y=0}^{z-1} g(\tilde{x},y,z)\cdot 2^y.
$$
From the fact that $g(\tilde{x},y,z)=0$ for $y\geq z,$ it follows that it is true that
$$
h(\tilde{x},z)\langle y\rangle = g(\tilde{x},y,z).
$$
Besides, $h(\tilde{x},z)<2^z,$ thus $h\in\ccxs.$ From the fact that
$g(\tilde{x},y,m(\tilde{x}))=f(\tilde{x})\langle y\rangle$ for any
$\tilde{x}$ and $y$, it follows that  there is
$$
h(\tilde{x},m(\tilde{x}))=f(\tilde{x}).
$$
The statement is proved.
\end{proof}
\begin{stat}$\cckhsplus^n\subseteq\cckhs^n$ for any $n.$\end{stat}
\begin{proof}
Let $f(\tilde{x})=g(\tilde{x},m(\tilde{x})),$ $g\in\ccxs,$
$m\in\ccexppoly^n.$ One has
$$
f(\tilde{x})\langle y\rangle = g(\tilde{x},m(\tilde{x}))\langle
y\rangle,
$$
$g(\tilde{x},z)\langle y\rangle\in\ccclh$ (by definition of
$\ccxs$).

The boundedness of $f$ by some function from $\ccexppoly^{n+1}$ is
obvious. The statement is proved.
\end{proof}
\begin{conseq}
$\cckhs^n=\cckhsplus^n$ for any $n.$
\end{conseq}

\subsection{Proof of Theorem \ref{theorem_main_h}}

\begin{stat}\label{stattwopowerxsn}
If $f\in\cckhs^n,$ then $2^f\in\cckhs^{n+1}.$
\end{stat}
\begin{proof}
Let one choose according to the statement \ref{statmedin} functions $m(\tilde{x})$ and
$g(\tilde{x},y,z)$ for $f(\tilde{x}).$ It is obvious that
$$
2^{f(\tilde{x})}\langle y\rangle=\chi_{\rho}(\tilde{x},y),
$$
where
$$
\rho(\tilde{x},y)\equiv(y=f(\tilde{x}))\equiv(\forall t)(y\langle
t\rangle=f(\tilde{x})\langle t\rangle)\equiv(\forall t)(y\langle
t\rangle=g(\tilde{x},t,m(\tilde{x})))\equiv
$$
$$
\equiv(y<2^{m(\tilde{x})})\&(\forall t)_{t<m(\tilde{x})}(y\langle
t\rangle=g(\tilde{x},t,m(\tilde{x}))).
$$
The last equality follows from the fact that
$f(\tilde{x})<2^{m(\tilde{x})}$ and from the fact that
$g(\tilde{x},y,z)=0$ for $y\geq z.$ Let
$$
\varphi(\tilde{x},y,z)\equiv(y<z)\&(\forall t)_{t<[\log_2
z]}(y\langle t\rangle=g(\tilde{x},t,[\log_2 z])).
$$
It is obvious that $\varphi\in\ccclh_*$ and
$\rho(\tilde{x},y)\equiv\varphi(\tilde{x},y,2^{m(\tilde{x})}).$ From this it follows that
$$
2^{f(\tilde{x})}\langle
y\rangle=\chi_{\varphi}(\tilde{x},y,2^{m(\tilde{x})}),
$$
$\chi_{\varphi}\in\ccclh.$ The boundedness of $2^f$ by some function
from $\ccexppoly^{n+2}$ is obvious. The statement is proved.
\end{proof}

\begin{stat}
$[\mains]_{x^y}^{n+1}\subseteq\cckhs^n$ for any $n.$
\end{stat}
\begin{proof}
It is obvious that to prove this one needs to prove the following statements:
\begin{enumerate}
\item If $f(\tilde{x})\in\cckhs^n,$ then
$f(\tilde{x})\in\cckhs^{n+1}.$ Let one choose for $f(\tilde{x})$ the
functions $m(\tilde{x}),$ $g(\tilde{x},y,z)$ according to the
statement \ref{statmedin}. One gets
$$
f(\tilde{x})\langle
y\rangle=g(\tilde{x},y,m(\tilde{x}))=h(\tilde{x},y,2^{m(\tilde{x})}),
$$
where
$$
h(\tilde{x},y,z)=g(\tilde{x},y,[\log_2 z]).
$$
It is obvious that $h\in\ccclh.$

Compliance with the restrictions on the speed of growth is obvious.
 \item If $f\in\cckhs^n,$ $g$ is obtained from $f$ by
permutation, identification of variables or introduction of dummy variables, then $g\in\cckhs^n.$ This is obviously satisfied. \item
If $h(y_1,\ldots,y_k)\in\mains$ and
$f_1(\tilde{x}),\ldots,f_k(\tilde{x})\in\cckhs^n,$ then
$h(f_1(\tilde{x}),\ldots,f_k(\tilde{x}))\in\cckhs^n.$ From the fact that
$\cckhs^n=\cckhsplus^n$ it follows that
$f_1,\ldots,f_k\in\cckhsplus^n.$ I.e.
$$
f_i(\tilde{x})=g_i(\tilde{x},m(\tilde{x})),\quad i=1,\ldots,k,
$$
where $g_i\in\ccxs,$ $m\in\ccexppoly^n$ (function $m$ can be assumed to be unified for all $f_i,$ see the proof of the statement \ref{statmedin}).
Thereby,
$$
h(f_1(\tilde{x}),\ldots,f_k(\tilde{x}))=h(g_1(\tilde{x},m(\tilde{x})),\ldots,g_k(\tilde{x},m(\tilde{x}))).
$$
It is obvious that  $h(g_1(\tilde{x},z),\ldots,g_k(\tilde{x},z))\in\ccxs$
(see ~\cite{diplom}). Thus,
$h(f_1,\ldots,f_k)\in\cckhsplus^n$ (and $\cckhs^n$). \item If
$f\in\cckhs^{n+1}$ and $g\in\cckhs^n,$ then $f^g\in\cckhs^{n+1}.$
One has
$$
f^g=f^{[\log_2 2^g]}.
$$
From the statement \ref{stattwopowerxsn} it follows that
$2^g\in\cckhs^{n+1}.$ Both from this and from that $x^{[\log_2
y]}\in\ccffom$ it follows that $f^g\in\cckhs^{n+1}$ (see previous point and ~\cite{diplom}).
\end{enumerate}
Statement is proved.
\end{proof}

\begin{stat}\label{stat_h_main}
$\cckhs^n\subseteq [\mains]_{2^x}^{n+1}$ for any $n.$
\end{stat}
\begin{proof}
Let $f(\tilde{x})\in\cckhs^n.$ Then
$f(\tilde{x})\in\cckhsplus^n,$ i.e.
$$
f(\tilde{x})=g(\tilde{x},m(\tilde{x})),
$$
where $g\in\ccxs,$ $m\in\ccexppoly^n.$ From ~\cite{diplom} it
follows that $g\in[\mains]_{2^x}^{1}.$ Besides it is obvious that
$m\in[\mains]_{2^x}^n.$ From all of the above and the definition of
$[\mains]_{2^x}^{n+1}$ follows the proof of the given statement.
\end{proof}
\begin{conseq}
$\cckhs^n=[\mains]_{2^x}^{n+1}=[\mains]_{x^y}^{n+1}$ for any
$n.$
\end{conseq}
\begin{proof}
Indeed, it is clear that $[\mains]_{2^x}^{n+1}\subseteq
[\mains]_{x^y}^{n+1}.$ Besides, according to the proof above
$$
\cckhs^n\subseteq [\mains]_{2^x}^{n+1}
$$
and
$$
[\mains]_{x^y}^{n+1}\subseteq\cckhs^n.
$$
From this it follows the proof of the given statement.
\end{proof}

\begin{stat}\label{stat_ffomh_one}
For any $n$ there is an inclusion
$\ccffomh^{n+1}\subseteq\cckhs^n$.
\end{stat}
\begin{proof}
Let $f(\tilde{x})\in\ccffomh^n$, $\rho(\tilde{x},y,z)\in\ccfomn$,
$m(\tilde{x})\in\ccexppoly^{n+1}$ be the predicate and the function
from the definition of $\ccffomh^{n+1}$ for $f$. Let one assume
$\rho'(\tilde{x},y,z)\equiv\rho(\tilde{x},y,2^z)$. One can note that
$\chi_{\rho'}\in\ccclh$ (see ~\cite{diplom}). Let
$m(\tilde{x})=2^{m'(\tilde{x})}$, where $m'\in\ccexppoly^n$. One
gets (for any $\tilde{x}$, $y$)
$$
(f(\tilde{x})\langle
y\rangle=1)\equiv\rho(\tilde{x},y,m(\tilde{x}))\equiv\rho(\tilde{x},y,2^{m'(\tilde{x})})
\equiv \rho'(\tilde{x},y,m'(\tilde{x})).
$$
Thereby, $\chi_{\rho'}$ and $m'$ fit the definition of $\cckhs^n$
for $f$. The statement is proved.
\end{proof}

\begin{stat}\label{stat_minexp_ffom}$\min(2^x,z)\in\ccffom$.
\end{stat}
\begin{proof}
This follows from the fact that
$$
(\min(2^x,z)\langle y\rangle=1)\equiv \begin{cases}(x=y),\text{
if $[\log_2 z]\geq x$,} \\ (z\langle y\rangle=1)\text{ otherwise}\end{cases}
$$
(see equivalent definitions of the class $\ccfom$ from
~\cite{uniformity}, for example based on a sequence of boolean
circuits generated by Turing machine).
\end{proof}

\begin{stat}\label{stat_ffomh_two}For any $n$ it satisfies
$[\mains]_{2^x}^{n+1}\subseteq\ccffomh^{n+1}$.
\end{stat}
\begin{proof}
Let $f(\tilde{x})\in[\mains]_{2^x}^{n+1}$. From the definition of
$[\mains]_{2^x}^{n+1}$ it follows that $f$ can be expressed in terms
of a formula over functions in $\mains$ and the function $2^x$ (in
the formula it is allowed only a substitution of functions and
variables into functions), such that for every subformula there is a
corresponding function bounded by the function from
$\ccexppoly^{n+1}$. Let $m(\tilde{x})\in\ccexppoly^{n+1}$ be the
function that bounds all of these functions. Let $g(\tilde{x},z)$ be
the function that can be expressed with this formula with the
replacement of every subformula of the type $2^F$ by $\min(2^F,z)$.
One can note that $g$ can be obtained from the superposition of
functions from $\ccffom$ (see statements \ref{statarythmffom} and
\ref{stat_minexp_ffom}). Thereby, from the statement
\ref{stat_ffom_closed} it follows that $g(\tilde{x},z)\in\ccffom$.

One gets (for any $\tilde{x},y$)
$$
(f(\tilde{x})\langle
y\rangle=1)\equiv(g(\tilde{x},m(\tilde{x}))\langle y\rangle=1).
$$
From the definition of $\ccffom$ it follows that $(g(x,z)\langle
y\rangle=1)\in\ccfomn$. Thus, $f\in\ccffomh^{n+1}$. The statement is
proved.
\end{proof}

From the consequence of the statement \ref{stat_h_main} and the statements \ref{stat_ffomh_one}, \ref{stat_ffomh_two} follows the claim of the theorem \ref{theorem_main_h}.

\chapter{Simple Basis by Superposition in the Class $\bigeps^2$ of Grzegorczyk
Hierarchy }\refstepcounter{chapterref}\label{chapter_gzhbasis}

\section{Minsky Machines }

Basic definitions can be looked up in sections \ref{subsection_definitions} and
\ref{subsection_basic_gzhbasis} of the introduction.

\emph{Minsky Machine} there is a multitape non-erasing Turing machine that has a finite number of one-sided, right side infinite tapes, the end cells of which contain symbol 1 and the rest of them contain 0 (see ~\cite{maltsev,minsky}); every tape has one reading head per each, at every step of its work the heads of the Minsky machines can move independently from each other by one cell to the left, right or remain in the same cell. The program of the machine is organized in such a way that the heads cannot move away from the cells that contain symbol 1.

One assumes that the Minsky machine $M,$ that has not less then $n$
tapes calculates everywhere defined function
$f(x_1,\,\ldots,\,x_n),$ if for every $x_1,\,\ldots,\,x_n$ it
satisfies the following conditions. If at the beginning of
calculating process first $n$ machine heads are in the cells with
numbers $x_1,\,\ldots,\,x_n$ respectively (the end cells are number
0) and the rest of those heads are in the end cells, then at the
final step of calculation (when the machine $M$ reaches its final
stage) the first head is going to be in the cell number
$f(x_1,\,\ldots,\,x_n).$

The time of calculation (the number of steps that the machine executes) in this case is labeled $T_M(x_1,\,\ldots,\,x_n).$

Let machine $M$ have $s$ inside states. These states can be marked
with numbers $0,\,1,\,\ldots,\,s-1.$ One can assume that 1 is the
initial state and 0 is the final one. The program of $k$-tape Minsky
machine $M$ consists of commands of the form
$$
e_1 \ldots e_k q \rightarrow d_1 \ldots d_k q',
$$
where
$$
e_1,\ldots,e_k \in \{0,\,1\},\quad q,q'\in\{0,1,\ldots,s-1\},\quad
q\neq 0,\quad d_1,\ldots,d_k\in \{-1,0,1\}
$$
and $e_i=1$ implies $d_i\neq -1$. The given command means that if
the machine $M$ at some point in time $t$ is in the state $q$ and
the vector that is being read by the heads is $(e_1 \ldots e_k),$
then at the moment $t+1$ the machine $M$ goes into the state $q'$
and the head with the number $i\ (1\leq i\leq k)$ moves to the left
by one cell
 ($d_i=-1$), to the right ($d_i=1$) or remains in the same cell
($d_i=0$).

\emph{Configuration} of $k$-tape Minksy machine $M$ at the moment of
time $t$ will be the tuple $(x_1,\,\ldots,\,x_k;\ q),$ where $x_i$
is the cell number, in which there is $i$-th head $(1\leq i\leq k)$,
$q$ is the inside state of the machine $M$ at the time $t.$

There is the following description of the class $\bigeps^2$ in terms
of Minsky machines calculations.
\begin{theorem1}[\cite{basis,ogr_rec}]
$\bigeps^2$ is the set of all functions that can be calculated on
Minsky machines in polynomial time. In other words, everywhere
defined function $f(x_1,\,\ldots,\,x_n)$ belongs to the class
$\bigeps^2$ if and only if there exists a Minsky machine $M$ and a
polynomial $t(x_1,\,\ldots,\,x_n)$ with natural coefficients such
that the machine $M$ calculates the function $f$ and for any
$x_1,\,\ldots,\,x_n$ it satisfies the inequality
$$
T_M(x_1,\,\ldots,\,x_n)\leq t(x_1,\,\ldots,\,x_n).
$$
\end{theorem1}

The Minsky machine is called \emph{reduced} if at any state $q$ it
can read information from only one tape (for every $q,$ generaly
speaking its own). The program of a reduced $k$-tape Minsky machine
with $s$ states consists of $s-1$ commands of the form
$$
q\rightarrow i;\ d_1^0 \ldots d_k^0 q^0;\ d_1^1 \ldots d_k^1 q^1,
$$
where
$$
q\in \{1,2,\ldots,s-1\},\quad 1\leq i \leq k,\quad
d_1^0,\ldots,d_k^0,d_1^1,\ldots,d_k^1\in\{-1,0,1\}, \,
$$
$q^0,q^1\in\{0,1,\ldots,s-1\}$. The given command signifies that if
at some point of time $t$ the machine is at state $q$ and $i$-th
head reads the number $e,$ then at the time $t+1$ the machines moves
to the state $q^e$ and the head with number $j\ (1\leq j\leq k)$
moves one cell to the left ($d_j^e=-1$), to the right ($d_j^e=1$) or
remains at the same spot ($d_j^e=0$).

Obviously one step of the work of a generic Minsky machine $M$ can
be modeled with $k$ steps of execution of a fitting reduced Minsky
machine $M',$ if each state of the machine $M$ will be represented
in $M'$ as $2^k$ states that "remember"\ binary tuples of length
$k$. Thereby, the following statement holds true.
\begin{statse}\label{e_2_prived}
Let Minsky machine $M$ compute everywhere defined function
$f(x_1,\,\ldots,\,x_n).$ Then there exists a reduced Minsky machine $M'$ and the constant $C$ such that $M'$ computes $f$ and
$$
T_{M'}(x_1,\,\ldots,\,x_n)\leq C\cdot T_M(x_1,\,\ldots,\,x_n).
$$
\end{statse}
\begin{conseq}
$\bigeps^2$ is the set of all functions that can be computed on
reduced Minsky machines within polynomial time.
\end{conseq}

\section{Vector-functions, Configurations, and Their Codes}

Further one is going to consider everywhere defined vector-functions
of the type
\begin{equation}\label{vector_function}
\tilde{F}:\nat^k\times \{0,\,1,\ldots,\,s-1\}\rightarrow
\nat^k\times\{0,\,1,\ldots,\,s-1\}.
\end{equation}

Let $(x_1,\,\ldots,\,x_k;\ q)$ be a configuration of the Minsky
machine $M$ and
\begin{equation}\label{cmdx78}
e_1 \ldots e_k q \rightarrow d_1 \ldots d_k q'
\end{equation} is a command from the machine $M$ programme such that
\begin{equation}\label{docmdx78}
e_1=\overline{\sg}(x_1),\,\ldots,\,e_k=\overline{\sg}(x_k).
\end{equation}
The command (\ref{cmdx78}) transforms configuration
$(x_1,\,\ldots,\,x_k;\ q)$ into a subsequent configuration
$(x_1+d_1,\,\ldots,\,x_n+d_n;\ q').$

Overall, for the Minsky machine $M$ the process of transformation of
an arbitrary configuration into the next one can be described with
the help of a vector function $\Con_M(x_1,\,\ldots,\,x_k;\ q),$
where
$$
\Con_M(x_1,\,\ldots,\,x_k;\ q)=(x_1+d_1,\,\ldots,\,x_k+d_k;\ q'),
$$
if in the machine program $M,$ there is a command (\ref{cmdx78}),
the following relations hold (\ref{docmdx78}), and
$$
\Con_M(x_1,\,\ldots,\,x_k;\ q)=(x_1,\,\ldots,\,x_k;\ q),
$$
if $(x_1,\,\ldots,\,x_n;\ q)$ is the final configuration.

Let one name the vector function $\tilde{F}$ of the type
$(\ref{vector_function})$ \emph{simple} one if there exist integer
(not necessarily $\nat$) numbers $a_1,\,\ldots,\,a_k,$ as well as
$i\in\nat$ ($i \leq k$) and $q',\,q''\in\{0,\,1,\,\ldots,\,s-1\}$
such that for any vector $(x_1,\ldots,x_k;\ q)$ it satisfies
\begin{equation}\label{simple}
\tilde{F}(x_1,\,\ldots,\,x_k;\ q)=
\begin{cases}
(x_1+a_1,\,\ldots,\,x_k+a_k;\ q'),\text{ if }x_i=0\text{ and
}q=q'', \\
(x_1,\,\ldots,\,x_k;\ q)\text{ else }
\end{cases}
\end{equation} (one can assume that $x_0=0$).

Obviously the following is true.
\begin{statse}
For any reduced $k$-tape Minsky machine $M$ there exist numbers
$s,m\in\nat$ and simple vector functions
$\tilde{F}_1,\,\ldots,\,\tilde{F}_m$ of the type
$(\ref{vector_function})$ such that
$$
\Con_M(\tilde{x})=\tilde{F}_1(\tilde{F}_2(\ldots\tilde{F}_m(\tilde{x})\ldots)).
$$
\end{statse}

The number $x\in\nat$ one calls \emph{$(w;l)$-code} of the
configaration $(x_1,\,\ldots,\,x_k;\ q),$ if binary digits of the
number $x$ from $l\cdot (i-1)$-th up to $(l\cdot i-1)$-th generate
binary notation $x_i$ ($1\leq i\leq k$) and the digits
 from $kl$-th up to $(kl+w-1)$-th have the binary
notation of the number $q.$ One can note that $(w;l)$-code of the
configuration is not unique.

Everywhere defined function $f(x)$ will be called
\emph{simplistic} if there exist $u,v\in\nat$ such that for all
$x$ it satisfies
$$
f(x)=
\begin{cases}
x+v,\text{ if }x\wedge u=0, \\
x\text{ else}.
\end{cases}
$$
\begin{statse}\label{conf_code}
Let $\tilde{F}$ be a simple vector function of the type
$(\ref{vector_function})$, $w,l\in\nat$ be such numbers that
$2^w\geq s$ and $l\geq 1.$ Then there exist such simplistic
functions $f_1,\,f_2,\,f_3,$ that for any vector
$(x_1,\,\ldots,\,x_k;\ q),$ $(y_1,\,\ldots,\,y_k;\ q^*)$ and a
number $x\in\nat,$ if
$$
\begin{array}{l}
x_1,\,\ldots,\,x_k,\,y_1,\,\ldots,\,y_k<2^l, \\
(y_1,\,\ldots,\,y_k;\
q^*)=\tilde{F}(x_1,\,\ldots,\,x_k;\ q),\\
x \text{ is } (w;l)\text{-code of configuration
}(x_1,\,\ldots,\,x_k;\ q),
\end{array}
$$
then $f_3(f_2(f_1(x)))$ is $(w;l)$-code of configuration
$(y_1,\,\ldots,\,y_k;\ q^*).$
\end{statse}
\begin{proof}
Let $\tilde{F}$ have the form $(\ref{simple})$. Then one can assume
$$
\begin{array}{l}
u_1=0, \\
v_1=(2^w-q'')\cdot 2^{lk}, \\
u_2=
\begin{cases}
(2^w-1)\cdot 2^{lk}+(2^l-1)\cdot 2^{l\cdot (i-1)},&\text{ if
}i>0, \\
(2^w-1)\cdot 2^{lk},&\text{ if }i=0,
\end{cases} \\
v_2=2^{lk+w}+(2^w+q'-q'')\cdot 2^{lk}+\sum_{j=1}^k a_j\cdot
2^{l\cdot (j-1)}, \\
u_3=0, \\
v_3=q'' \cdot 2^{lk}, \\
f_j(x)=
\begin{cases}
x+v_j,\text{ if }x\wedge u_j=0, \\
x\text{ else}
\end{cases} \\
(1\leq j\leq 3).
\end{array}
$$
To begin with one can notice that for all $j\ (1\leq j\leq 3)$ the
numbers $u_j,\,v_j$ are non-negative.

Let $f_1(x),\,f_2(f_1(x)),\,f_3(f_2(f_1(x)))$ be the codes of the
configurations $K_1,\,K_2,\,K_3$ respectively. And let
$$
K_j=(x_{j1},\,\ldots,\,x_{jk};\ q_j),\quad 1\leq j\leq 3.
$$

Obviously $f_1(x)=x+(2^w-q'')\cdot 2^{lk}.$ Thus,
$$ K_1=(x_1,\,\ldots,\,x_k;\ q-q'') $$
(here and further in this proof addition and subtraction of numbers
$q,\,q',\,q'',\,q_j$ are performed modulus $2^w$).

Further, let one prove that
\begin{equation}\label{f1_u2_equation}
(f_1(x)\wedge u_2=0)\Leftrightarrow (x_i=0\text{ and }q=q'').
\end{equation}
For this let one consider two cases: $i>0$ and $i=0.$ If $i>0,$ then
in binary notation of the number $u_2$ the ones are placed in digit
places from $l\cdot (i-1)$-th up to $(l\cdot i-1)$-th and from
$lk$-th up to $(lk+w-1)$-th. From this it follows that $f_1(x)\wedge
u_2=0$ if and only if the corresponding digits of the number
$f_1(x)$ are zeroes. That is
$$
(f_1(x)\wedge u_2=0)\Leftrightarrow (x_{1i}=0\text{ and }q_1=0).
$$
One can notice that $q_1=q-q''$ and $x_{1i}=x_i,$ thus one obtains
$(\ref{f1_u2_equation}).$ In case with $i=0$ analogous to
reasonings one gets the following statement
$$
(f_1(x)\wedge u_2=0)\Leftrightarrow (q=q'').
$$
According to the assumption $x_0=0,$ thus, it satisfies
$(\ref{f1_u2_equation}).$

Let $f_1(x)\wedge u_2=0,$ i.e. $x_i=0\text{ and }q=q''.$ In this
case it is obvious that for all $j\ (1\leq j\leq k)$ it satisfies
$y_j=x_j+a_j$. By definition $f_1(x)$ is $(w;l)$-code of the
configuration $K_1.$ Thus the following equalities hold true
$$
\begin{array}{c}
f_1(x)\equiv q_1\cdot 2^{lk}+\Sigma_{j=1}^k x_{1j}\cdot 2^{l\cdot
(j-1)}\equiv (q-q'')\cdot 2^{lk}+\Sigma_{j=1}^k x_j\cdot
2^{l\cdot (j-1)}\equiv \\
\equiv \Sigma_{j=1}^k x_j\cdot 2^{l\cdot (j-1)}\quad (\modx\
2^{lk+w}).
\end{array}
$$
Then
$$
\begin{array}{c}
f_2(f_1(x))\equiv f_1(x)+v_2\equiv (q'-q'')\cdot
2^{lk}+\Sigma_{j=1}^k (x_j+a_j)\cdot 2^{l\cdot (j-1)}\equiv\\
\equiv (q'-q'')\cdot 2^{lk}+\Sigma_{j=1}^k y_j\cdot 2^{l\cdot
(j-1)}\quad (\modx\ 2^{lk+w}).
\end{array}
$$
From this and from the fact that for all $j\ (1\leq j\leq k)$ it satisfies
$y_j<2^l,$ one obtains that
$$
K_2=(y_1,\,\ldots,\,y_k;\ q'-q'').
$$
And if $f_1(x)\wedge u_2\neq 0,$ then it is obvious that for all $j\ (1\leq
j\leq k)$ it satisfies $y_j=x_j.$ Besides, $f_2(f_1(x))=f_1(x)$ and
$$
K_2=K_1=(x_1,\,\ldots,\,x_k;\ q-q'')=(y_1,\,\ldots,\,y_k;\ q-q'').
$$
Further, it is obvious that
$$
f_3(f_2(f_1(x)))=f_2(f_1(x))+v_3=f_2(f_1(x))+q''\cdot 2^{lk}.
$$
from this it follows that $q_3=q_2+q''$ and for all $j\ (1\leq j\leq
k)$ it satisfies
$$
x_{3j}=x_{2j}=y_j.
$$
If $x_i=0$ and $q=q'',$ then
$$
q_3=q_2+q''=(q'-q'')+q''=q'=q^*.
$$
Otherwise,
$$
q_3=q_2+q''=(q-q'')+q''=q=q^*.
$$
Thereby,
$$
K_3=(y_1,\,\ldots,\,y_k;\ q^*).
$$
The statement is proved.
\end{proof}
\begin{conseq}
Let $M$ be a reduced $k$-tape Minsky machine. Then there exists
such $w,r\in\nat$ that for any $l\geq 1$ there are simplistic functions
 $f_{r-1},\,f_{r-2},\,\ldots,\,f_0$ such that for any vectors
$(x_1,\,\ldots,\,x_k;\ q)$ and $(y_1,\,\ldots,\,y_k;\ q')$ and
number $x\in\nat$ from conditions
$$
\begin{array}{l}
x_1,\,\ldots,\,x_k,\,y_1,\,\ldots,\,y_k<2^l, \\
(y_1,\,\ldots,\,y_k;\
q')=\Con_M(x_1,\,\ldots,\,x_k;\ q),\\
x \text{ is } (w;l)-\text{code of configuration
}(x_1,\,\ldots,\,x_k;\ q)
\end{array}
$$
it follows that $f_{r-1}(f_{r-2}(\ldots f_0(x) \ldots))$ is
$(w;l)$-code of the configuration $(y_1,\,\ldots,\,y_k;\ q').$
\end{conseq}
\begin{comment}
It is easy to see that $u_j$ and $v_j$ can be expressed as
polynomials with integer coefficients of $2^l.$
\end{comment}

\section{Basic Property of the Function $Q$}

Let one denote by $h_c(x)$ the number the binary notation of which
is composed of $c$ smaller digit places of the number $x$
($h_c(x)=0,$ if $c=0,$ $h_c(x)=x,$ if the binary notation of $x$ has
less then $c$ digits). One can note that for any $x,\,y,\,c$ the
following relation is satisfied:
$$
h_c(x+y)=h_c(h_c(x)+h_c(y)).
$$
Now let one formulate and prove the basic property of the function $Q.$

\begin{statse}\label{q_prop}
Let integer non-negative numbers $r\geq 1,\
u_0,\,\ldots,\,u_{r-2},\,v_0,\,\ldots,\,v_{r-2}$ and
a sequence of everywhere defined functions
$f_0(x),\,f_1(x),\,\ldots$ is such that for all $i\in\nat$
it satisfies:
$$
\begin{array}{l}
f_i(x)=
\begin{cases}
\begin{cases}
x+v_{\rrm(i,r)},\text{ if }x\wedge u_{\rrm(i,r)}=0, \\
x\text{ else},
\end{cases} \text{ if $\rrm(i,r)\neq r-1,$} \\
x,\text{ if $\rrm(i,r)=r-1.$}
\end{cases}
\end{array}
$$
and let it for numbers
$t_0,\,p_1,\,p_2,\,c_1,\,c_2,\,x,\,u_{r-1},\,v_{r-1}\in\nat$
satisfy the following conditions:
\begin{gather}
t_0\geq 1, \label{q_prop_7}\\
u_{r-1}=2^{c_1}-1, \label{q_prop_5}\\
2^{c_2-1}\leq v_{r-1}<2^{c_2}, \label{q_prop_4}\\
p_1=\Sigma_{i=0}^{r-1} 2^{c_1\cdot i}\cdot u_i, \label{q_prop_8}\\
p_2=\Sigma_{i=0}^{r-1} 2^{c_2\cdot i}\cdot v_i, \label{q_prop_9}\\
x+2p_2 t_0<2^{c_1}, \label{q_prop_3}\\
x+t_0\cdot \max(v_0,\ldots,v_{r-2})<2^{c_2}, \label{q_prop_2}\\
u_i<2^{c_2}\ (0\leq i\leq r-2), \label{q_prop_6}\\
c_1\geq c_2, \label{q_prop_10}\\
x\geq 1. \label{q_prop_1}
\end{gather}
Then
$h_{c_2}(Q(x,\,p_1,\,p_2,\,c_1,\,c_2,\,t_0))=f_{t_0-1}(f_{t_0-2}(\ldots
f_0(x)\ldots)).$
\end{statse}
\begin{proof}
Let
$$
g(t)=
\begin{cases}
f_{t-1}(f_{t-2}(\ldots f_0(x)\ldots)),\text{ if }t>0, \\
x,\text{ if }t=0.
\end{cases}
$$
Obviously for all $t$ it satisfies $g(t+1)\leq
g(t)+\max(v_0,\,\ldots,\,v_{r-2}).$ As a consequence of that,
 (\ref{q_prop_2}) and (\ref{q_prop_1}) one can note that for $t\leq
t_0$ it satisfies
\begin{equation}\label{g_range}
0<g(t)<2^{c_2}.
\end{equation}
Besides, obviously for $t<t_0$ it satisfies
\begin{equation}\label{g_range2}
g(t)+\max(v_0,\,\ldots,\,v_{r-2})<2^{c_2}.
\end{equation}

Further, by induction one can prove that for all $t\leq t_0$ it holds that
$$
h_{c_2}(Q(x,\,p_1,\,p_2,\,c_1,\,c_2,\,t))=g(t).
$$

Induction basis:
$$
h_{c_2}(Q(x,\,p_1,\,p_2,\,c_1,\,c_2,\,0))=x.
$$
This is correct, because
$$
Q(x,\,p_1,\,p_2,\,c_1,\,c_2,\,0)=x
$$
and $x<2^{c_2}$ (see (\ref{q_prop_2})).

Induction step: let $t<t_0$ and
$$
h_{c_2}(Q(x,\,p_1,\,p_2,\,c_1,\,c_2,\,t))=g(t).
$$
One can prove that
$$
h_{c_2}(Q(x,\,p_1,\,p_2,\,c_1,\,c_2,\,t+1))=g(t+1).
$$

From (\ref{q_prop_7}) and (\ref{q_prop_2}) it follows that
$v_0,\,\ldots,\,v_{r-2}<2^{c_2},$ and from (\ref{q_prop_4}) it
follows the fact that $v_{r-1}<2^{c_2}.$ Analogously, from
(\ref{q_prop_6}) and (\ref{q_prop_5}) it follows that
$u_0,\,\ldots,\,u_{r-1}<2^{c_1}.$ Thereby, for all $j\ (0\leq j<r)$
it satisfies the following inequalities
$$
u_j<2^{c_1},\quad v_j<2^{c_2}.
$$

From this and from (\ref{q_prop_8}), (\ref{q_prop_9}) it follows
that the digits from $(c_1\cdot j)$-th up to $(c_1\cdot (j+1)-1)$-th
in binary notation of the number $p_1$ create binary notation of the
number $u_j,$ while digits from $(c_2\cdot j)$-th up to $(c_2\cdot
(j+1)-1)$-th in binary notation of the number $p_2$ form the binary
notation of $v_j$ ($0\leq j<r$). In its turn the relations
(\ref{q_prop_5}) and (\ref{q_prop_4}) show that the binary notation
of numbers $p_1$ and $p_2$ have $c_1\cdot r$ and $c_2\cdot r$ digits
respectively. From this one can conclude that for any$t\in\nat$ the
following inequalities are satisfied:
\begin{gather}
h_{c_1}(R(p_1,\,c_1\cdot t))=u_{\rrm(t,r)}, \label{r_p1}\\
h_{c_2}(R(p_2,\,c_2\cdot t))=v_{\rrm(t,r)}. \label{r_p2}
\end{gather}

Using definiton of the function $Q$ and (\ref{q_prop_3}) one can notice that
$$
Q(x,\,p_1,\,p_2,\,c_1,\,c_2,\,t)\leq x+2p_2 t\leq x+2p_2
t_0<2^{c_1}.
$$
Thus, if $\rrm(t,r)\neq r-1,$ then
$$
\begin{array}{c}
Q(x,\,p_1,\,p_2,\,c_1,\,c_2,\,t)\wedge R(p_1,\,c_1\cdot
t)=Q(x,\,p_1,\,p_2,\,c_1,\,c_2,\,t)\wedge h_{c_1}(R(p_1,\,c_1\cdot
t))=\\
=Q(x,\,p_1,\,p_2,\,c_1,\,c_2,\,t)\wedge u_{\rrm(t,r)}.
\end{array}
$$
From the induction proposal and from (\ref{q_prop_6}) it follows that
$$
Q(x,\,p_1,\,p_2,\,c_1,\,c_2,\,t)\wedge u_{\rrm(t,r)}=g(t)\wedge
u_{\rrm(t,r)}.
$$
Thereby, if $\rrm(t,r)\neq r-1,$ then
\begin{equation}\label{q_prop_norm}
Q(x,\,p_1,\,p_2,\,c_1,\,c_2,\,t)\wedge R(p_1,\,c_1\cdot
t)=g(t)\wedge u_{\rrm(t,r)}.
\end{equation}

If $\rrm(t,r)=r-1,$ then
$$
\begin{array}{c}
h_{c_2}(Q(x,\,p_1,\,p_2,\,c_1,\,c_2,\,t))\wedge
h_{c_2}(R(p_1,\,c_1\cdot t))=\\
=h_{c_2}(Q(x,\,p_1,\,p_2,\,c_1,\,c_2,\,t))\wedge
h_{c_2}(h_{c_1}(R(p_1,\,c_1\cdot t)))= \\
=g(t)\wedge h_{c_2}(u_{\rrm(t,r)})=g(t)\wedge h_{c_2}(u_{r-1})=
\\ g(t)\wedge h_{c_2}(2^{c_1}-1)=g(t)\wedge (2^{c_2}-1)\neq 0.
\end{array}
$$
Here the first inequality follows from (\ref{q_prop_10}), the second
one from the inductive proposal and (\ref{r_p1}), the third one from
the fact that $\rrm(t,r)=r-1,$ the forth one from (\ref{q_prop_5}),
the fifth comes from (\ref{q_prop_10}) and the last inequality comes
from (\ref{g_range}). Consequently, if $\rrm(t,r)=r-1,$ then
\begin{equation}\label{q_prop_extr}
Q(x,\,p_1,\,p_2,\,c_1,\,c_2,\,t)\wedge R(p_1,\,c_1\cdot t)\neq 0.
\end{equation}

Thereby, from (\ref{q_prop_norm}), (\ref{q_prop_extr}) and the definition of the function
$Q$ it follows that
$$
Q(x,\,p_1,\,p_2,\,c_1,\,c_2,\,t+1)=
$$
$$
=\begin{cases} Q(x,\,p_1,\,p_2,\,c_1,\,c_2,\,t)+R(p_2,\,c_2\cdot
t),\text{ if
}g(t)\wedge u_{\rrm(t,r)}=0 \text{ and } \\
\quad \quad \quad \rrm(t,r)\neq r-1,\\
Q(x,\,p_1,\,p_2,\,c_1,\,c_2,\,t)\text{ else}.
\end{cases}
$$
One can notice that if $\rrm(t,r)\neq r-1,$ then from the properties
of function $h_{c_2},$ the inductive proposal, (\ref{r_p2}) and
(\ref{g_range2}) one obtains
$$
\begin{array}{c}
h_{c_2}(Q(x,\,p_1,\,p_2,\,c_1,\,c_2,\,t)+R(p_2,\,c_2\cdot
t))=h_{c_2}(h_{c_2}(Q(x,\,p_1,\,p_2,\,c_1,\,c_2,\,t))+\\
+h_{c_2}(R(p_2,\,c_2\cdot
t)))=h_{c_2}(g(t)+v_{\rrm(t,r)})=g(t)+v_{\rrm(t,r)}.
\end{array}
$$

Thereby,
$$
h_{c_2}(Q(x,\,p_1,\,p_2,\,c_1,\,c_2,\,t+1))=
$$
$$
=\begin{cases}
g(t)+v_{\rrm(t,r)},\text{ if }g(t)\wedge u_{\rrm(t,r)}=0\text { and }\rrm(t,r)\neq r-1, \\
g(t)\text{ else}.
\end{cases}
$$
Cosequently,
$$
h_{c_2}(Q(x,\,p_1,\,p_2,\,c_1,\,c_2,\,t+1))=g(t+1).
$$
The statement is proved.
\end{proof}

\section{Proof of Theorem \ref{theorem_gzh_quazi}}

Let one notate the closure of the system (\ref{arithm}) by $\Phi.$
One can note that
$$
0=x\dotminus x,\quad x+y=(x+1)\cdot (y+1)\dotminus (xy+1).
$$
Consequently, all polynomials with integer coeffcients that take integer non-negative values in arrays of integer non-negative numbers belong to $\Phi.$

Let $f(y_1,\,\ldots,\,y_n)\in \bigeps^2.$ Then according to the consequence from the statement
 \ref{e_2_prived} there exists a reduced Minsky machine $M$ and a polynomial $t(\tilde{y})$ with coefficients
 from $\nat$ such that $M$ calculates $f$ and for any array $\tilde{y}$ it holds that
$$
T_M(\tilde{y})\leq t(\tilde{y}).
$$
Let $M$ has $k$ tapes. One can assume that for any $\tilde{y}$ it
satisfies $t(\tilde{y})\geq 1.$

Further one can choose $w,\,r'\in\nat$ such that for the machine $M$
there holds all conditions of the consequence from the statement
\ref{conf_code}. Let
$$
m(\tilde{y})=t(\tilde{y})+\Sigma_{i=1}^n y_i.
$$
It is obvious that $m(\tilde{y})\in\Phi.$ One can notice at the
beginning $t(\tilde{y})$ steps of the Minsky machine $M$ execution
at the entrance array $\tilde{y}$ the heads will be positioned at
the cells with numbers not exceeding $m(\tilde{y}).$ Let
$$
l(\tilde{y})=\left[ \log_2 m(\tilde{y}) \right]+1.
$$
Then obviously for any $\tilde{y}$ it holds that
$$
m(\tilde{y})<2^{l(\tilde{y})}.
$$
It is obvious that if $(y'_1,\,\ldots,\,y'_k;\ q')$ is some
configuration of the Minsky machine $M$ at the time $t'\leq
t(\tilde{y})$ that is initiated at the entrance array $\tilde{y}$ at
the moment $0$, then it satisfies
$$
y'_1,\,\ldots,\,y'_n<2^{l(\tilde{y})}.
$$
Then according to the consequence from the statement \ref{conf_code}
and remarks to it there exist functions
$u_0(\tilde{y}),\,\ldots,\,u_{r'-1}(\tilde{y}),\,v_0(\tilde{y}),\,\ldots,\,v_{r'-1}(\tilde{y})$
that can be expressed as polynomials of $2^{l(\tilde{y})}$ with
integer coefficients such that for any $\tilde{y}$ it holds that if
$\sigma$ is $(w;\ l(\tilde{y}))$--code of the initial configuration
$(y_1,\,\ldots,\,y_n,\,0,\,\ldots,\,0;\ 1)$ of the Minsky machine
$M$ and if
$$
\begin{array}{l}
f_{i,\,\tilde{y}}(\sigma)=
\begin{cases}
\sigma+v_i(\tilde{y}),\text{ if }\sigma\wedge u_i(\tilde{y})=0,
\\
\sigma\text{ else},
\end{cases}\ (0\leq i<r') \\
\phi_{\tilde{y}}(\sigma)=f_{r'-1,\,\tilde{y}}(f_{r'-2,\,\tilde{y}}(\ldots f_{0,\,\tilde{y}}(\sigma)\ldots)), \\
\psi_{\tilde{y}}(\sigma)=\underbrace{\phi_{\tilde{y}}(\phi_{\tilde{y}}(\ldots
\phi_{\tilde{y}}(\sigma)\ldots))}_{t(\tilde{y})\text{ times}},
\end{array}
$$
then $\psi_{\tilde{y}}(\sigma)$ is $(w;\ l(\tilde{y}))$-code of the
machine $M$ configuration at the time $t(\tilde{y}),$ i.e. the code
of the final configuration. Here $\phi_{\tilde{y}}(\sigma)$ is the
transformation of the code of the Minsky machine $M$ configuration
over the course of one step.

It is obvious that
$$
2^{l(\tilde{y})}=\min(2\cdot m(\tilde{y}),2^{l(\tilde{y})})\in
\Phi.
$$
Consequently,
$$
u_0(\tilde{y}),\,\ldots,\,u_{r'-1}(\tilde{y}),\
v_0(\tilde{y}),\,\ldots,\,v_{r'-1}(\tilde{y})\in \Phi.
$$

Let
$$
x(\tilde{y})=(2^w+1)\cdot (2^{l(\tilde{y})})^k+\Sigma_{i=1}^n
(2^{l(\tilde{y})})^{i-1}\cdot y_i.
$$
It is obvious that $x(\tilde{y})$ is $(w;\l(\tilde{y}))$-code of the initial начальной
configuration $(y_1,\,y_2,\,\ldots,\,y_n,\,0,\,\ldots,\,0;\ 1)$ and
$x(\tilde{y})\in\Phi.$

Let
$$
z(\tilde{y})=\psi_{\tilde{y}}(x(\tilde{y})).
$$

Further, one can define functions
$u_{r'}(\tilde{y}),\,v_{r'}(\tilde{y}),\,p_1(\tilde{y}),\,p_2(\tilde{y}),\,c_1(\tilde{y}),\,c_2(\tilde{y}),\,t'(\tilde{y})$
in the following way:
$$
\begin{array}{l}
t'(\tilde{y})=t(\tilde{y})\cdot (r'+1), \\
c_2(\tilde{y})=\left[ \log_2(x(\tilde{y})+\Sigma_{i=0}^{r'-1}
u_i(\tilde{y})+t'(\tilde{y})\cdot \Sigma_{i=0}^{r'-1} v_i(\tilde{y})) \right]+1, \\
v_{r'}(\tilde{y})=2^{c_2(\tilde{y})-1}, \\
p_2(\tilde{y})=\Sigma_{i=0}^{r'}(2^{c_2(\tilde{y})})^{i}\cdot
v_i(\tilde{y}), \\
c_1(\tilde{y})=\left[ \log_2(x(\tilde{y})+2\cdot
t'(\tilde{y})\cdot
p_2(\tilde{y}))\right]+1, \\
u_{r'}(\tilde{y})=2^{c_1(\tilde{y})}-1, \\
p_1(\tilde{y})=\Sigma_{i=0}^{r'}(2^{c_1(\tilde{y})})^{i}\cdot
u_i(\tilde{y}). \\
\end{array}
$$
Let, furthermore, $r=r'+1$ and $f_{r-1,\tilde{y}}(\sigma)\equiv
\sigma.$ Then obviously the numbers
$$
r,\
u_0(\tilde{y}),\,u_1(\tilde{y}),\,\ldots,\,u_{r-2}(\tilde{y}),\
v_0(\tilde{y}),\,v_1(\tilde{y}),\,\ldots,\,v_{r-2}(\tilde{y}),
$$
the sequence of functions
$$
f_{0,\tilde{y}},\,f_{1,\tilde{y}},\,\ldots,\,f_{r-1,\tilde{y}},
f_{0,\tilde{y}},\,f_{1,\tilde{y}},\,\ldots,\,f_{r-1,\tilde{y}},\ldots
$$
and the numbers
$$
t'(\tilde{y}),\,p_1(\tilde{y}),\,p_2(\tilde{y}),\,c_1(\tilde{y}),\,c_2(\tilde{y}),
\,x(\tilde{y}),\,u_{r-1}(\tilde{y}),\,v_{r-1}(\tilde{y})
$$
satisfy the conditions of the statement \ref{q_prop}. Consequently,
$$
z(\tilde{y})=h_{c_2(\tilde{y})}(Q(x(\tilde{y}),\,p_1(\tilde{y}),\,p_2(\tilde{y}),\,
c_1(\tilde{y}),\,c_2(\tilde{y}),\,t'(\tilde{y}))).
$$
That is,
$$
z(\tilde{y})=\rrm(Q(x(\tilde{y}),\,p_1(\tilde{y}),\,p_2(\tilde{y}),\,
c_1(\tilde{y}),\,c_2(\tilde{y}),\,t'(\tilde{y})),2^{c_2(\tilde{y})}).
$$
It is obvious that
$$
c_1(\tilde{y}),\,c_2(\tilde{y})\in\Phi.
$$
Besides
$$
2^{c_2(\tilde{y})}=\min(2\cdot (x(\tilde{y})+\Sigma_{i=0}^{r-1}
u_i(\tilde{y})+t'(\tilde{y})\cdot \Sigma_{i=0}^{r-1}
v_i(\tilde{y})),2^{c_2(\tilde{y})})
$$
and
$$
2^{c_1(\tilde{y})}=\min(2\cdot(x(\tilde{y})+2\cdot
t'(\tilde{y})\cdot p_2(\tilde{y})),2^{c_1(\tilde{y})}),
$$
Therefore,
$$
2^{c_1(\tilde{y})},\,2^{c_2(\tilde{y})}\in\Phi.
$$
Consequently,
$$
p_1(\tilde{y}),\,p_2(\tilde{y}),\,u_{r'}(\tilde{y}),\,v_{r'}(\tilde{y})\in\Phi.
$$
Since $z(\tilde{y})$ is $(w;l(\tilde{y}))$-code of the concluding
Minsky machine $M$ configuration, it holds that
$$
f(\tilde{y})=\rrm(z(\tilde{y}),2^{l(\tilde{y})}).
$$
Further, one can note that
$$
h(x,\,\tilde{y})=\rrm(\rrm(x,\,2^{c_2(\tilde{y})}),2^{l(\tilde{y})})\in\Phi
$$
and
$$
f(\tilde{y})=h(Q(x(\tilde{y}),\,p_1(\tilde{y}),\,p_2(\tilde{y}),\,
c_1(\tilde{y}),\,c_2(\tilde{y}),\,t'(\tilde{y})),\,\tilde{y}).
$$
Theorem proved.


\chapter{Finite Generability of Some Groups of Recursive Permutations}
\refstepcounter{chapterref}\label{chapter_per}

\section{Definitions}

The majority of definitions and notation can be looked up in sections
\ref{subsection_definitions} and \ref{subsection_basic_per}
of the introduction.

For any set $A$ that is regular in $Q$ one fixes functions from the
definition of regularity ($\mu$ and $\nu$) and will notate them as
$\mu_A$ and $\nu_A$ respectively.

For one-place functions $f,g,h$ the notation $h=f\circ g$ denotes
that $h(x)=f(g(x)).$ If $f$ is a permutation, then $f^{-1}$ denotes
a permutation inverse to $f$.

For the class of functions $Q$  one denotes  $Q^{(1)}$ the set of all one-place functions from $Q.$

\Def{The \emph{graph} of permutation $f(x)$ is the directed graph
with the set of vertexes  $\nat$ and the set of arrows $\{(x,f(x)):\
x\in\nat\}.$}

Let
$$
f^z=\begin{cases}\underbrace{f\circ\ldots\circ f}_{z\text{ times}},
\text{ if $z>0,$} \\ I,\text{ if $z=0,$}
\\ (f^{-1})^{|z|},\text{ if $z<0,$}\end{cases}
$$
where $I(x)=x$ for all $x.$

\Def{Permutation $f$ is called \emph{matching} over the set
 $A,$ if $f=f^{-1}$ and for any $x\notin A$ it satisfies
$f(x)=x.$}

A matching over $\nat$ one calls simply a matching.

\Def{\emph{The characteristic function} of a set $A\subseteq\nat$ is
the function  $\chi_A(x),$ that is defined by the equality
$$
\chi_A(x)=\begin{cases}1,\text{ if $x\in A,$} \\ 0\text{
else}.\end{cases}
$$}

\Def{A permutation $f$ is called \emph{stationary} over a set $A,$
if for any $x\in A$ the following equality holds $f(x)=x.$}

Further one will use the contracted notation for permutations. For
example the notation
$$
f:\ g(y)\leftrightarrow h(y),\ y\geq 2,\ \ t(y)\rightarrow
u(y)\rightarrow v(y)\rightarrow t(y)
$$
means that
$$
f(x)=\begin{cases} h(y),\text{ if $x=g(y)$ for some $y\geq
2,$}
\\
g(y),\text{ if $x=h(y)$ for some $y\geq 2,$} \\
u(y),\text{ if $x=t(y)$ for some $y\in\nat,$} \\
v(y),\text{ if $x=u(y)$ for some $y\in\nat,$} \\
t(y),\text{ if $x=v(y)$ for some $y\in\nat,$} \\
x\text{ in other cases.}\end{cases}
$$
It is noteworthy that this notation is not always correct, its correctness will be proved for every individual case aside from those in which it is obvious.

\Def{The three $(f,g,B),$ where $f,g$ are matchings, $B$ is the set
of vectors from $\nat^4$, called \emph{correct} if all the
components of vectors from $B$ are different (both inside those
vectors and in different vectors) and the following relations are
satisfied
\begin{equation}\label{eq_razb_f}
f:\ b_1\leftrightarrow b_3,\ b_2\leftrightarrow b_4,\
(b_1,b_2,b_3,b_4)\in B,
\end{equation}
$$
g:\ b_1\leftrightarrow b_2,\ (b_1,b_2,b_3,b_4)\in B.
$$
}

\Def{A correct three $(f,g,B)$ is called a correct one over $Q$ if
$f,g\in Q.$}

It is   noteworthy that, from the requirements \ref{demnumerate} and
\ref{demzamkn} it follows the existence in $Q$ of numerating
functions $\ffp_n(x_1,\ldots,x_n)$ mapping one-to-one $\nat^n$ into
$\nat$ and functions inverse to them
$\ffp_{n,1}(x),\ldots,\ffp_{n,n}(x)$ (see ~\cite{maltsev}). Further
in the text there are definitions of some functions that use
numerating functions, i.e. those that depend on their choice. The
assumption is that if any statement mentions the class  $Q,$
satisfying the requirements  \ref{demnumerate}, \ref{demzamkn}, and
some functions from those given below then these functions are
generated based on numeration functions from the class $Q$ (fixed
for the given class).

For every function $f(x)$ let one assume the following
\begin{equation}\label{eq_perf}
\perf_f:\ \begin{array}{l}\ffp_3(x,2y,z)\rightarrow\ffp_3(f(x),2\ffp_2(x,y),z),\\
\ffp_3(x,4y+1,z)\rightarrow \ffp_3(x,2y+1,z),\\
\ffp_3(x,4y+3,z)\rightarrow\ffp_3(x,2y,z),\ x<f(\ffpr_{2,1}(y)),\\
\ffp_3(x,4y+3,z)\rightarrow\ffp_3(x+1,2y,z),\ x\geq
f(\ffpr_{2,1}(y)).\end{array}
\end{equation}
The correctness of this definition will be proved later.

\Def{The pairwise matching $f(x)$ is called \emph{the code} of a
partially defined function $g(x_1,x_2),$ if
\begin{equation}\label{eq_code}
f:\ \ffp_3(x,y,0)\leftrightarrow\ffp_3(x,y,g(x,y)+2),\
g(x,y)\text{ defined}.
\end{equation}
}

Let one label  $\ffpx$ the code of the function $g(x_1,x_2)=x_1.$

Let one assume
\begin{align}
\label{eq_delone} &\ffdelone:\
  \ffp_3(x,2y,0)\leftrightarrow\ffp_3(x,2y,1), \\
&\ffs_{ij}:\  4x+i\leftrightarrow 4x+j\quad(0\leq i<j<3), \notag\\
 \label{eq_move}
&\ffmove:  \begin{array}{l} \ \\ \ffp_3(x,2y,0)\rightarrow\ffp_3(x,2y+2,0),\\
\ffp_3(x,1,0)\rightarrow\ffp_3(x,0,0), \\
\ffp_3(x,2y+3,0)\rightarrow\ffp_3(x,2y+1,0), \\ \ \end{array}
\\
\label{eq_place}
&\ffplace:  \begin{array}{l}\ffp_3(x,0,0)\rightarrow 2x,\\
\ffp_3(x,y+1,0)\rightarrow 4\ffp_2(x,y)+1,\\
\ffp_3(x,y,z+1)\rightarrow 4\ffp_3(x,y,z)+3,\end{array}
\\
\label{eq_swapone} &\ffswapone:
\begin{array}{l}\ \\ \ffp_3(x,2y,z)\rightarrow\ffp_3(x+2,2y,z),\
z\geq 2,\\
\ffp_3(x,2y,0)\rightarrow\ffp_3(x,2y,0), \\
\ffp_3(x+2,2y+1,z)\rightarrow\ffp_3(x,2y+1,z),\ z\geq 2, \\
\ffp_3(x,2y+1,z)\rightarrow\ffp_3(x,2y,z),\ x\in\{0,1\},\ z\geq 2,
\\ \
\end{array}
\\
\label{eq_swaptwo} &\ffswaptwo:  \begin{array}{l}
\ffp_3(x,2y,z)\rightarrow\ffp_3(z,2y,x+2),\ z\geq 2,\ \\
\ffp_3(x,2y,0)\rightarrow\ffp_3(x,2y,0), \\
\ffp_3(x+2,2y+1,z)\rightarrow\ffp_3(x,2y+1,z),\ z\geq 2, \\
\ffp_3(x,2y+1,z)\rightarrow\ffp_3(x,2y,z),\ x\in\{0,1\},\ z\geq 2.
\end{array}
\end{align}

%%%---------------------------------------------------- ДОКАЗАТЕЛЬСТВА ----------------------------------------------------------

\section{Finite Generability of a Group $\gr(Q)$}

\begin{statse}\label{stat_perf_good}
The definition of $\perf_f$ is correct and $\perf_f$ is a
permutation for any function $f(x).$ Besides if $Q$ satisfies the
requirements
 \textup{\ref{dembase}}, \textup{\ref{demnumerate}},
\textup{\ref{demzamkn}} and $f\in Q,$ then
$\perf_f\in\gr(Q).$\end{statse}
\begin{proof}
Let one prove that the definition (\ref{eq_perf}) is correct and
$\perf_f$ is a permutation. Indeed  it is not hard to see that
$i$-th rule in (\ref{eq_perf}) ($1\leq i\leq 4$) one-to-one maps the
set $A_i$ onto $B_i,$ where
$$
A_1=\{\ffp_3(u,v,w):\ \ffrm(v,2)=0\},
$$
$$
A_2=\{\ffp_3(u,v,w):\ \ffrm(v,4)=1\},
$$
$$
A_3=\{\ffp_3(u,v,w):\ \ffrm(v,4)=3,\ u<f(\ffp_{2,1}([v/4]))\},
$$
$$
A_4=\{\ffp_3(u,v,w):\ \ffrm(v,4)=3,\ u\geq f(\ffp_{2,1}([v/4]))\},
$$
$$
B_1=\{\ffp_3(u,v,w):\ \ffrm(v,2)=0,\ u=f(\ffp_{2,1}([v/2]))\},
$$
$$
B_2=\{\ffp_3(u,v,w):\ \ffrm(v,2)=1\},
$$
$$
B_3=\{\ffp_3(u,v,w):\ \ffrm(v,2)=0,\ u<f(\ffp_{2,1}([v/2]))\},
$$
$$
B_4=\{\ffp_3(u,v,w):\ \ffrm(v,2)=0,\ u>f(\ffp_{2,1}([v/2]))\}.
$$
It is clear that $\{A_1,A_2,A_3,A_4\},$ $\{B_1,B_2,B_3,B_4\}$ are
partitions of the set $\nat.$ From that follows the correctness of
the definition and the fact that $\perf_f$ is a permutation.

The fact that  $\perf_f,\perf_f^{-1}\in Q,$ is derived directly from
(\ref{eq_perf}) (see ~\cite{erf}). The statement is proved.
\end{proof}

The proofs of the statements \ref{stat_mps_correct} and
\ref{stat_s_px} are analogous to the proof of the statement
\ref{stat_perf_good}. \
\begin{statse}\label{stat_mps_correct}The definitions of $\ffmove,$ $\ffplace,$
$\ffswapone,$ $\ffswaptwo$ are correct and are the the definitions
of permutations. Besides, if the class $Q$ satisfies requirements
\textup{\ref{dembase}}, \textup{\ref{demnumerate}},
\textup{\ref{demzamkn}}, then these permutations belong to $\gr(Q)$.
\end{statse}

\begin{statse}\label{stat_s_px}If the class $Q$ satisfies the requirements
\textup{\ref{dembase}}, \textup{\ref{demnumerate}} and
\textup{\ref{demzamkn}}, then
$$
\ffdelone=\ffdelone^{-1}\in Q,\quad \ffs_{ij}=\ffs_{ij}^{-1}\in
Q\quad (0\leq i<j\leq 3),\quad \ffpx=\ffpx^{-1}\in Q.
$$
\end{statse}

\begin{statse}\label{statrazbregular}
Let the class $Q$ satisfies the requirements \textup{\ref{dembase}},
\textup{\ref{demzamkn}}. Then any regular in $Q$ set $A$ can be
partitioned into two regular over $Q$ sets $A_1$ and $A_2.$
\end{statse}
\begin{proof}
For $i=1,2$ one can assume that $$\begin{array}{l}A_i=\{x:\ x\in A\text{ и
}\mu_A(x)\equiv i-1\ (\mathrm{mod}\ 2)\},\\
\mu_{A_i}(x)=\begin{cases} (\mu_A(x)-i+1)/2,\text{ если }x\in A_i,\\
0 \text{ иначе},
\end{cases} \\
\nu_{A_i}(x)=\nu_A(2x+i-1).
\end{array}
$$
From \ref{dembase} and \ref{demzamkn} it follows that
$\chi_{A_i},\mu_{A_i},\nu_{A_i}\in Q$ for $i=1,2.$ It is obvious
that the functions $\mu_{A_i},\nu_{A_i}$ fit the definition of
regularity for $A_i$. Besides, it is obvious that $A_1,A_2$ are
infinite and create a partition of $A.$ The statement is proved.
\end{proof}

\begin{statse}\label{statunionregular}
Let the class $Q$ satisfy the requirements \textup{\ref{dembase}},
\textup{\ref{demzamkn}}. Then for any non-intersecting regular in
$Q$ sets $A$ and $B$ the set $A\cup B$ is regular in $Q.$
\end{statse}
\begin{proof}
One can assume
$$
\begin{array}{l}
\mu(x)=\begin{cases}2\mu_A(x),\text{ else $x\in A,$} \\
2\mu_B(x)+1,\text{ если $x\in B,$} \\ 0\text{ else,}\end{cases}\\
\nu(x)=\begin{cases}\nu_A(x/2),\text{ if $x\equiv 0\
(\mathrm{mod}\ 2),$}\\ \nu_B((x-1)/2),\text{ если $x\equiv 1\
(\mathrm{mod}\ 2).$}\end{cases}
\end{array}
$$
From the requirements \ref{dembase} and \ref{demzamkn} it follows that
$\mu,\nu\in Q.$ It is clear that the functions $\mu,\nu$ do comply with the definition of regularity for $A\cup B,$ $\chi_{A\cup
B}=\chi_A+\chi_B\in Q.$ The statement is proved.
\end{proof}

\begin{statse}\label{statstacionar}
Let the class $Q$ satisfy the requirements
\textup{\ref{dembase}}--\textup{\ref{demzamkn}}. Then for any
permutation $f\in\gr(Q)$ there exist permutations $f_1,f_2\in\gr(Q)$
and sets $A_1,A_2$ that are regular in $Q$ such that $\nat\backslash
A_1,$ $\nat\backslash A_2$ are regular in $Q,$ $f=f_1\circ f_2,$
$f_i$ is stationary over the set $A_i$ ($i=1,2$).
\end{statse}
\begin{proof}
Let $A,$ $B$ be the sets from the requirement \ref{demrazb} for $Q$
and $f$. Let one divide $B$ into two regular in $Q$ sets $B_1,B_2$
(which can be done according to the statement
\ref{statrazbregular}).

Let one assume that
$$
f_1:\begin{array}{l}x\rightarrow f(x),\ x\in A;\\ x\rightarrow
\nu_{B_1}(\ffp_2(x,0)),\ x\notin A,f^{-1}(x)\in A; \\
\nu_{B_1}(\ffp_2(x,0))\rightarrow x,\ x\in A,f^{-1}(x)\notin A; \\
\nu_{B_1}(\ffp_2(x,y))\rightarrow\nu_{B_1}(\ffp_2(x,y+1)),\
x\notin
A,f^{-1}(x)\in A; \\
\nu_{B_1}(\ffp_2(x,y+1))\rightarrow\nu_{B_1}(\ffp_2(x,y)),\ x\in
A,f^{-1}(x)\notin A.
\end{array}
$$

It is easy to show that the set of arrows of the graph of the
permutation $f_1$ is derived from the set $\{(x,f(x)):\ x\in A\}$ by
adding for each sink vertex $x$ (a vertex into which an arrow enters
but there is no arrow coming out of it) of the infinite chain
$$
x\rightarrow\nu_{B_1}(\ffp_2(x,0))\rightarrow
\nu_{B_1}(\ffp_2(x,1))\rightarrow\ldots\ ,
$$
for every source vertex $x$ (a vertex out of which there is an
out-coming arrow but no incoming one) of the infinite chain
$$
\ldots\rightarrow
\nu_{B_1}(\ffp_2(x,1))\rightarrow\nu_{B_1}(\ffp_2(x,0))\rightarrow
x
$$
and by adding loops $(x,x)$ for every vertex $x$ that has no
out-coming arrow and no incoming arrow. After setting up this
construction there is a graph in which there is only one out-coming
and one incoming arch. From that it follows that the definition of
$f_1$ is correct and $f_1$ is a permutation.

From the definition of $f_1$ it follows that $f_1$ is stationary
over the set $\nat\backslash(A\cup f(A)\cup B_1).$ From the
requirement \ref{demrazb} for $Q$ it follows that $(A\cup f(A))\cap
B=\varnothing.$ And thus,
$$
B_2\subseteq \nat\backslash(A\cup f(A)\cup B_1).
$$
Thereby, it follows that $f_1$ is stationary over $B_2.$

Let $f_2=f_1^{-1}\circ f.$ From the fact that $f_1$ coincides with
$f$ over $A$ (see the definition of $f_1$) it follows that $f_2$ is
stationary over $A.$

From drawing an analogy with the statement \ref{stat_perf_good} it is easy to show that
 $f_1,f_1^{-1}\in Q.$ From that it follows that $f_2\in Q$ and
$f_2^{-1}=f^{-1}\circ f_1\in Q.$ Put $A_1=B_2,$ $A_2=A.$ It is
noteworthy that the regularity $\nat\backslash A_2$ in $Q$ follows
from choosing the set $A$, the regularity of $\nat\backslash
A_1=(\nat\backslash B)\cup B_1$ follows from choosing the set $B,$
from regularity of $B_1$ and from statement \ref{statunionregular}.
The statement is proved.
\end{proof}

\begin{statse}\label{statstacionarpair}
Let the class $Q$ satisfy the requirements \textup{\ref{dembase}},
\textup{\ref{demnumerate}}, \textup{\ref{demzamkn}}. Then if the
permutation $f\in\gr(Q)$ is stationary  over some regular in $Q$ set
$A$ such that $\nat\backslash A$ also regular in $Q,$ then there
exists such correct in $Q$ triplets
$(f_1,g_1,B_1),\ldots,(f_k,g_k,B_k)$ that
$$
f=f_1\circ\ldots\circ f_k.
$$
\end{statse}
\begin{proof}
By using $twice$ the statement \ref{statrazbregular}, let one divide
the set $A$ into regular in $Q$ sets $A_1,$ $A_2,$ $A_3.$

Let one introduce  auxiliary functions $\ffp'$ и $\ffp'',$ mapping
$\nat\times\mathbb{Z}$ to $\nat$ ($\mathbb{Z}$ is the set of all
integer numbers):
\begin{equation}\label{eq_sp_ps}
\ffp'(x,z)=\begin{cases} \nu_{\nat\backslash A}(x),\text{
if $z=0,$} \\ \nu_{A_1}(x),\text{ if $z=1,$} \\
\nu_{A_3}(\ffp_2(x,z-2)),\text{ if $z>1,$} \\
\nu_{A_2}(\ffp_2(x,-z-1)),\text{ if $z<0,$}
\end{cases}
\end{equation}
\begin{equation}\label{eq_sp_pss}
\ffp''(x,z)=\begin{cases} \ffp'(x,z),\text{ if $z\leq 0,$} \\
\ffp'(\mu_{\nat\backslash A}\circ f\circ \nu_{\nat\backslash
A}(x),z),\text{ if $z>0.$}
\end{cases}
\end{equation}
It is clear that the mapping $\ffp'$ is one-to-one. Besides, it is
clear that
 $f$ one-to-one maps $\nat\backslash A$ to
$\nat\backslash A$ (because $f$ is stationary over $A$). Thus,
$\mu_{\nat\backslash A}\circ f\circ \nu_{\nat\backslash A}$ is a
permutation ($\nu_{\nat\backslash A}$ that is a one-to-one mapping
of $\nat$ onto $\nat\backslash A,$ $f$
--- $\nat\backslash A$ onto $\nat\backslash A,$ and
$\mu_{\nat\backslash A}$ --- $\nat\backslash A$ onto $\nat$). Thus,
it follows that $\ffp''$ is as well one-to-one. Put
$$ r_1:\ \ffp''(x,z)\leftrightarrow\ffp''(x,1-z),\
z\in\mathbb{Z},
$$
$$
r_2:\ \ffp''(x,z)\leftrightarrow\ffp''(x,-z),\ z\in\mathbb{Z}.
$$
The permutations $h_1,$ $h_2$ is defined with the help of the same
formulas with a replacement of $\ffp''$ to $\ffp'.$ Let one remark
that $r_1,$ $r_2,$ $h_1,$ $h_2$ are matchings that belong to $Q$
(this is proved analogously to the statement \ref{stat_perf_good}).

Let one assume
$$
r=r_1\circ r_2,\quad h=h_1\circ h_2.
$$
It is easy to check that
\begin{equation}\label{eq_sp_g}
r:\ \ffp''(x,z)\rightarrow\ffp''(x,z+1),
\end{equation}
\begin{equation}\label{eq_sp_h}
h:\ \ffp'(x,z)\rightarrow\ffp'(x,z+1).
\end{equation}

 From (\ref{eq_sp_pss}), (\ref{eq_sp_g}) and from the fact that
 $\mu_{\nat\backslash A}\circ f\circ \nu_{\nat\backslash A}$
 is a permutation it follows that the set of the arrows of the
 graph of the permutation $r$ is the union of all infinite chains of the form
$$
\begin{array}{l}
\ldots\rightarrow\ffp'(x,-1)\rightarrow\ffp'(x,0), \\
\ffp'(x,1)\rightarrow\ffp'(x,2)\rightarrow\ldots\end{array}
\quad\quad (x\in\nat)
$$
and the set $\{(\ffp'(x,0),\ffp'(\mu_{\nat\backslash A}\circ f\circ
\nu_{\nat\backslash A}(x),1)):\ x\in\nat\}.$ On the other hand, from
(\ref{eq_sp_h}) it follows that the set of all arrows of the graph
of the permutation
 $h$ is a union of non-intersecting infinite chains of the form
$$
\ldots\rightarrow\ffp'(x,-1)\rightarrow\ffp'(x,0)\rightarrow\ffp'(x,1)\rightarrow\ldots\quad(x\in\nat).
$$
From this is follows that $r$ and $h$ coincide on the set
$$
\{\ffp'(x,z):\ z\in\mathbb{Z},\ z\neq 0\}=A
$$
(see (\ref{eq_sp_ps})). From this one can conclude that
$h^{-1}\circ g$ is stationary on $A.$

Further for any $x\in\nat$ the following equalities hold
\begin{equation}\label{eq_sp_gc}
r(\ffp'(x,0))=r(\ffp''(x,0))=\ffp''(x,1)=\ffp'(\mu_{\nat\backslash
A}\circ f\circ \nu_{\nat\backslash A}(x),1)
\end{equation}
(the equalities follow from (\ref{eq_sp_pss}), (\ref{eq_sp_g}),
(\ref{eq_sp_pss}) respectively). Furthermore, one has
$$
h^{-1}\circ r\circ \nu_{\nat\backslash A}(x)=h^{-1}\circ
r(\ffp'(x,0))=h^{-1}(\ffp'(\mu_{\nat\backslash A}\circ f\circ
\nu_{\nat\backslash A}(x),1))=
$$
$$
=\ffp'(\mu_{\nat\backslash A}\circ f\circ \nu_{\nat\backslash
A}(x),0)=\nu_{\nat\backslash A}\circ \mu_{\nat\backslash A}\circ
f\circ \nu_{\nat\backslash A}(x)=f\circ \nu_{\nat\backslash A}(x).
$$
(first equality through the forth one follow from (\ref{eq_sp_ps}),
(\ref{eq_sp_gc}), (\ref{eq_sp_h}), (\ref{eq_sp_ps})
respectively). Thereby, $h^{-1}\circ r$ and $f$ coincide over
$\nat\backslash A.$ Considering that $f$ and $h^{-1}\circ r$
are stationary at $A,$ the following holds
$$
f=h^{-1}\circ r=h_2\circ h_1\circ r_1\circ r_2.
$$
Let there be
$$
s_1:\ \ffp'(x,-2y)\leftrightarrow\ffp'(x,-2y-1),\ y\geq 0,
$$
$$
s_2:\ \ffp'(x,-2y-1)\leftrightarrow\ffp'(x,-2y-2),\ y\geq 0,
$$
$$
M_1=\{(\ffp'(x,-2y),\ffp'(x,-2y-1),\ffp'(x,2y+1),\ffp'(x,2y+2)),\
y\geq 0\},
$$
$$
M_2=\{(\ffp'(x,-2y-1),\ffp'(x,-2y-2),\ffp'(x,2y+1),\ffp'(x,2y+2)),\
y\geq 0\}.
$$
Obviously $s_1,s_2$ are machings in $Q.$ Thus, it is easy to notice that $(r_1,s_1,M_1),$ $(r_2,s_2,M_2),$ $(h_1,s_1,M_1),$
$(h_2,s_2,M_2)$
are correct in $Q$ threesomes. The statement is proved.
\end{proof}

\begin{theorem}\label{statanyhpair}Let the class $Q$ satisfy the requirements
\textup{\ref{dembase}}--\textup{\ref{demzamkn}}. Then for any permutation $f\in\gr(Q)$ there exist the correct in $Q$ threesomes
$(f_1,g_1,B_1),\ldots,(f_k,g_k,B_k)$ such that
$$
f=f_1\circ\ldots\circ f_k.
$$
\end{theorem}
\begin{proof}
Using the statement \ref{statstacionar} let one represent $f$ in the
form of compositions of permutations that are in $Q$ together with
the permutations reverse to them, stationary in some sets, regular
together with their complements. Then using the statement
\ref{statstacionarpair} let one represent each of these permutations
in the form of compositions of matchings having the corresponding
correct in $Q$ threesomes. Theorem proved.
\end{proof}

\begin{statse}\label{statdomainswap}
Let the permutation $f$ be correctly defined by the formula
$$
f:\ g_1(\tilde{x_1})\rightarrow h_1(\tilde{x_1}),\
\rho_1(\tilde{x_1})\text{ is true,}\ \ldots,\
g_n(\tilde{x_n})\rightarrow h_n(\tilde{x_n}),\
\rho_n(\tilde{x_n})\text{ is true,}
$$
where $g_i,h_i$ are some functions, $\rho_i$ are some predicates
 ($1\leq i\leq n$). Besides, let $p$
be some permutation. Then it holds that
$$
p\circ f\circ p^{-1}:\
\begin{array}{l}p(g_1(\tilde{x_1}))\rightarrow p(h_1(\tilde{x_1})),\
\rho_1(\tilde{x_1})\text{ true,} \\ ............................, \\
p(g_n(\tilde{x_n}))\rightarrow p(h_n(\tilde{x_n})),\
\rho_n(\tilde{x_n})\text{ true.}\end{array}
$$
\end{statse}
\begin{proof}
One can be convinced by checking.
\end{proof}

\begin{statse}\label{statfunctionswap}
Let $g(x,y)$ be a partially defined function and $f$ is a permutation
that is its code. Besides, let
$(h_1(x,y),h_2(x,y))$ be a permutation on the set $\nat^2,$ $p$
is a permutation that is defined by the formula
$$
p:\ \ffp_3(x,y,z)\rightarrow \ffp_3(h_1(x,y),h_2(x,y),z).
$$
Then $p^{-1}\circ f\circ p$ is a code of the function
$g(h_1(x,y),h_2(x,y)).$
\end{statse}
\begin{proof}
Let $(h'_1,h'_2)$ be an inverse permutation to $(h_1,h_2)$ (i.e.
$h'_1(h_1(x,y),h_2(x,y))=x,$ $h'_2(h_1(x,y),h_2(x,y))=y$ for any
$x,y$). It is clear that
$$
p^{-1}:\ \ffp_3(x,y,z)\rightarrow \ffp_3(h'_1(x,y),h'_2(x,y),z).
$$
By using the statement \ref{statdomainswap}, one gets
$$
p^{-1}\circ f\circ p:\
\ffp_3(h'_1(x,y),h'_2(x,y),0)\leftrightarrow\ffp_3(h'_1(x,y),h'_2(x,y),g(x,y)+2),
$$
$$
g(x,y)\text{ defined}.
$$
Considering the fact that $(h'_1,h'_2)$ is a permutation one can make the following substitution
$z=h'_1(x,y),$ $t=h'_2(x,y)$:
$$
p^{-1}\circ f\circ p:\
\ffp_3(z,t,0)\leftrightarrow\ffp_3(z,t,g(h_1(z,t),h_2(z,t))+2),
$$
$$
g(h_1(z,t),h_2(z,t))\text{ defined}.
$$
The statement proved.
\end{proof}

\begin{statse}\label{statuncompelementar}
Let there be some functions $f(x),g(x)$ , $r(x_1,x_2)$ is a function,
in which for any $x,y$ it holds that
$$
r(x,2y)=f(x),
$$
$s$ is a permutation that is a code of $r.$ Then $\perf_g^{-1}\circ
s\circ\perf_g$ is a permutation that is the code of such function
$q(x_1,x_2),$ that
$$
q(x,2y)=f(g(x)).
$$
\end{statse}
\begin{proof}
From the definition of permutation $\perf_g$ (see (\ref{eq_perf}))
it follows that
$$
\perf_g:\ \ffp_3(x,y,z)\rightarrow\ffp_3(h_1(x,y),h_2(x,y),z),
$$
where $(h_1,h_2)$ is a permutation on the $\nat^2,$ such that for any
$x,y$ it holds that
$$
h_1(x,2y)=g(x),\quad h_2(x,2y)=2\ffp_2(x,y).
$$
(the equalities follow from the first rule in (\ref{eq_perf})).
Using the statement \ref{statfunctionswap}, one gets
$$
q(x,2y)=r(h_1(x,2y),h_2(x,2y))=r(g(x),2\ffp_2(x,y))=f(g(x)).
$$
The statement is proved.
\end{proof}

\begin{statse}\label{statunarcompos}
Let $f_1(x),\ldots,f_n(x)$ be such functions that
$$
f=f_1\circ\ldots\circ f_n,
$$
$$
p=\perf_{f_n}^{-1}\circ\ldots\perf_{f_1}^{-1}\circ\ffpx\circ\perf_{f_1}\circ\ldots\circ\perf_{f_n}.
$$
Then $p$ is the code of the function $g(x,y)$ such that for any
$x,y$ it satisfies
$$
g(x,2y)=f(x).
$$
\end{statse}
\begin{proof}
By definition the permutation $\ffpx$ is the code of the function
$p'(x,y)=x,$ thus for any $x$ and $y$ it holds that $p'(x,2y)=x.$
Thereby, by applying $n$ times the statement
\ref{statuncompelementar}, one obtains the statement that is being proved.
\end{proof}

\begin{statse}\label{statdelete}
Let $f_1,f_2$ be matchings, $A\subseteq\nat.$  Besides, let it
satisfy the following conditions:
\begin{enumerate}
\item\label{dem_statdel_one} $A,$ $f_1(A)$ and $f_2(A)$ do not intersect pairwise;
\item $f_1$ stationary at $f_2(A);$ \item $f_2$
stationary at $\nat\backslash(A\cup f_2(A))$.
\end{enumerate}
Then
$$
(f_1\circ f_2)^4\circ f_2:\ x\leftrightarrow f_1(x),\ x\in A.
$$
\end{statse}
\begin{proof}
It is clear that the set of all arrows (excluding loops) of the
graph of the permutation $f_1$ consist of pairs of arrows of the
form $x\leftrightarrow f_1(x),$ $x\in A$ and some other pairs of the
arrows of the form $x\leftrightarrow y,$ where $x,y\notin A\cup
f_1(A)\cup f_2(A).$ An analogous set for $f_2$ consists of only of
pairs of arrows $x\leftrightarrow f_2(x)$ ($x\in A$). From this with
the consideration of the condition \ref{dem_statdel_one} it follows
that the graph of the permutation $f_1\circ f_2$ consist of
non-intersecting cycles of length $3$ of the form $x\rightarrow
f_2(x)\rightarrow f_1(x)\rightarrow x$ ($x\in A$), cycles of the
length $2$ of the form $x\leftrightarrow f_1(x)$ ($x,f_1(x)\notin
A\cup f_1(A)\cup f_2(A)$) and loops. After raising to the $4$-th
power the cycles of length $3$ and loops remain at their place,
cycles of length $2$ turn into pairs of loops. After multiplying the
result by  $f_2$ cycles of length $3$ transform into cycles of
lengths $2$ of the form $x\leftrightarrow f_1(x)$ ($x\in A$) and
loops $f_2(x)\rightarrow f_2(x)$ ($x\in A$), all the rest remains at
its place. The statement is proved.
\end{proof}

\begin{statse}\label{statdelodd}
If $f$ is a permutation, which is the code of everywhere defined function $g(x_1,x_2),$ then $(f\circ\ffdelone)^4\circ\ffdelone$ is
the code of a partially defined function $g'(x_1,x_2),$ where
$$
g'(x_1,x_2)=\begin{cases}g(x_1,x_2),\text{ if  $x_2$ is even,}
\\ \text{undefined otherwise.}\end{cases}
$$
\end{statse}
\begin{proof}
Let $A=\{\ffp_3(x,2y,0),\ x,y\in\nat\}.$ From (\ref{eq_code}) and
(\ref{eq_delone}) it follows that
$$
f(A)\subseteq \{\ffp_3(x,2y,z):\ z\geq 2\},
$$
$$
\ffdelone(A)=\{\ffp_3(x,2y,1):\ x,y\in\nat\}.
$$
Thereby, $A,$ $f(A),$ $\ffdelone(A)$ are pairwise non-intersecting.
Besides, it is easy to check that $f$ is stationary at
$\{\ffp_3(x,2y,1):\ x,y\in\nat\},$ and $\ffdelone$ is at
$\{\ffp_3(x,2y,z):\ z\geq 2\}.$ Thereby, it satisfies all of the
conditions of the claim \ref{statdelete} for the set $A$ and the
functions $f,$ $\ffdelone.$ From that it follows that
$$
(f\circ\ffdelone)^4\circ\ffdelone:\ x\leftrightarrow f(x),\ x\in
A.
$$
This can be rewritten in the following way
$$
(f\circ\ffdelone)^4\circ\ffdelone:\
\ffp_3(x,2y,0)\leftrightarrow\ffp_3(x,2y,g(x,2y)+2).
$$
The statement is proved.
\end{proof}

\begin{statse}\label{statchetnhpair}
Let the class $Q$ satisfies the requirements \textup{\ref{dembase}},
\textup{\ref{demnumerate}}, \textup{\ref{demzamkn}},
$$Q^{(1)}=[\{q_1,\ldots,q_n\}].$$
Then any matching over the set of all even numbers  $f\in Q$ can be
expressed in terms of compositions of permutations
$\perf_{q_1},\ldots,\perf_{q_n},$ $\ffpx,$ $\ffswapone,$
$\ffswaptwo,$ $\ffmove,$ $\ffplace,$ $\ffdelone$ and reverse to it.
\end{statse}
\begin{proof}
Let
\begin{equation}\label{eq_hpair_shtrih}
f'(x)=f(2x)/2.
\end{equation}
It is obvious that $f'(x)$ is a matching that belongs to $\ccq.$ Let
there be
$$
f'=r_1\circ\ldots\circ r_k,
$$
where $r_1,\ldots,r_k\in\{q_1,\ldots,q_n\}.$ Let there be
$$
\psi=\perf_{r_k}^{-1}\circ\ldots\circ\perf_{r_1}^{-1}\circ\ffpx\circ\perf_{r_1}\circ\ldots\circ\perf_{r_k}.
$$
From the statement \ref{statunarcompos} it follows that $\psi$ is the code of a function
$g(x,y),$ and further, for any $x$ and $y$
it satisfies
\begin{equation}\label{eq_hpair_compos}
g(x,2y)=f'(x).
\end{equation}
 Let one assume
$$
\psi_1=(\psi\circ\ffdelone)^4\circ\ffdelone.
$$
From the statement \ref{statdelodd} and (\ref{eq_hpair_compos}) it follows that
 $\psi_1$ is the code of a partially defined function $g'(x,y),$
 that is being defined by the equality

$$
g'(x,y)=\begin{cases}f'(x),\text{ if $y$ is even,}\\ \text{undefined otherwise.}\end{cases}
$$
From this it follows that
\begin{equation}\label{eq_hpair_code}
\psi_1:\ \ffp_3(x,2y,0)\leftrightarrow\ffp_3(x,2y,f'(x)+2).
\end{equation}
Let
$$
\psi_2=\ffswapone\circ\psi_1\circ\ffswapone^{-1},\quad
\psi_3=\ffswaptwo\circ\psi_1\circ\ffswaptwo^{-1}.
$$
From the first two rules in the formulas (\ref{eq_swapone}) and
(\ref{eq_swaptwo}), from (\ref{eq_hpair_code}) and from the
statement \ref{statdomainswap} it follows that
\begin{equation}\label{eq_hpair_swapone}
\psi_2:\ \ffp_3(x,2y,0)\leftrightarrow\ffp_3(x+2,2y,f'(x)+2),
\end{equation}
\begin{equation}\label{eq_hpair_psithree}
\psi_3:\ \ffp_3(x,2y,0)\leftrightarrow\ffp_3(f'(x)+2,2y,x+2).
\end{equation}
Minding the fact that $f'$ is a matching, let one substitute
$u=f'(x)$ and rewrite (\ref{eq_hpair_psithree}) in terms of
\begin{equation}\label{eq_hpair_swaptwo}
\psi_3:\ \ffp_3(f'(u),2y,0)\leftrightarrow\ffp_3(u+2,2y,f'(u)+2).
\end{equation}
Let
$$
\psi_4=\psi_3\circ\psi_2.
$$
From (\ref{eq_hpair_swapone}), (\ref{eq_hpair_swaptwo}) and the fact
that $f'$ is a matching it follows that
\begin{equation}\label{eq_hpair_musor}
\begin{array}{c}
\psi_4:\ \ffp_3(x,2y,0)\leftrightarrow\ffp_3(f'(x),2y,0);\\
\ffp_3(x+2,2y,f'(x)+2)\leftrightarrow\ffp_3(f'(x)+2,2y,x+2).
\end{array}
\end{equation}
Let
$$
\psi_5=\ffmove\circ\psi_4\circ\ffmove^{-1}.
$$
From (\ref{eq_move}), (\ref{eq_hpair_musor}) and the statement
\ref{statdomainswap} it follows that
\begin{equation}\label{eq_hpair_moved}
\begin{array}{l}
\psi_5:\ \ffp_3(x,2y,0)\leftrightarrow\ffp_3(f'(x),2y,0),\ y>0;\\
\quad \ffp_3(x+2,2y,f'(x)+2)\leftrightarrow\ffp_3(f'(x)+2,2y,x+2).
\end{array}
\end{equation}
One can assume
\begin{equation}\label{eq_hpair_badplace}
\psi_6:\ \ffp_3(x,0,0)\leftrightarrow\ffp_3(f'(x),0,0).
\end{equation}
From (\ref{eq_hpair_musor}), (\ref{eq_hpair_moved}) and
$\psi_5=\psi_5^{-1}$ it follows that
$$
\psi_6=\psi_5\circ\psi_4.
$$
From the first rule in (\ref{eq_place}), from
(\ref{eq_hpair_badplace}) and from the statement
\ref{statdomainswap} it follows that
$$
\ffplace\circ\psi_6\circ\ffplace^{-1}:\ 2x\leftrightarrow 2f'(x).
$$
From this and from the fact that
(\ref{eq_hpair_shtrih}) it follows that
$$
f=\ffplace\circ\psi_6\circ\ffplace^{-1}.
$$
The statement is proved.
\end{proof}

\begin{theorem}\label{theorem_main}
Let class $Q$ satisfy all of the requirements
\textup{\ref{dembase}}--\textup{\ref{demzamkn}},
$$Q^{(1)}=[\{q_1,\ldots,q_n\}].$$
Then any permutation $f\in\gr(Q)$ can be expressed in terms of compositions of permutations
$\perf_{q_1},\ldots,\perf_{q_n},$ $\ffpx,$
$\ffswapone,$ $\ffswaptwo,$ $\ffmove,$ $\ffplace,$ $\ffdelone,$
$\ffs_{01},$ $\ffs_{02},$ $\ffs_{03},$ $\ffs_{12},$ $\ffs_{13},$
$\ffs_{23}$ and their inverses.
\end{theorem}
\begin{proof}
According to the theorem \ref{statanyhpair} one can contend that $f$
is a matching.

Let there be
$$
f_{ij}:\ x\leftrightarrow f(x),\ x\equiv i\ (\mathrm{mod}\ 4),\
f(x)\equiv j\ (\mathrm{mod}\ 4),\quad 0\leq i\leq j\leq 3.
$$
It is obvious that $f_{ij}$ is a matching, $f=f_{00}\circ
f_{01}\circ\ldots\circ f_{33}.$

Obviously for any $i,j$ ($0\leq i\leq j\leq 3$) there exists a
permutation $q_{ij},$ that can be expressed in terms of composition
of permutations $\ffs_{01},$ $\ffs_{02},$ $\ffs_{03},$ $\ffs_{12},$
$\ffs_{13},$ $\ffs_{23},$ that maps the set $\{x:\
\ffrm(x,4)\in\{i,j\}\}$ into a subset of the set of all even numbers
(one can notice that $q_{ij}^{-1}$ can also be expressed in terms of
the composition of these permutations). Let
$$
f'_{ij}=q_{ij}\circ f_{ij}\circ q_{ij}^{-1}\quad (0\leq i\leq
j\leq 3).
$$
From the statement \ref{statdomainswap} and the fact that $f_{ij}$
is a matching over the set $\{x:\ \ffrm(x,4)\in\{i,j\}\}$ it follows
that $f'_{ij}$ is a matching over the set of all even numbers.
Therefore, for any  $i,j$ ($0\leq i\leq j\leq 3$) $f'_{ij}$ it can
be expressed in terms of the composition of permutations from the
conditions of the theorem (statement \ref{statchetnhpair}). From the
fact that
$$
f_{ij}=q_{ij}^{-1}\circ f'_{ij}\circ q_{ij}\quad (0\leq i\leq
j\leq 3),
$$
follows the claim of the theorem.
\end{proof}
\begin{conseq}
If the class $Q$ satisfies the requirements
\textup{\ref{dembase}}--\textup{\ref{demrazb}},
\textup{\ref{demunarbasis}}, then the group $\gr(Q)$ is finitely generated.
\end{conseq}
\begin{proof}
Indeed, from ~\cite{basis} is known that from the requirements
\ref{dembase}, \ref{demnumerate}, \ref{demunarbasis} follows the existence of the finite basis with respect to superpositioning in
$Q^{(1)}$. Further applying the theorem
 \ref{theorem_main}.
\end{proof}

\section{Generatability of a Group $\gr(Q)$ by Using Two Permutations}

\begin{statse}\label{statthreetwo}
Let class $Q$ satisfy the requirements \textup{\ref{dembase}},
\textup{\ref{demzamkn}}. Besides let $A,$ $B,$ $C$ be
non-intersecting sets, $B$ is regular in $Q$. Then any permutation
$f\in Q,$ that is a matching over the set $A\cup B\cup C,$ can be
expressed in terms of compositions of permutations from $Q,$ when
each of them is a matching over the set $A\cup B$ or $B\cup C.$
\end{statse}
\begin{proof}
Let
$$
f_1:\ x\leftrightarrow f(x),\ x,f(x)\in A\cup B,
$$
$$
f_2:\ x\leftrightarrow f(x),\ x\in C,f(x)\in B\cup C,
$$
$$
f_3:\ x\leftrightarrow f(x),\ x\in A,f(x)\in C.
$$
Obviously
$$
f=f_1\circ f_2\circ f_3.
$$
Let one assume that
$$
g_1:\ x\leftrightarrow\nu_B(x),\ x\in A,f(x)\in C,
$$
$$
g_2:\ \nu_B(x)\leftrightarrow f(x),\ x\in A,f(x)\in C.
$$
One can remark that $\nu_B(x)$ is injective and therefore the definitions are correct.
Furthermore, it is easy to check that
$$
g_2\circ g_1:\ x\rightarrow f(x)\rightarrow \nu_B(x)\rightarrow
x,\ x\in A,f(x)\in C.
$$
From this is follows that
$$
g_1\circ g_2\circ g_1:\ x\leftrightarrow f(x),\ x\in A,f(x)\in C.
$$
Hereby, $f_3=g_1\circ g_2\circ g_1,$ or
$$
f=f_1\circ f_2\circ g_1\circ g_2\circ g_1.
$$
One can note that $f_1,f_2,g_1,g_2\in Q$. Besides, $f_1,$ $g_1$ are
matchings over $A\cup B,$ and $f_2,$ $g_2$  are matchings over
$B\cup C.$ The claim is proved.
\end{proof}

\begin{statse}\label{statsequence}
Let class $Q$ satisfy the requirements \textup{\ref{dembase}},
 \textup{\ref{demzamkn}}. Besides,
let $\{A_1,\ldots,A_n\}$ be the partition of the set
$A\subseteq\nat$ into regular over $Q$ sets $(n\geq 2)$. Then any
matching $f\in Q$ over the set $A$ can be expressed in terms of
compositions of permutations over $Q,$ where each of them is a
 matching over the set of the type $A_i\cup
A_{i+1}$ $(1\leq i\leq n-1)$.
\end{statse}
\begin{proof}
Let one prove this statement by inducting on $n.$ For $n=2$ the
statement is obvious. Let $n\geq 3$ and the statement is proved for
the values $2,\ldots,n-1.$ Let one apply the statement
\ref{statthreetwo} for sets $A_1,$ $A_2\cup\ldots\cup A_{n-1}$ and
$A_n$ (from the statement
 \ref{statunionregular} it follows that these sets are regular). Thus, $f$ can be expressed in terms of the composition of permutations from $Q,$ each of which is a  matching over the set $A_1\cup\ldots\cup A_{n-1}$ or $A_2\cup\ldots\cup A_n.$
 By applying for each of these permutations and corresponding sets the inductive hypothesis, one obtains the prove for the statement.
\end{proof}

\begin{statse}\label{statmegadelete}
Let the  matching  $f$ and the set of four dimensional vectors $B$
with different components \textup{(}inside the vectors and in
different vectors\textup{)} satisfy \textup{(\ref{eq_razb_f})}.
Besides, let the  matchings $f_1'$ and $f_2'$ over the set
$A\subseteq\nat$ be defined by the relations
\begin{equation}\label{eq_mdel_fones}
f_1':\ b_1\leftrightarrow b_2,\ b_3\leftrightarrow b_4,\
(b_1,b_2,b_3,b_4)\in B,
\end{equation}
\begin{equation}\label{eq_mdel_ftwos}
f_2':\ b_1\leftrightarrow b_3,\ (b_1,b_2,b_3,b_4)\in B.
\end{equation}
And let $f_1'',$ $f_2''$ be the  matchings over the sets $A_1''$ and
$A_2''$ respectively, $A,$ $A_1'',$ $A_2''$ do not intersect
pairwise,
$$
f_1=f_1'\circ f_1'',\quad f_2=f_2'\circ f_2''.
$$
Then $f=(f_1\circ f_2)^2.$
\end{statse}
\begin{proof}
One can note that if $\varphi,$ $\psi$ are matchings over the set
$A_{\varphi}$ and $A_{\psi}$ respectively, $A_{\varphi}\cap
A_{\psi}=\varnothing,$ then
$$
\varphi\circ\psi=\psi\circ\varphi.
$$
From this and from the fact that $A,$ $A_1'',$ $A_2''$ do not intersect pairwise it follows that
$$
(f_1\circ f_2)^2=f_1'\circ f_1''\circ f_2'\circ f_2''\circ
f_1'\circ f_1''\circ f_2'\circ f_2''=(f_1'\circ
f_2')^2\circ(f_1'')^2\circ(f_2'')^2=(f_1'\circ f_2')^2.
$$
From (\ref{eq_mdel_fones}) and (\ref{eq_mdel_ftwos}) it follows that
$$
f_1'\circ f_2':\ b_1\rightarrow b_4\rightarrow b_3\rightarrow
b_2\rightarrow b_1,\ (b_1,b_2,b_3,b_4)\in B.
$$
From this it follows that
$$
(f_1'\circ f_2')^2:\ b_1\leftrightarrow b_3,\ b_2\leftrightarrow
b_4,\ (b_1,b_2,b_3,b_4)\in B.
$$
The right part of this formula coincides with the right part
(\ref{eq_razb_f}). Thereby, it is justified to claim that
$$
f=(f_1'\circ f_2')^2.
$$
The claim is proved.
\end{proof}

Let one introduce a few axillary definitions. Let class $Q$ satisfy
the requirements \ref{dembase}--\ref{demrazb}, \ref{demunarbasis}.
Then from the consequence of the theorem \ref{theorem_main} it
follows that there exists a finite number of permutations from
$\gr(Q)$ in terms of compositions of them one can express any
permutation from $\gr(Q)$. From this and from the theorem
\ref{statanyhpair}
 it follows that there exist correct in
$Q$ threesomes
$$(f_1,g_1,B_1'),\ldots,(f_n,g_n,B_n')$$
such that the set that consists of  matchings
\begin{equation}\label{eq_two_pairbasis}
f_1,\ldots,f_n,
\end{equation}
generates $\gr(Q)$.

Let one define vector-function $\delta:\ \nat^4\rightarrow\nat^4$
through the equality
$$
\delta(b_1,b_2,b_3,b_4)=\begin{cases}(b_1,b_2,b_3,b_4),\text{ if
$b_1<b_2$},\\ (b_2,b_1,b_4,b_3)\text{ otherwise}.\end{cases}
$$
For all $i$ ($1\leq i\leq n$) let one assume that
$$
B_i=\{\delta(b_1,b_2,b_3,b_4):\ (b_1,b_2,b_3,b_4)\in B_i'\}.
$$
Let one remark that $(f_1,g_1,B_1),\ldots,(f_n,g_n,B_n)$ are correct in
$Q$ threesomes. Besides,
\begin{equation}\label{eq_two_bi}
B_i=\{(x,g_i(x),f_i(x),f_i(g_i(x))):\ x<g_i(x)\},\quad 1\leq i\leq
n.
\end{equation}

Let one assume that
\begin{equation}\label{eq_two_hi_def}
h_i:\ b_1\leftrightarrow b_2,\ b_3\leftrightarrow b_4,\
(b_1,b_2,b_3,b_4)\in B_i,
\end{equation}
\begin{equation}\label{eq_two_hipn_def}
h_{i+n}:\ b_1\leftrightarrow b_3,\ (b_1,b_2,b_3,b_4)\in B_i
\end{equation}
($1\leq i\leq n$). Considering that  (\ref{eq_two_bi}) this can be rewritten in the following way
\begin{equation}\label{eq_two_hi}
h_i:\ x\leftrightarrow g_i(x),\ f_i(x)\leftrightarrow
f_i(g_i(x)),\ x<g_i(x),
\end{equation}
\begin{equation}\label{eq_two_hipn}
h_{i+n}:\ x\leftrightarrow f_i(x),\ x<g_i(x).
\end{equation}

Let
\begin{equation}\label{eq_two_rol}
\ffrol(x)=x-\ffrm(x,2^{2n+1})+\ffrm(x+1,2^{2n+1}),
\end{equation}
\begin{equation}\label{eq_two_ezero}
E_0=\{2^{2n+1}x,\ x\in\nat\},
\end{equation}
\begin{equation}\label{eq_two_ei}
E_i=\ffrol^i(E_0)\quad (1\leq i<2^{2n+1}).
\end{equation}
Let one remark that $\{E_0,\ldots,E_{2^{2n+1}-1}\}$ is a partition
of the set $\nat.$ Obviously $E_0$ is regular in $Q$
($\mu_{E_0}(x)=[x/2^{2n+1}],$ if $\ffrm(x,2^{2n+1})=0,$
$\mu_{E_0}(x)=0$ otherwise, $\nu_{E_0}(x)=2^{2n+1}x$), analogously
the regularity in $Q$ for all sets $E_i$ ($1\leq i<2^{2n+1}$) can be
proved. From the statement \ref{statunionregular} it follows that
$E_0\cup E_1$ is regular in $Q.$

Let one assume that
\begin{equation}\label{eq_two_ui}
u_i:\ \nu_{E_0\cup E_1}(x)\rightarrow \nu_{E_0\cup E_1}(h_i(x)),
\end{equation}
\begin{equation}\label{eq_two_vi}
v_i=\ffrol^{2^i}\circ u_i\circ \ffrol^{-2^i}
\end{equation}
$(1\leq i\leq 2n),$
\begin{equation}\label{eq_two_all}
\ffall=v_1\circ\ldots\circ v_{2n},
\end{equation}
\begin{equation}\label{eq_two_wi}
w_i:\ \nu_{E_0\cup E_1}(x)\rightarrow \nu_{E_0\cup E_1}(f_i(x))
\end{equation}
$(1\leq i\leq n)$.

\begin{statse}\label{stat_rol_all_q}
$\ffrol,\ffrol^{-1},\ffall=\ffall^{-1}\in Q.$
\end{statse}
\begin{proof}
The statement for $\ffrol$ and $\ffrol^{-1}$ follows directly from
(\ref{eq_two_rol}). From (\ref{eq_two_hi}) and (\ref{eq_two_hipn})
it follows that $h_i\in Q$ ($1\leq i\leq 2n$), from this and from
(\ref{eq_two_ui}), (\ref{eq_two_vi}), (\ref{eq_two_all}) it follows that
$$\ffall\in Q.$$
Besides from (\ref{eq_two_hi}) and (\ref{eq_two_hipn}) it follows that
 $h_i$ is a  matching for any $i$ ($1\leq i\leq 2n$). From this
and from (\ref{eq_two_ui}) it follows that $u_i$ is a  matching over
$E_0\cup E_1$, from (\ref{eq_two_vi}) it follows that $v_i$ is a
 matching over $E_{2^i}\cup E_{2^i+1}$ ($1\leq
i\leq 2n$). Thereby,
 $v_1,\ldots,v_{2n}$ are the matchings at pairwise non-intersecting sets.
 From this and from (\ref{eq_two_all})
it follows that $\ffall$ is a  matching. Thus,
$$
\ffall^{-1}=\ffall\in Q.
$$
\end{proof}

\begin{statse}\label{stat_two_numbers}
Let $n\geq 1,$ $1\leq i\leq n,$ for all $j$ $(1\leq j\leq 2n)$ the
numbers $\alpha_j$ and $\beta_j$ are being defined by the equalities
$$
\alpha_j=\ffrm(2^{2n+1}+2^j-2^i,2^{2n+1}),\quad
\beta_j=\ffrm(2^{2n+1}+2^j-2^{i+n},2^{2n+1}).
$$

Then the numbers $0,$ $1,$ $\alpha_j,$ $\alpha_j+1$ $(j\neq i)$,
$\beta_j,$ $\beta_j+1$ $(j\neq i+n)$ are pairwise distinct.
\end{statse}
\begin{proof}
Considering that all numbers $\alpha_j,$ $\beta_j$ are even, it
sufficies to say that the numbers $0,$ $\alpha_j$ ($j\neq i$),
$\beta_j$ ($j\neq i+n$) are pairwise different. Let one prove this
from contradiction. Let $\alpha_j=0.$ Then $2^j\equiv 2^i\
(\mathrm{mod}\ 2^{2n+1}),$ i.e. $i=j.$ Analogously, if $\beta_j=0,$
then $j=i+n.$ If $\alpha_{j_1}=\alpha_{j_2},$ then $2^{j_1}\equiv
2^{j_2}\ (\mathrm{mod}\ 2^{2n+1}),$ i.e. $j_1=j_2.$ One can proceed
analogously with the case $\beta_{j_1}=\beta_{j_2}.$ If
$\alpha_{j_1}=\beta_{j_2},$ then $2^{j_1}+2^{i+n}\equiv 2^{j_2}+2^i\
(\mathrm{mod}\ 2^{2n+1}),$ i.e. $2^{j_1}+2^{i+n}=2^{j_2}+2^i$
(because the left and the right part are between $4$ and
$2^{2n+1}$). This equality is possible only if $j_1=i$ and
$j_2=i+n.$ The statement is proved.
\end{proof}

\begin{statse}\label{statw}
Permutations $w_i$ $(1\leq i\leq n)$ can be expressed in terms of a composition of permutations
$\ffall$ and $\ffrol.$
\end{statse}
\begin{proof}
Let one fix $i.$ Let
\begin{equation}\label{eq_two_sone}
s_1=\ffrol^{-2^i}\circ\ffall\circ\ffrol^{2^i}=\ffrol^{2^{2n+1}-2^i}\circ\ffall\circ\ffrol^{2^i},
\end{equation}
\begin{equation}\label{eq_two_stwo}
s_2=\ffrol^{-2^{i+n}}\circ\ffall\circ\ffrol^{2^{i+n}}=\ffrol^{2^{2n+1}-2^{i+n}}\circ\ffall\circ\ffrol^{2^{i+n}}.
\end{equation}
One has
$$
s_1=\ffrol^{-2^i}\circ v_1\circ\ldots\circ
v_{2n}\circ\ffrol^{2^i}=(\ffrol^{-2^i}\circ
v_1\circ\ffrol^{2^i})\circ\ldots\circ(\ffrol^{-2^i}\circ
v_{2n}\circ\ffrol^{2^i})=
$$
$$
=p_1'\circ\ldots\circ p_{2n}',
$$
where
\begin{equation}\label{eq_two_pjs}
p_j'=\ffrol^{2^j-2^i}\circ u_j\circ\ffrol^{2^i-2^j}\quad (1\leq
j\leq 2n).
\end{equation}
From (\ref{eq_two_rol}), (\ref{eq_two_ei}) the statement
\ref{statdomainswap} and from the fact that $u_j$ is a  matching at
$E_0\cup E_1,$ it follows that $p_j'$ is a matching at
$E_{\alpha_j}\cup E_{\alpha_j+1},$ where
$$
\alpha_j=\ffrm(2^{2n+1}+2^j-2^i,2^{2n+1}),\quad 1\leq j\leq 2n.
$$
Analogously,
$$
s_2=p_1''\circ\ldots\circ p_{2n}'',
$$
where
\begin{equation}\label{eq_two_pjss}
p_j''=\ffrol^{2^j-2^{i+n}}\circ u_j\circ\ffrol^{2^{i+n}-2^j}\quad
(1\leq j\leq 2n),
\end{equation}
$p_j''$ is a  matching at $E_{\beta_j}\cup E_{\beta_j+1},$ where
$$
\beta_j=\ffrm(2^{2n+1}+2^j-2^{i+n},2^{2n+1}),\quad 1\leq j\leq 2n.
$$

From the statement \ref{stat_two_numbers} it follows that the set
$E_0\cup E_1,$ and all sets $E_{\alpha_j}\cup E_{\alpha_j+1}$
($j\neq i$), $E_{\beta_j}\cup E_{\beta_j+1}$ ($j\neq i+n$) do not intersect pairwise. From this it follows that
$$
p_i'\circ p_j'=p_j'\circ p_i'
$$
for all $j\neq i$ and
$$
p_{i+n}''\circ p_j''=p_j''\circ p_{i+n}''
$$
for all $j\neq i+n.$ From this one can conclude that
\begin{equation}\label{eq_two_sones_ss}
s_1=p_1'\circ\ldots\circ p_{2n}'=s_1'\circ s_1'',
\end{equation}
where
\begin{equation}\label{eq_two_sones}
s_1'=p_i',
\end{equation}
$$
s_1''=p_1'\circ\ldots\circ p_{i-1}'\circ p_{i+1}'\circ\ldots\circ
p_{2n}'.
$$
Analogously,
\begin{equation}\label{eq_two_stwos_ss}
s_2=s_2'\circ s_2'',
\end{equation}
where
\begin{equation}\label{eq_two_stwos}
s_2'=p_{i+n}'',
\end{equation}
$$
s_2''=p_1''\circ\ldots\circ p_{i+n-1}''\circ
p_{i+n+1}''\circ\ldots\circ p_{2n}''.
$$
Besides, from the fact that there is no pairwise intersection for the given sets it follows that
 $s_1'',$ $s_2''$ are  matchings.

From (\ref{eq_two_hi_def}), (\ref{eq_two_hipn_def}),
(\ref{eq_two_ui}) it follows that
\begin{equation}\label{eq_two_ui_b}
u_i:\ \nu_{E_0\cup E_1}(b_1)\leftrightarrow\nu_{E_0\cup
E_1}(b_2),\ \nu_{E_0\cup E_1}(b_3)\leftrightarrow\nu_{E_0\cup
E_1}(b_4),\ (b_1,b_2,b_3,b_4)\in B_i,
\end{equation}
\begin{equation}\label{eq_two_uipn_b}
u_{i+n}:\ \nu_{E_0\cup E_1}(b_1)\leftrightarrow\nu_{E_0\cup
E_1}(b_3),\ (b_1,b_2,b_3,b_4)\in B_i
\end{equation}
($1\leq i\leq n$). Besides, from (\ref{eq_two_wi}) and from the fact that
$(f_i,g_i,B_i)$ is a correct threesome it follows that
\begin{equation}\label{eq_two_wi_b}
w_i:\ \nu_{E_0\cup E_1}(b_1)\leftrightarrow\nu_{E_0\cup
E_1}(b_3),\ \nu_{E_0\cup E_1}(b_2)\leftrightarrow\nu_{E_0\cup
E_1}(b_4),\ (b_1,b_2,b_3,b_4)\in B_i
\end{equation}
($1\leq i\leq n$).

From (\ref{eq_two_pjs}), (\ref{eq_two_sones}) it follows that
\begin{equation}\label{eq_two_suone_eq}
s_1'=u_i,
\end{equation}
from (\ref{eq_two_pjss}), (\ref{eq_two_stwos}) ---
\begin{equation}\label{eq_two_sutwo_eq}
s_2'=u_{i+n}.
\end{equation}
One can note that $s_1',s_2'$ are matchings over $E_0\cup E_1$ (it
follows from
 (\ref{eq_two_ui}), (\ref{eq_two_suone_eq}),
(\ref{eq_two_sutwo_eq})), $s_1''$
--- at $\bigcup_{j\neq i}(E_{\alpha_j}\cup E_{\alpha_j+1}),$
$s_2''$
--- at $\bigcup_{j\neq i+n}(E_{\beta_j}\cup E_{\beta_j+1}),$
the given sets don't intersect. From (\ref{eq_two_sones_ss}),
(\ref{eq_two_stwos_ss}), (\ref{eq_two_ui_b}),
(\ref{eq_two_uipn_b}), (\ref{eq_two_wi_b}),
(\ref{eq_two_suone_eq}), (\ref{eq_two_sutwo_eq}) it follows that for permutations
$w_i,$ $s_1,$ $s_2,$ $s_1',$ $s_2',$ $s_1'',$ $s_2''$
(together with $f,$ $f_1,$ $f_2,$ $f'_1,$ $f'_2,$ $f''_1,$ $f''_2$
respectively) and the set
$$
B''_i=\{(\nu_{E_0\cup E_1}(b_1),\nu_{E_0\cup
E_1}(b_2),\nu_{E_0\cup E_1}(b_3),\nu_{E_0\cup E_1}(b_4)),\
(b_1,b_2,b_3,b_4)\in B_i\}
$$
it satisfies all of the conditions of for the statement \ref{statmegadelete}. From this it follows that
$$
w_i=(s_1\circ s_2)^2.
$$
One can note that  $s_1$ and $s_2$ can be expressed in terms of a
composition $\ffrol$ and $\ffall$ ((\ref{eq_two_sone}),
(\ref{eq_two_stwo})). The claim is proved.
\end{proof}

\begin{statse}\label{stat_two_ezeroeone}
Any  matching over $E_0\cup E_1,$ that belong to $Q,$ can be
expressed in terms of a composition of $\ffrol$ and $\ffall.$
\end{statse}
\begin{proof}
Let $f$ be the given  matching. Let one assume that
$$
g(x)=\mu_{E_0\cup E_1}\circ f\circ \nu_{E_0\cup E_1}(x).
$$
It is obvious that $g$ is a  matching that belongs to $Q$. Thus,
there exist $i_1,\ldots,i_k$ such that $1\leq i_1,\ldots,i_k\leq n$
and
\begin{equation}\label{eq_two_ezeo_g}
g=f_{i_1}\circ\ldots\circ f_{i_k}
\end{equation}
($f_1,\ldots,f_n$ are  matchings from (\ref{eq_two_pairbasis})). It
is easy to notice that there is
$$
f:\ \nu_{E_0\cup E_1}(x)\rightarrow \nu_{E_0\cup E_1}(g(x)).
$$
From this, (\ref{eq_two_wi}) and (\ref{eq_two_ezeo_g}) it follows
that
$$
f=w_{i_1}\circ\ldots\circ w_{i_k}.
$$
From this and from the satement \ref{statw} follows the proof of the
claim.
\end{proof}

\begin{statse}\label{stat_two_eieip}
If $0\leq i<2^{2n+1}-1,$ $f\in Q$ is a matching over $E_i\cup
E_{i+1},$ then $f$ can be expressed as a composition $\ffrol$ and
$\ffall.$
\end{statse}
\begin{proof}
Let
$$
f'=\ffrol^{-i}\circ f\circ\ffrol^i.
$$
From (\ref{eq_two_rol}), (\ref{eq_two_ezero}), (\ref{eq_two_ei}) and
the statement \ref{statdomainswap} it follows that $f'$ is a
matching over
 $E_0\cup E_1.$ Besides, it is obvious that
$$
f=\ffrol^i\circ f'\circ\ffrol^{-i}=\ffrol^i\circ
f'\circ\ffrol^{2^{2n+1}-i}.
$$
From this and the statement \ref{stat_two_ezeroeone} it follows that
the claim is true.
\end{proof}

\begin{theorem1}[additional notes of the theorem \ref{theorem_per_basic}]Any permutation $f\in\gr(Q)$
can be expressed in terms of the composition of $\ffrol$ and $\ffall.$
\end{theorem1}
\begin{proof}
Indeed, from the theorem \ref{statanyhpair} it follows that $f$ can
be expressed in terms of the composition of   matchings in $Q,$
according to the statement \ref{statsequence} every such  matching
can be expressed in terms of compositions of matchings over the sets
of type $E_i\cup E_{i+1}$ ($0\leq i<2^{2n+1}-1$), that belong to
$Q,$ due to the statement
 \ref{stat_two_eieip} each of such  matchings
 can be expressed in terms of a composition $\ffrol,$ $\ffall.$
\end{proof}

\section{Finite Generability of a Group $\gr(Q)$ for Specific Classes $Q$}

\noindent\textbf{Proof of  Theorem \ref{theorem_per_concrete}.} Let
one consider from the start the case of the class $\ccfp$. The
requirements \ref{dembase}, \ref{demnumerate} for $\ccfp$ can be
easily proved. The requirement \ref{demunarbasis} follows from
~\cite{muchnik}. Let one prove that it satisfies the requirement
\ref{demrazb}.

Let $f$ be a permutation, $f,f^{-1}\in\ccfp.$ Let one pick the
function $h(x)$ of the form $2^{[\log_2(x+20)]^n}+2x$ ($n\geq 2$)
such that for any $x$ there is
$$
f(x),f^{-1}(x)<h(x).
$$
Let one note that $h(x)>x$ for any $x,$ and the function $h(x)-x$
increases. Let
$$
A_i=\{x:\ h^{i}(0)\leq x<h^{i+1}(0)\},\quad i\geq 0.
$$
It is clear that $\{A_i\}$ is a partition of the set $\nat.$ Let
\begin{equation}\label{stat_ex_rone}
R_1=A_0\cup A_4\cup A_8\cup\ldots,
\end{equation}
$$
R_2=A_2\cup A_6\cup A_{10}\cup\ldots.
$$
If $x\in A_i,$ $f(x)\in A_j,$ then $x<h^{i+1}(0)$ and, therefore,
 $f(x)<h^{i+2}(0),$ i.e. $j\leq i+1.$ Analogously,
by noting that $x=f^{-1}(f(x)),$ one obtains the result $i\leq j+1,$
i.e. $|i-j|\leq 1.$ From this is follows that
$$
f(R_1)\cap R_2=\varnothing.
$$
Let one prove the regularity of $R_1$ in $\ccfp$ (the regularity $R_2,$
$\nat\backslash R_1,$ $\nat\backslash R_2$ can be proved analogously). Let
$$
\mu_1(x)=\begin{cases}\text{the number of $x$ in the set $R_1$} \\
\quad\quad \text{(numeration with respect to increasing starting from zero),} \\ \quad\quad\text{if $x\in R_1,$} \\
\text{$0$ else,}\end{cases}
$$
$$
\nu_1(x)=\text{the element of the set $R_1$ with the number $x$}.
$$
One can note that $h(x)>2x$ for all $x,$ thus to calculate values $\mu_1(x)$ and $\chi_{R_1}(x)$ it is sufficient to have $[\log_2 x]+1$
 iterations of function $h$. From this it obviously follows that
$\mu_1,\chi_{R_1}\in\ccfp.$ Finally, one needs to prove $\nu_1\in\ccfp$
(for this obviously it is sufficient to prove that it is upper bounded by some function form $\ccfp$).

One can remark that $A_0,A_1,\ldots$ are non-intersecting intervals in
$\nat,$ their lengths $|A_i|$ increase with the increase of $i$ (becasuse
$h(x)-x$ increases). From this and from (\ref{stat_ex_rone}) it follows that for any
$i$ it satisfies
$$
\mu_1(h^{4i+1}(0)-1)+1\geq \frac{h^{4i+1}(0)}{4}.
$$

From this it follows that for some $x$ and $i$ it holds that
$\nu_1(x)=h^{4i+1}(0)-1,$ and, thus, it is true that
$$
\nu_1(x)\leq 4x+3.
$$
Let one prove that for any $x$ there is $\nu_1(x)<h^5(4x+3).$
Indeed for $x<h(0)-1$ it is obvious, for $x\geq h(0)-1$
one chooses the biggest $i$ such that $h^{4i+1}(0)-1\leq \nu_1(x).$
Let $h^{4i+1}(0)-1=\nu_1(x')$ Then one has
$$
\nu_1(x)<h^{4i+5}(0)-1\leq h^5(h^{4i+1}(0)-1)=h^5(\nu_1(x'))\leq
h^5(4x'+3)\leq h^5(4x+3).
$$
From this inequality it follows that $\nu_1(x)\in\ccfp.$ Thus,
$f,$ $R_1$ and $R_2$ satisfy all of those conditions from \ref{demrazb}.

Now let one consider the class $\ccffom$. Let one prove that the
sets built in the same way $R_1$ and $R_2$ can work here as well.
For this it suffices to show that $\mu_1$, $\mu_2$, $\nu_1$,
$\nu_2$, $\chi_{R_1}$, $\chi_{R_2}$ belong to $\ccffom$. The most
difficult part of the process to compute these two functions is the
iteration of the function $h$ (the remaining parts do not have any
problems, see ~\cite{arythm,uniformity}). One can notice that for
any $x$ it holds that
$$
2^{[\log_2(x+20)]^n}>2x+20,
$$
thus
$$
[\log_2 h(h(x))]=[\log_2 h(x)]^n.
$$
From that it follows that
$$
h^k(x) = 2^k x + \sum_{i=1}^k 2^{[\log_2(x+20)]^{n^i}+k-i}.
$$
Based on this representation and the results
~\cite{arythm,uniformity} it is easy to prove that all the necessary
functions to the class $\ccffom$. The requirement \ref{demunarbasis}
is proved in the section \ref{section_ffom} of the chapter
\ref{chapter_arithm}. The rest of those requirements are obvious
(for example as a numerating function one can take a function that
places binary digits of the first number into the even places, for
the second number it places them onto the odd ones; although one can
use the standard polynomial (Peano function) but in this case the
proof that the inverse functions belong to $\ccffom$ will be harder,
see ~\cite{arythm}).

For the class $\ccfl$, minding the fact that $\ccffom\subseteq\ccfl$
(see ~\cite{uniformity}), all requirements but
\ref{demunarbasis}, can be proved in the same fashion. Let one prove the requirement
\ref{demunarbasis}. One can notice that the system of the functions
$$
0,\quad x+1,\quad x+y,\quad xy,\quad 2^{[\log_2 x]^2},\quad
U(n,x,s),
$$
where $U(n,x,s)$ is the result of calculating multitape Turing
machine (with no recording onto the input tape) at the input $x$ (in
binary representation) with the space restriction $[\log_2
s]/(\text{the number of tapes})$ ($U(n,x,s)=0$, if the machine
doesn't stop or there is a mistake in calculations), is the basis in
$\ccfl$ (which can easily be proved using the method from
~\cite{muchnik}).

Now let $Q$ be an $\bigeps^2$-closed class that has a finite basis
with respect to superposition (i.e. automatically satisfying the
requirement \ref{demunarbasis}). Then obviously it contains all
functions from $\bigeps^2$ and, therefore, satisfies the
requirements \ref{dembase} and \ref{demnumerate} (see ~\cite{erf}).
Let one prove that it satisfies the requirement  \ref{demrazb}.
Indeed, let $f(x)\in\gr(Q)$. Then let one assume
$$
h(x)=\max_{0\leq y\leq x}\max(f(y),f^{-1}(y))+2x+1,
$$
let one define the sets $R_1$ and $R_2$ analogously to how one did it for the class $\ccfp$ (using just the given function $h$). Based on the technique from ~\cite{erf} it is easy to show that $R_1$ and $R_2$
satisfy the requirements \ref{demrazb} for the class $Q$.
The theorem is proved.

\end{document}